\documentclass[a4paper,11pt]{article}

\usepackage{epsf,epsfig,psfrag}
\usepackage{graphicx}
\usepackage{caption}
\usepackage{subcaption}
\usepackage[usenames,dvipsnames]{xcolor}
\definecolor{darkgreen}{RGB}{34,139,34}

\usepackage{hyperref} 

\usepackage{authblk}

\usepackage{enumerate}

\usepackage{url}
\usepackage{amsmath,amsfonts,amssymb,dsfont}
\usepackage{amsthm}

\usepackage{cite}

\usepackage{breqn}

\usepackage{a4wide}
\numberwithin{equation}{section}

\newtheorem{lemma}{Lemma}
\newtheorem{pr}{Principle}

\newtheorem*{cons}{Consequence}
\newtheorem*{claim}{Claim}

\newtheorem*{holid*}{Holographic principle}

\makeatletter \@addtoreset{equation}{section} \makeatother

\DeclareMathOperator{\tr}{tr}

\DeclareMathOperator{\Tr}{Tr}

\def\be{\begin{equation}}
\def\ee{\end{equation}}
\def\ba{\begin{eqnarray}}
\def\ea{\end{eqnarray}}

\newcommand\nn{\nonumber}
\newcommand\q{\quad}

\newcommand{\cq}{\mathcal Q}

\newcommand{\cs}{\mathcal S}
\newcommand{\ct}{\mathcal T}

\def\nn{\nonumber}

\newcommand{\eqa}{\begin{eqnarray}}
\newcommand{\neqa}{\end{eqnarray}}

\newcommand{\p}{\partial}

\def\f{\frac}

\usepackage{bbm}

\newcommand{\lalg}[1]{\mathfrak{#1}}  
\newcommand{\SU}{\mathrm{SU}}
\newcommand{\SO}{\mathrm{SO}}
\newcommand{\PSU}{\mathrm{PSU}}
\newcommand{\PSO}{\mathrm{PSO}}

\newcommand{\su}{\lalg{su}}

\newcommand{\so}{\lalg{so}}

\def\q{{\quad}}

\newcommand{\rbx}[1]{\raisebox{6ex}[0pt]{#1}} 

\begin{document}
\setlength{\unitlength}{1mm}

\allowdisplaybreaks
\thispagestyle{empty}


\title{Quantum theory from questions}

\author[1]{Philipp Andres H\"ohn\thanks{\texttt{hoephil@gmail.com}}\thanks{Present address: Vienna Center for Quantum Science and Technology, and Institute for Quantum Optics and Quantum Information, Austrian Academy of Sciences, Boltzmanngasse 3, 1090 Vienna, Austria}}
\author[2,3]{Christopher S.\ P.\ Wever\thanks{\texttt{christopher.wever@kit.edu}}}
\affil[1]{\small Perimeter Institute for Theoretical Physics, 31 Caroline Street North, Waterloo, ON N2L 2Y5, Canada}
\affil[2]{\small Institute of Nuclear and Particle Physics, NCSR ``Demokritos", Agia Paraskevi, 15310 Athens, Greece}
\affil[3]{\small Institute for Theoretical Particle Physics (TTP), Karlsruhe Institute of Technology, Engesserstra{\ss}e 7, D-76128 Karlsruhe \& Institute for Nuclear Physics (IKP), Karlsruhe Institute of Technology, Hermann-von-Helmholtz-Platz 1, D-76344 Eggenstein-Leopoldshafen, Germany}


\date{{\small TTP15-039}}

\setcounter{footnote}{0}
\maketitle
\setcounter{page}{1}


\begin{abstract}
We reconstruct the explicit formalism of qubit quantum theory from elementary
rules on an observer's information acquisition. Our approach is purely operational:
we consider an observer $O$ interrogating a system $S$ with binary questions and define
$S$'s state as $O$'s `catalogue of knowledge' about $S$. From the rules we derive the state spaces for $N$ elementary systems and show that (a) they coincide with the set of density matrices over an $N$-qubit Hilbert space $\mathbb{C}^{2^N}$; (b) states evolve unitarily under the group $\rm{PSU}(2^N)$ according to the von Neumann evolution equation; and (c) $O$'s binary questions correspond to projective Pauli operator measurements with outcome probabilities given by the Born rule. As a by-product, this results in a propositional formulation of quantum theory. Aside from offering an informational explanation for the theory's architecture, the reconstruction also unravels new structural insights. We show that, in a derived quadratic information measure, (d) qubits satisfy inequalities which bound the information content in any set of mutually complementary questions to $1$ \texttt{bit}; and (e) maximal sets of mutually complementary questions for one and two qubits must carry precisely $1$ \texttt{bit} of information in pure states. The latter relations constitute conserved informational charges which define the unitary groups and, together with their conservation conditions, the sets of pure quantum states. These results highlight information as a `charge of quantum theory' and the benefits of this informational approach. This work emphasizes the sufficiency of restricting to an observer's information to reconstruct the theory and completes the
quantum reconstruction initiated in \cite{Hoehn:2014uua}.
\end{abstract}


\section{Introduction}

Quantum theory has enjoyed an outstanding success, allowing us to make precise predictions about the physical microcosm, leading to new information technologies and withstanding every experimental test to which it has been exposed thus far. Yet, in contrast to special and general relativity, quantum theory has evaded a commonly accepted apprehension and interpretation of its physical content, in part as a consequence of a lack of physical statements that fully characterize it. But what makes quantum theory special? Quantum theory has perhaps become so successful that questioning its foundations and physical content have become peripheral matters in physics. However, with the ambition of developing more fundamental theories, involving or going beyond quantum theory, the question as regards its physical meaning and characterizing features returns. How could the world be different if we dropped some of the latter? The answer requires a better understanding of quantum theory within a larger landscape of alternative theories. Furthermore, a convincing conceptual scheme for a putative quantum theory of gravity presumably requires a deeper understanding of what quantum theory tells us about Nature -- and of what we can say about it.

To be sure, within the simplified context of finite dimensional Hilbert spaces, there have been considerable efforts to identify physical attributes special to quantum theory to remedy the flaw that quantum theory is still defined by operationally obscure textbook axioms rather than transparent physical statements. Among them are violation of the Bell \cite{bell1964einstein,bell1966problem} and more generally Clauser-Horne-Shimony-Holt inequalities \cite{clauser1969proposed}, the `no-signaling' principle \cite{popescu1994quantum} and, its generalization, `information causality' \cite{pawlowski2009information}, interference effects in mixtures \cite{Fivel:1994nj}, absence of third and higher order interference \cite{Sorkin:1994dt,Barnum:2014fk}, a limit on the information content carried by systems \cite{Rovelli:1995fv,zeilinger1999foundational,Brukner:ys,Brukner:1999qf,Brukner:2002kx,brukner2009information,Brukner:vn,spekkens2007evidence,Spekkens:2014fk,Paterek:2010fk,weizsaeckerbook,gornitz2003introduction,lyre1995quantum} and others. However, all of these attributes yield incomplete characterizations, being shared by other probabilistic theories some of which admit unphysical correlation structures, as well as exotic information communication and processing tasks. 

In fact, there actually exist a number of successful reconstructions of finite dimensional quantum theory from operational axioms, most of which have been performed within the frameworks of generalized or operational probability theories \cite{Hardy:2001jk,Dakic:2009bh,masanes2011derivation,Mueller:2012ai,Masanes:2012uq,chiribella2011informational,de2012deriving,Mueller:2012pc,Hardy:2013fk,Barnum:2014fk,Garner:2014uq} (see also \cite{Oeckl:2014uq} which is adapted to a space-time language and the more mathematical reconstructions \cite{kochen2013reconstruction,2008arXiv0805.2770G}). Despite the beauty and great technical achievements of some of these reconstructions, they arguably come short of providing a fully satisfactory physical and conceptual picture of quantum theory. Firstly, the underlying axioms, while mathematically crisp, are operationally and intuitively less transparent than a statement of the type ``all inertial observers agree on the speed of light" underlying special relativity. However, for clarity it would be desirable to have easily understandable, yet powerful postulates. Secondly, the ensuing derivations of the quantum formalism are rather implicit than constructive, lacking, in particular, simple and intuitively comprehensible explanations for typical quantum phenomena such as entanglement or for the origin of the explicit structure of the formalism. By contrast, in special relativity, most of its characteristic traits, such as relativity of simultaneity, Lorentz contraction, etc., can be explained in simple thought experiments invoking essentially only the constancy of the speed of light. Thirdly, apart from showing that an operational perspective is sufficient for deriving quantum theory, these reconstructions are interpretationally fairly neutral, focusing on characterizing the formalism rather than the physical and conceptual content of the theory.

The goal of this manuscript is to improve the situation; we shall show that, at least in the simple context of qubit systems, one can understand the physical content of quantum theory from an informational perspective. To this end, we shall exhibit, using the novel framework developed in \cite{Hoehn:2014uua}, how the quantum formalism can be constructively and explicitly derived from simple operationally comprehensible rules which restrict an observer's acquisition of information about systems he is observing. The acquisition of information of the observer about the systems will proceed by interrogation with binary questions. This reconstruction yields the detailed (and not only general) structure of qubit quantum theory and is thereby much less abstract than previous reconstructions. However, it also involves many more steps.

In contrast to earlier works which aim at intrinsic properties and states of systems, here we shall solely focus on the relation of the observer with the systems, i.e., ultimately on the information which the observer has experimentally access too. In particular, we take the quantum state to represent the observer's `catalogue of knowledge' about the observed system(s), rather than an intrinsic state of the latter. This is conceptually motivated by and in line with the relational interpretation of quantum mechanics \cite{Rovelli:1995fv,Smerlak:2006gi,Hoehn:2014uua,hartle1968quantum,zheng1996quantum}, the informational interpretation in \cite{zeilinger1999foundational,Brukner:ys,Brukner:vn,Brukner:2002kx} and (at least elements of) QBism \cite{Fuchs:fk,Caves:2002uq,caves2002unknown}. While this general philosophy goes back to Rovelli's seminal {\it relational quantum mechanics} \cite{Rovelli:1995fv}, none of these earlier works provide a concrete framework from which to reconstruct the theory. This is a shortcoming which has been overcome in \cite{Hoehn:2014uua} and which will be exploited in the sequel. As such, the present manuscript (together with \cite{Hoehn:2014uua}) can be viewed as a completion -- in the context of qubit systems -- of many of the ideas put forward in these earlier works and, in particular, of  relational quantum theory \cite{Rovelli:1995fv}.

Denoting the observer by $O$ and the system by $S$, the rules on information acquisition from which we derive the quantum theory of $N$ qubits can be schematically summarized as follows:
\begin{enumerate}
\item $O$ can maximally acquire $N$ independent \texttt{bits} of information about $S$ at any time.
\item $O$ can always get up to $N$ new independent \texttt{bits} of information about $S$.
\item $O$'s total amount of information about $S$ is preserved in between interrogations.
\item $O$'s `catalogue of knowledge' about $S$ evolves continuously and every consistent such evolution is physically realizable.
\item $O$ can ask $S$ any binary question that ``makes sense''.
\end{enumerate}
In fact, these five rules cannot distinguish two-level systems over complex and real Hilbert spaces. Since the latter is both mathematically and physically a subcase of the former, these five rules are sufficient. However, if one also wishes to distinguish these two cases operationally, then an additional rule, imposed solely for this purpose, will do the job:
\begin{enumerate}
 \setcounter{enumi}{5}
 \item $O$ can determine his `catalogue of knowledge' of a composite system $S$ by interrogating only its constituents. 
\end{enumerate}

While here we shall focus less on conceptual matters than in \cite{Hoehn:2014uua}, the successful reconstruction from this perspective underscores the sufficiency of taking a purely operational perspective, addressing only what an observer can say about the observed systems, in order to understand and derive the formalism of quantum theory. Ontic statements about a reality underlying the observer's interactions with the physical systems are unnecessary. This lends weight to Bohr's famous quote: ``It is wrong to think that the task of physics is to find out how Nature {\it is}. Physics concerns what we can {\it say} about Nature" \cite{petersen}.

Even apart from the fact that this reconstruction offers a novel perspective on the physical content of quantum theory, it leads to new practically useful results. The tools of \cite{Hoehn:2014uua}, while not geared for doing concrete physics with them, are simple and especially devised to expose the structure of qubit quantum theory. They not only admit intuitive graphical representations of the ensuing logical and informational structure of the theory. But they also permit to unravel novel structural insights into qubit quantum theory that have gone unnoticed in the literature. In particular, we shall show how finitely many {\it conserved informational charges}, resulting from complementarity relations, elucidate the origin of the unitary (time evolution) group and characterize pure state spaces.
As clarified along the way, these observations emphasize information as a `charge of quantum theory'; the observer's information provides the conserved quantities of the unitary group which can be transferred among his questions in between measurements.

Certainly, there are also some shortcomings of our approach. Firstly, at present the language of \cite{Hoehn:2014uua} is only applicable to qubit systems, although a suitable generalization appears feasible. Secondly, while our background assumptions are operationally and conceptually transparent, they may be mathematically stronger than those underlying, e.g., \cite{Hardy:2001jk,Dakic:2009bh,masanes2011derivation,Masanes:2012uq,Mueller:2012pc}, thus, in a strict sense, admitting a mathematically weaker reconstruction from within a smaller landscape of theories. Nevertheless, the derivation is a non-trivial proof of principle of the approach and can presumably be strengthened since the set of assumptions and postulates may be non-optimal in the sense of containing partially redundant information. Thirdly, the explicit framework restricts to projective measurements on a subset of qubit observables (although, once one has reconstructed the quantum formalism, one ultimately has access to all quantum operations and POVMs).

The content of the paper is organized as follows. In section \ref{sec_post} we give a review of the framework developed in \cite{Hoehn:2014uua}, which provides the context for our reconstruction of qubit quantum theory. All relevant assumptions and postulates for the reconstruction are summarized in order to make the paper self-contained and we refer to \cite{Hoehn:2014uua} for a more detailed account. In section \ref{sec_n2} we reconstruct the correct unitary time evolution group and state space of quantum theory for $N=2$ qubits and in section \ref{sec_n>2} we extend their reconstruction to $N>2$ qubits. The derivation of the set of binary questions which we permit an observer to ask a system of $N$ qubits is performed in section \ref{sec_qn}. This also involves a derivation of the Born rule for projective measurements, the details of which are spelled out in appendix \ref{app_Born}. In section \ref{sec_Neum} we briefly discuss how the von Neumann evolution equation for the density matrix arises from our reconstruction and finally we present our conclusions in section \ref{sec_con}. The appendices \ref{app_A} and \ref{app_B} contain the detailed derivations of statements made in sections \ref{sec_n2}-\ref{sec_n>2} and \ref{sec_qn} respectively.

\tableofcontents

\section{Background assumptions and postulates}\label{sec_post}

The focus of this approach lies on the acquisition of information of an observer $O$ about a set of systems and the relation this establishes between $O$ and these systems. We shall follow the premise that we may only speak about the information which $O$ has access to through interaction with the systems. This approach is thus purely operational, focusing on what an observer can say about a system rather than on the latter's intrinsic properties and states. This general philosophy has been inspired by Rovelli's {\it relational quantum mechanics} \cite{Rovelli:1995fv} and by the Brukner-Zeilinger informational interpretation of quantum theory \cite{zeilinger1999foundational,Brukner:ys,Brukner:vn,Brukner:2002kx}, neither of which, however, offer a concrete framework for a reconstruction of quantum theory. The lack of a suitable mathematical framework for this endeavour has been overcome in \cite{Hoehn:2014uua} and this is what will be exploited in the remainder.

We shall begin by reviewing the landscape of theories within which the postulates for qubit quantum theory are formulated. This landscape is established by a set of operational background assumptions to which we expose $O$ and the systems. The quantum postulates will constitute rules on $O$'s acquisition of information about the systems. We refer to Ref.\ \cite{Hoehn:2014uua} for further details and more thorough explanations of the concepts employed below.

\subsection{Basic setup: questions and answers}

As in figure \ref{fig_inter}, the setup consists of a preparation device spitting out an ensemble of (identical) systems $S_a$, $a=1,\ldots,n$, which then are interrogated by $O$ with {\it binary} questions. Every way of preparing the systems is assumed to yield a specific statistics over the answers to the binary questions which $O$ may ask the $S_a$ (for a sufficiently large ensemble).
\begin{figure}[htb!]
\begin{center}
\psfrag{p}{Preparation}
\psfrag{i}{Interrogation}
\psfrag{S}{\hspace*{-.1cm}$S$\hspace{.1cm}}
\psfrag{O}{$O$}
\psfrag{q}{$Q_i$?}
{\includegraphics[scale=.4]{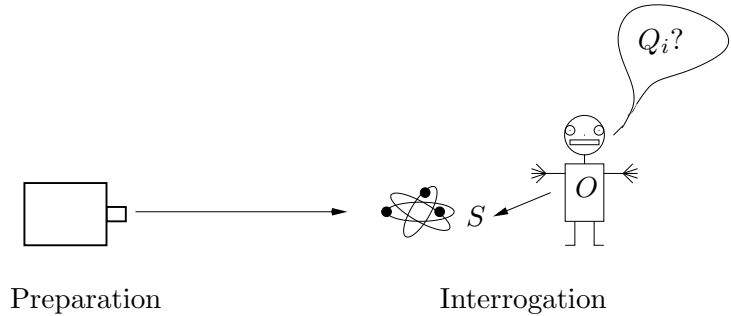}}
  \caption{\small An observer $O$ interrogating a system $S$. }\label{fig_inter}
\end{center}
\end{figure}
More precisely, we shall employ two basic ingredients:
\begin{itemize}
\item[$\cq$] denotes the set of those binary questions $Q_i$ which in this approach we permit $O$ to ask a system $S$. We shall subject $\cq$ to a number of restrictions such that $\cq$ will ultimately be a strict subset of all possible binary questions which $O$ could, in principle, ask $S$. For instance, whenever $O$ asks $S$ any $Q_i\in\cq$ he shall always get an answer\footnote{In this work, we tacitly assume the probability for $S$ being present to be $1$.} and any $Q_i\in\cq$ shall be non-trivial such that $S$'s answer to it is not independent of its preparation. Furthermore, any $Q_i\in\cq$ shall be {\it repeatable}, i.e.\ if $O$ asks the same $Q_i$ $m$ times in immediate succession on the same $S$ he shall receive $m$ times the same answer. 
\item[$\Sigma$] denotes the set of all possible answer statistics for all $Q_i\in\cq$ for all possible ways of preparing the $S_a$.
\end{itemize}

In this work, we therefore do {\it not} address the measurement problem: we simply assume a division between the system $S$ and a `classical' observer $O$ and shall neither explain the origin and nature of this `classical' $O$, nor why $S$ gives definite answers (i.e., yields definite measurement outcomes) upon being asked some $Q_i\in\cq$ by $O$. This will nevertheless allow us to derive the quantum formalism for qubits relative to this $O$.

Just like any experimenter in a real laboratory, we assume $O$ to have developed a theoretical model by means of which he interprets the outcomes to his interrogations (and which, up to his experimental accuracy, is consistent with his observations). In particular, we shall assume $O$ to have a model for both $\cq$ and $\Sigma$ and thereby to be able to decide whether a given question is contained in (his model for) $\cq$ or not. In this work it is not our ambition to clarify how $O$ has arrived at this model. Instead, it will be our task to determine what this model is -- subject to the background assumptions and postulates below.

\subsection{Probabilities and notions of independence and compatibility}\label{sec_2b}
 
For any specific system $S$ to be interrogated next, $O$ assigns a probability $y_i$ that the answer to any $Q_i\in\cq$ will be `yes' in a broadly Bayesian manner as a `degree of belief'. $O$ will estimate $y_i$ according to his model of $\Sigma$ and to any {\it prior} information about $S$, which consists of frequencies of `yes' and `no' answers recorded in a previous interrogation of an ensemble of systems prepared identically to $S$. In the sequel, $O$ is only permitted to acquire information about the systems through the binary questions in $\cq$. Hence, the $y_i$ encode $O$'s entire information about a system $S$. We thus identify {\it the state of $S$ relative to $O$} as $O$'s `catalogue of knowledge' about $S$, namely as the collection of $\{y_i\}_{\forall\,Q_i\in\cq}$. It thus is a state of information associated to the relation of $O$ with $S$ and not an intrinsic state of $S$. The state is an element of $\Sigma$ which therefore constitutes the {\it state space} of $S$ and any state in $\Sigma$ assigns a probability $y_i$ for all $Q_i\in\cq$.\footnote{For consistency, we tacitly assume that the set $\Sigma$ of all possible answer statistics coincides with the set of all possible `beliefs' \cite{Hoehn:2014uua}.}

For operational reasons $\Sigma$ is assumed to be convex and closed. This will permit $O$ to build convex combinations of states; thereby $O$ is able to assign a single prior state to a collection of identical systems (i.e., systems with identical $(\cq,\Sigma)$, but not necessarily in the same state) when he uses a (possibly biased) coin toss cascade to decide which of the systems to interrogate. Since perfect and arbitrarily good preparation are operationally indistinguishable, it is legitimate to assume closure of $\Sigma$ (see \cite{Hoehn:2014uua} for more details). 
  
We require that a special method of preparation exists which produces entirely random question outcomes. More precisely, we assume that there exists a special state in $\Sigma$, defined by $y_i=\f{1}{2}$, $\forall\,Q_i\in\cq$ and referred to as the {\it state of no information}. Note that the existence of this state is a restriction on the pair $(\cq,\Sigma)$.\footnote{Clearly, not all pairs $(\cq,\Sigma)$ will satisfy this. E.g., $(\{\text{binary POVMs}\},\{\text{density matrices}\})$ would not satisfy this restriction since there does not exist a quantum state which yields probability $1/2$ for {\it all} binary POVMs. Namely, there exist binary POVMs with an inherent bias, such as $(E_1=2/3\cdot\mathds{1},E_2=1/3\cdot\mathds{1})$.} This state of no information serves two purposes: (1) it is the prior state $O$ will start with in a state updating once he has `no prior information' about a system other than what the corresponding model $(\cq,\Sigma)$ is (incl.\ also the set of possible time evolutions $\ct$, see below); and (2) it permits us to define a notion of {\it independence} of questions. Indeed, the notion of independence of questions is state dependent\footnote{E.g., in quantum theory, the questions $Q_{x_1}=$``Is the spin of qubit 1 up in $x$-direction?'' and $Q_{x_2}=$``Is the spin of qubit 2 up in $x$-direction?" are independent relative to the completely mixed state, however, not relative to an entangled state (with correlation in $x$-direction).} such that we need a distinguished state relative to which we can unambiguously define it.

More precisely, consider $Q_i\in\cq$ and assume $O$ receives $S$ in the state of no information. On account of repeatability, upon asking $S$ the question $Q_i$ and receiving the answer `yes' or `no', $O$ will assign a probability of $y_i=1$ or $y_i=0$, respectively, that he will receive a `yes' answer from $S$ if asking $Q_i$ again. This is part of a {\it state update rule} which permits $O$ to update his information about the specific $S$ which he is interrogating according to the answers he receives. Clearly, the probability $y_j$ for any other $Q_j\in\cq$ will depend on this update rule. We shall not specify the update rule much further, but just assume that there is a consistent one. Given such an update rule, we shall call $Q_i,Q_j\in\cq$
\begin{description}
\item[independent] if, after having asked $Q_i$ to $S$ in the state of no information, the probability $y_j=\f{1}{2}$. That is, if the answer to $Q_i$ relative to the state of no information tells $O$ `nothing' about the answer to $Q_j$.
\item[dependent] if, after having asked $Q_i$ to $S$ in the state of no information, the probability $y_j=0,1$. That is, if the answer to $Q_i$ relative to the state of no information implies also the answer to $Q_j$.  
\item[partially dependent] if, after having asked $Q_i$ to $S$ in the state of no information, the probability $y_j\neq0,\f{1}{2},1$. That is, if the answer to $Q_i$ relative to the state of no information gives $O$ partial information about the answer to $Q_j$. 
\end{description}
We shall require that these (in-)dependence relations be symmetric such that, e.g.\ $Q_i$ is independent of $Q_j$ iff $Q_j$ is independent of $Q_i$,\footnote{This means that $Q_i,Q_j$ are stochastically independent with respect to the state of no information, i.e.\ the joint probabilities factorize relative to the latter, $p(Q_i,Q_j)=y_i\cdot y_j=\f{1}{2}\cdot\f{1}{2}=\f{1}{4}$, where $p(Q_i,Q_j)=p(Q_j,Q_i)$ denotes the probability that $Q_i$ and $Q_j$ give `yes' answers if asked in {\it sequence} on the same $S$.} etc. We emphasize that these notions of (in-)dependence are a priori update rule dependent.

We also need a notion of compatibility and complementarity; $Q_i,Q_j\in\cq$ are called
\begin{description}
\item[(maximally) compatible] if $O$ may know the answers to both $Q_i,Q_j$ simultaneously, i.e.\ if there exists a state in $\Sigma$ such that $y_i,y_j$ can be simultaneously $0$ or $1$.

\item[(maximally) complementary] if every state in $\Sigma$ which features $y_i=0,1$ necessarily implies $y_j=\f{1}{2}$ (and vice versa). 

\end{description}

One can also define notions of partial compatibility \cite{Hoehn:2014uua}.

This brings us to our last constraint on the update rule: if $Q_i,Q_j$ are maximally compatible and independent then asking $Q_i$ shall not change $y_j$, and vice versa -- regardless of $S$'s state. That is, by asking a question $Q$, $O$ shall not gain or lose information about questions which are compatible with but independent of $Q$.

For clarification, we emphasize that the assumption is: $\cq$ is sufficient to describe the properties of $S$ and, in particular, any of its states. From the perspective of information acquisition it is also natural to assume that there exists a state corresponding to $O$ having `no information' about the measurement outcomes of those properties that he uses to characterize $S$; $\cq$ contains questions that are ``natural'' in this sense. The assumption, however, is {\it not} that $\cq$ encodes a complete description of {\it all} the binary measurements $O$ can physically carry out on $S$. We say nothing about whether $O$ cannot, in principle, also physically perform other measurements. While these would not be contained in $\cq$, they would also not be necessary to do tomography on $S$ and thus to describe its state. For our purposes, it is therefore sufficient to restrict $O$ to the ``natural'' set $\cq$. Ultimately, upon imposing the quantum principles, this will result in projective binary measurements as the ``natural'' questions, while non-projective POVMs would not be encompassed by $\cq$.

\subsection{Informational completeness}

We shall call a set of pairwise independent questions $\cq_{M}:=\{Q_1,\ldots,Q_D\}$ {\it maximal} if no further question from $\cq\setminus\cq_M$ can be added to it that is also pairwise independent of all  other members of $\cq_M$ too. Pairwise independent questions shall constitute the fundamental building blocks of the theories we consider. As such, we shall assume that any maximal set of pairwise independent questions $\cq_M$ also constitutes an {\it informationally complete set of questions} in the sense that, for $S$ in {\it any} state from $\Sigma$, the probabilities $\{y_i\}_{i=1}^D$ are sufficient to compute all $y_j$ $\forall\,Q_j\in\cq$. This is a non-trivial restriction; if it was not satisfied, $O$ would require additional questions that are at least partially dependent on some elements in $\cq_M$ in order to parametrize $\Sigma$. 
This would complicate the discussion and conflict with our premise that pairwise independent questions form the fundamental building blocks of system descriptions. We thus simply preclude this complication by making the assumption of informational completeness\footnote{We do not rule out the possibility that this property of informational completeness of maximal sets could also be proven using the principles below. In fact, for certain subcases the authors were able to prove it. However, the general case remains open.} which, however, still leaves open a large landscape of theories compatible with it. Then one can show that any such $\cq_M$ contains the same number $D$ of elements \cite{Hoehn:2014uua}. In consequence, we may represent $\Sigma$ as a $D$-dimensional convex set, with states as vectors
\ba
\vec{y}=\left(\begin{array}{c}y_1 \\y_2 \\\vdots \\y_D\end{array}\right)\nn.
\ea

Any convex set is defined by its extremal points \cite{Krein1940}. The extremal states in $\Sigma$ are special because they cannot be written as convex mixtures of other states, but all other states are convex mixtures of these. Since (finite) convex mixtures can be operationally understood in terms of (cascades of biased) coin flips, $O$ may prepare non-extremal states by applying cascades of coin flips to ensembles of extremal states. 
But, since the extremal states themselves cannot be prepared via coin flip cascades from other states, their preparation must have an unambiguous operational meaning. For this purpose, we wish any extremal state to be achievable by $O$ as the posterior state of an individual system in an interrogation. 
More specifically, we shall require that $O$ can prepare {\it any} extremal state from the state of no information in a {\it single shot interrogation}\footnote{In a single shot interrogation a single system $S$ is prepared in some state and subsequently exposed to questions (see \cite{Hoehn:2014uua} for more details).} by only asking questions from an informationally complete set $\cq_{M}$ and possibly letting the resulting state evolve in time.

It will become crucial to appropriately quantify $O$'s information about any system. To this end, we quantify $O$'s information about the outcome to any $Q_i\in\cq$ implicitly by a function $\alpha(y_i)$ with $0\leq\alpha(y_i)\leq 1$ \texttt{bit} and $\alpha(y)=0$ $\Leftrightarrow$ $y=\f{1}{2}$ and $\alpha(0)=\alpha(1)=1$ \texttt{bit}. $O$'s total information about $S$ must be a function of the state $\vec{y}$; we define it to be
\ba
I(\vec{y}):=\sum_{i=1}^D\,\alpha(y_i)\nn.
\ea
The specific form of $\alpha$ is derived from the principles.

\subsection{Complementarity properties}

For practical purposes, we shall also sharpen the notion of complementarity 
of questions. Firstly, we shall permit $O$ to use classical rules of inference (in the form of Boolean logic) exclusively on sets of {\it mutually compatible} questions. Classical rules of inference assume propositions to have truth values simultaneously which, in $O$'s description of the world, is only true for mutually compatible questions because any truth value must be operational. This is to prevent him from making statements about logical connectives of complementary questions whose truthfulness he could never test by interrogations 
(see \cite{Hoehn:2014uua}). 

Secondly, we shall require that any set of $n\in\mathbb{N}$ mutually (maximally) complementary questions $\{Q_1,\ldots,Q_n\}$ (not to be confused with $\cq_M$ above) cannot support more than $1$ \texttt{bit} of information:
\ba
\alpha(y_1)+\cdots+\alpha(y_n)\leq 1 \,\,\texttt{bit}.\label{compstrong}
\ea
This statement follows trivially from the definition of complementarity and the basic requirements on $\alpha$ whenever $O$ has maximal information $\alpha(y_i)=1$ \texttt{bit} about any question in this set. However, we require (\ref{compstrong}) for {\it all} states for otherwise it would be possible for $O$ to {\it reduce} his total information about this set by asking another question from it. Namely, suppose (\ref{compstrong}) was violated. Then, upon asking {\it any} question from this set, he will have maximal information about this question and none about the others such that (\ref{compstrong}) would be saturated again and $O$ has experienced a net loss of information about the set. Such peculiar situations will be ruled out in $O$'s world. We shall call the informational relations defined by (\ref{compstrong}) {\it complementarity inequalities}. They can be viewed as informational uncertainty relations, describing how the information gain about one question enforces an information loss about questions complementary to it.

\subsection{Composite systems}

Since we will be dealing with systems composed of $N$ qubits below, we must clarify what kind of composite systems we shall consider in this language. Let $\cq_{A,B}$ be the question sets associated to systems $S_{A,B}$. We shall say that they form the composite system $S_{AB}$ if all questions in $\cq_A$ are maximally compatible with and independent of all questions in $\cq_B$ and if
\ba
\cq_{AB}=\cq_A\cup\cq_B\cup\tilde{\cq}_{AB},
\label{composite}
\ea
where $\tilde{\cq}_{AB}$ contains only composite questions that are iterative compositions, $Q_a\,*_{\tiny1}\,Q_b, Q_a\,*_2(Q_{a'}*_3Q_b), (Q_a*_4Q_b)*_5Q_{b'}, (Q_a*_6Q_b)*_7(Q_{a'}*_8Q_{b'}),\ldots$, via some logical connectives $*_1,*_2,*_3,\cdots$, of questions $Q_a,Q_{a'},\ldots\in\cq_A$ and $Q_b,Q_{b'},\ldots\in\cq_B$. 
Given the assumption about rules of inference above, note that $O$ can only logically connect two (possibly composite) questions {\it directly} with some $*$ if they are compatible \cite{Hoehn:2014uua}. The logical connective $*$ which can be used to build informationally complete sets for composite systems will be determined later. We use this definition of composite systems recursively for more than two systems. 

\subsection{Time evolution}

The state that $O$ assigns to the system $S$ is allowed to evolve in time. We shall assume temporal translation invariance, such that any time evolution defines a map $T_{\Delta t}(\vec{y}(t_0))=\vec{y}(t_0+\Delta t)$ from $\Sigma$ to itself which only depends on the time interval $\Delta t$ and not, however, on the instant of time itself. The set of all possible time evolutions $T$ will be denoted by $\ct$ and constitutes another ingredient of $O$'s model for describing $S$. $O$'s theoretical model for $S$ is thus encoded in the triple $(\cq,\Sigma,\ct)$.

\subsection{How to compute probabilities}

The assumption of informational completeness asserts that the probabilities for `yes'-outcomes to the questions in an informationally complete set are sufficient to compute the outcome probabilities for {\it all} questions in the set $\cq$
 in {\it any} given state. Hence, by assumption, the probability function $Y(Q|\vec{y})$ 
 that $Q=$
  `yes', given the state $\vec{y}$,
   exists and is meaningful for all $Q\in\cq$. But how do we compute it?
   
   Suppose $O$ has access to two identical (but not necessarily identically prepared) systems\footnote{By identical systems we mean systems featuring the same triple $(\cq,\Sigma,\ct)$ (see \cite{Hoehn:2014uua} for further details on identical systems and a definition of `identically prepared').} $S_1,S_2$ such that $O$ may ask the same questions to both. $O$ may perform a biased coin flip which yields `heads' with probability $\lambda\in[0,1]$, in which case he will interrogate $S_1$, and `tails' with probability $(1-\lambda)$ in which case he will interrogate $S_2$. This implies that the state of the combined system (before tossing the coin) reads $\vec{y}_{12}=\lambda\,\vec{y}_1+(1-\lambda)\,\vec{y}_2$ since this holds for each component $y_i$ (see also \cite{Hoehn:2014uua}). We recall from section \ref{sec_2b} that $O$ determines the probabilities by recording the frequencies of question outcomes in repeated interrogations of identically prepared systems. Hence, $Y(Q|\vec{y})$ is determined from the frequency of `yes'-outcomes of $Q$ when asked to a very large (ideally infinite) ensemble of systems identically prepared in the state $\vec{y}$. But then 
\ba
Y(Q|\lambda\,\vec{y}_1+(1-\lambda)\,\vec{y}_2)=\lambda\,Y(Q|\vec{y}_1)+(1-\lambda)\,Y(Q|\vec{y}_2)\nn
\ea
since $O$ may repeat this interrogation of $S_1,S_2$ a very large number of times. In that case, the `heads' and `tails' ensembles constitute sub-ensembles of the total ensemble of systems being interrogated and $O$ could just record the frequencies of the $Q=$ `yes' outcomes relative to (1) the total ensemble, (2), the `heads' ensemble, and (3) the `tails' ensemble. Taking the relative frequency $\lambda$ of the `heads' and `tails' ensembles into account, it is clear that the total frequency of $Q=$ `yes' outcomes must be of the form above (see also \cite{Hardy:2001jk}).

Using that in the state of no information $\vec{y}=\f{1}{2}\,\vec{1}$ we must have $Y(Q|\f{1}{2}\,\vec{1})=1/2$ for all $Q\in\cq$, where $\vec{1}$ is a vector with each coefficient equal to $1$ (in the basis corresponding to the informationally complete set), we show in appendix \ref{app_Born} that this implies affine-linearity in the state $\vec{y}$
\ba
Y(Q|\vec{y})=Y(\vec{q}|\vec{y})=\f{1}{2}\left(\vec{q}\cdot(2\vec{y}-\vec{1})+1\right)\label{ansatz}.
\ea
Here $\vec{q}\in\mathbb{R}^D$ is a vector which depends on $Q\in\cq$. This formula will ultimately give rise to the Born rule (for projective measurements).

We thus see that every question $Q\in\cq$ can be parametrized by a {\it question vector} $\vec{q}\in\mathbb{R}^D$ such that $Y(Q|\vec{y})\in[0,1]\,\forall\,\vec{y}\in\Sigma$. $O$ can chose to remove any redundancy from his description of $\cq$. Clearly, if $Q,Q'\in\cq$ were both represented by the same $\vec{q}$, then they would give rise to exactly the same `yes'-probabilities in every state. But if $Q,Q'$ are probabilistically not distinguishable, $O$ must regard them as being logically equivalent in his world. $O$ is free to restrict his description of $\cq$ by erasing any questions from it that are redundant through equivalence.\footnote{For example, if $Q\in\cq$ then clearly $Q\wedge Q$ and $Q\vee Q$ can be safely omitted by $O$ from a non-redundant description of $\cq$.} As a result, every question vector $\vec{q}$, if physically permitted at all, will correspond to a unique $Q\in\cq$.

Given the assumption that $S$ always gives an answer to any $Q\in\cq$ if asked by $O$, it is clear that for {\it every} $Q\in\cq$ there exists a state $\vec{y}_Q$ of $S$ encoding the situation that $O$ has asked {\it only} the single question $Q$ to $S$ in the state of no information $\vec{y}=\f{1}{2}\vec{1}$ and received a `yes' answer (i.e., $\vec{y}_Q$ is the updated state after receiving $Q=$`yes' relative to $\vec{y}=\f{1}{2}\vec{1}$). We shall make one natural (but non-trivial) requirement: since $O$ had precisely $0$ \texttt{bits} of information prior to asking $Q$ and $\vec{y}_Q$ corresponds to only having received the answer to this question, $\vec{y}_Q$ shall encode precisely $1$ independent \texttt{bit} of information. We thus demand that for every $Q\in\cq$ there exists $\vec{y}_Q\in\Sigma$ with $I(\vec{y}_Q)=1$ \texttt{bit} such that $Y(Q|\vec{y}_Q)=1$.

This concludes our review of the landscape of inference theories. 


\subsection{The quantum principles as rules on information acquisition}

Within this landscape, we shall impose five rules on the acquisition of information of $O$ about a composite system $S$ of $N\in\mathbb{N}$ generalized bits (or gbits) from Ref.\ \cite{Hoehn:2014uua} to which we refer for motivation. The rules are given both in colloquial and mathematical form. For clarification we shall attach the number $N$ henceforth to $\cq_N,\Sigma_N,\ct_N$. The first two principles assert a limit on the information available to $O$ and the existence of complementarity.

\begin{pr}\label{lim}{\bf(Limited Information)}
\emph{``The observer $O$ can acquire maximally $N\in\mathbb{N}$ {\it independent} \texttt{bits} of information about the system $S$ at any moment of time.''} \\
There exists a maximal set $Q_i$, $i=1,\ldots,N$, of $N$ mutually maximally independent and compatible questions in $\cq_N$ and no subset in $\cq_N$ can contain more than $N$ questions with that property. 
\end{pr}

\begin{pr}\label{unlim}{\bf(Complementarity)}
\emph{``The observer $O$ can always get up to $N$ \emph{new} independent \texttt{bits} of information about the system $S$. But whenever $O$ asks $S$ a new question, he experiences no net loss in his total amount of information about $S$.''}\\
There exists another maximal set $Q_i'$, $i=1,\ldots,N$, of $N$ mutually maximally independent and compatible questions in $\cq_N$ such that $Q'_i,Q_i$ are maximally complementary and $Q'_i,Q_{j\neq i}$ are maximally compatible and independent.
\end{pr}

The systems are thus characterized by the number $N$. Principles \ref{lim} and \ref{unlim} are conceptually motivated by earlier proposals by Rovelli \cite{Rovelli:1995fv} and Brukner and Zeilinger \cite{zeilinger1999foundational,Brukner:ys,Brukner:1999qf,Brukner:2002kx,brukner2009information}. However, they do not suffice. We also require $O$ not to gain or lose information without asking questions.


\begin{pr}\label{pres}{\bf(Information Preservation)}
\emph{``The total amount of information $O$ has about (an otherwise non-interacting) $S$ is preserved in-between interrogations."}\\
$I(\vec{y})$ is \emph{constant} in time in-between interrogations for (an otherwise non-interacting) $S$.
\end{pr}
In fact, this principle can also be used to define the notion of `non-interacting'. 


In order to render $O$'s world interesting for him, it should be as dynamical and interactive as possible. We shall thus require that it `maximizes' the number of ways in which any given state of $S$ can change in time -- rather than the number of states in which it can be relative to $O$.

\begin{pr}\label{time}{\bf(Time Evolution)}
\emph{``$O$'s `catalogue of knowledge' about $S_N$ evolves \emph{continuously} in time in-between interrogations and every consistent such evolution is physically realizable.''}\\
$\ct_N$ is the maximal set of transformations $T_{\Delta t}$ on states which 
is \emph{continuous} in $\Delta t$ and compatible with principles \ref{lim}-\ref{pres} (and the structure of the theory landscape).
\end{pr}


These four rules on $O$'s acquisition of information about $S$ will determine $(\Sigma_N,\ct_N)$ and informationally complete sets $\cq_{M_N}$, however, not the {\it full} $\cq_N$. We thus add another rule: we shall allow $O$ to ask $S$ {\it any} question which ``makes sense''.

\begin{pr}\label{Q}{\bf(Question Unrestrictedness)}
\emph{``Every question which yields legitimate probabilities for every way of preparing $S$ is physically realizable by $O$.''}\\
Every question vector $\vec{q}\in\mathbb{R}^{D_N}$ which satisfies $Y(\vec{q}|\vec{y})\in[0,1]$ $\forall\,\vec{y}\in\Sigma_N$ and for which there exists $\vec{y}_Q\in\Sigma_N$ with $I(\vec{y}_Q)=1$ \texttt{bit} such that $Y(\vec{q}|\vec{y}_Q)=1$ corresponds to a $Q\in\cq_N$.
 \end{pr}

It is our task to derive what the triple $(\cq_N,\Sigma_N,\ct_N)$ compatible with the rules is. As shown in \cite{Hoehn:2014uua}, there are only two solutions to these five principles within the established landscape of theories: in this manuscript we shall complete the proof that one solution is standard qubit quantum theory and, as exhibited in a companion paper \cite{hw2}, the second solution is rebit quantum theory, i.e.\ two level systems over real Hilbert spaces. The second solution is mathematically a subcase of the former and also experimentally realizable in a laboratory. Therefore, the above five rules are physically sufficient. If, however, one wishes to discriminate between the two solutions, one may invoke the following additional rule adapted from \cite{masanes2011derivation,Mueller:2012ai,barrett2007information,Dakic:2009bh,Masanes:2011kx,de2012deriving,Masanes:2012uq}:

\begin{pr}\label{loc}{\bf(Tomographic Locality)}
\emph{``If $S$ is a composite system, $O$ can determine its state by interrogating only its subsystems."}
\end{pr}

It follows from \cite{Hoehn:2014uua} that this last rule eliminates rebits in favour of qubits. We shall appeal to tomographic locality in this manuscript {\it solely} for this purpose.\footnote{In fact, this rule is quite possibly a partially redundant addition. At least in the context of generalized probability theories \cite{masanes2011derivation,Mueller:2012ai,barrett2007information,Dakic:2009bh,Masanes:2011kx,de2012deriving,Masanes:2012uq} tomographic locality implies some of the properties that already follow from the other rules.}

More precisely, we shall prove that principles \ref{lim}--\ref{loc} are equivalent to (part of) the textbook axioms:
\begin{claim}
The only solution to principles \ref{lim}--\ref{loc} is qubit quantum theory where
\begin{itemize}
\item $\Sigma_N\simeq\text{convex hull of }\mathbb{CP}^{2^N-1}$ is the space of $2^N\times2^N$ density matrices over $\mathbb{C}^{2^N}$,
\item states evolve unitarily according to $\ct_N\simeq\rm{PSU}(2^N)$ and the equation describing the state dynamics is (equivalent to) the von Neumann evolution equation.
\item $\cq_{N}\simeq\mathbb{CP}^{2^N-1}$ is (isomorphic to) the set of projective measurements onto the $+1$ eigenspaces of $N$-qubit Pauli operators\footnote{The set of Pauli operators is given by all hermitian operators on $\mathbb{C}^{2^N}$ with two eigenvalues $\pm1$ of equal eigenspace dimensions.} and the probability for $Q\in\cq_N$ to be answered with `yes' in some state is given by the Born rule for projective measurements.
\end{itemize}
\end{claim}

\subsection{Summary of previous results and strategy}

The essential steps of the proof and derivation of qubit quantum theory, involving results from \cite{Hoehn:2014uua}, can be summarized diagrammatically, see figure \ref{fig_strat}.
\begin{figure}[hbt!]
\begin{center}
\psfrag{L}{{limited information}}
\psfrag{C}{{complementarity}}
\psfrag{p}{{information preservation}}
\psfrag{t}{{time evolution}}
\psfrag{o}{{\hspace*{1.4cm}tomographic locality}}
\psfrag{q}{independence, compatibility and cor-}
\psfrag{q1}{relation structure on $\cq_N$  (in \cite{Hoehn:2014uua})}
\psfrag{i}{linear reversible time evolution}
\psfrag{i1}{and information measure (in \cite{Hoehn:2014uua})}
\psfrag{1}{$\Sigma_1$ is a Bloch ball with $D_1=3$, and}
\psfrag{4}{$\ct_1\simeq\rm{SO}(3)\simeq\rm{PSU}(2)$  (in \cite{Hoehn:2014uua})}
\psfrag{2}{$\ct_2\simeq\rm{PSU}(4)$ and}
\psfrag{3}{ $\Sigma_2$ is the convex hull of $\mathbb{CP}^3$ (in sec.\ \ref{sec_n2})}
\psfrag{n}{$\Sigma_N,\ct_N$ for $N>2$ (in sec.\ \ref{sec_n>2})}
\psfrag{zz}{question unrestrictedness}
\psfrag{q2}{$\cq_N$ for arbitrary $N$ (in sec.\ \ref{sec_qn})}
{\includegraphics[scale=.6]{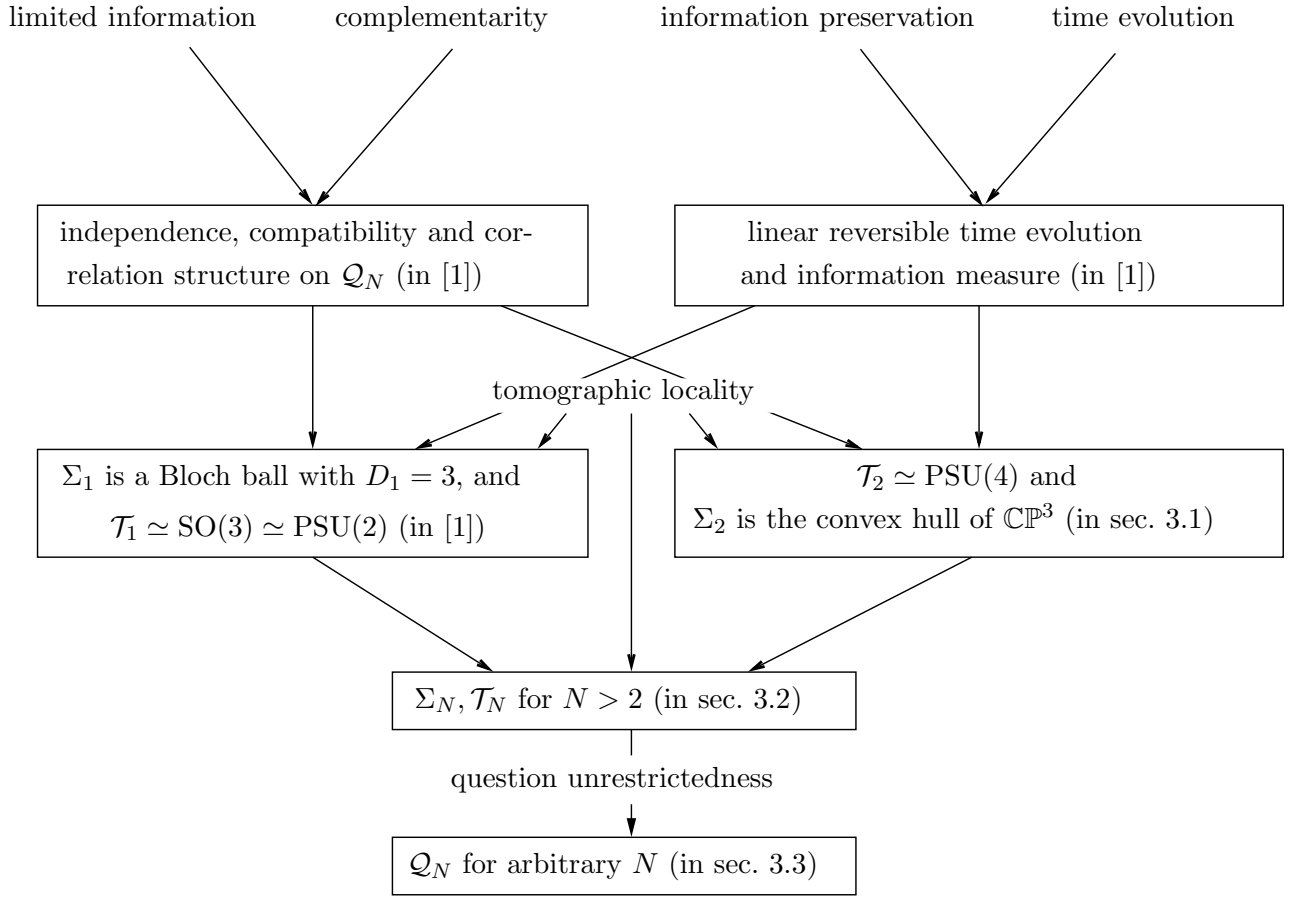}}
\caption{\small Strategy and main steps of the reconstruction.}\label{fig_strat}
\end{center}
\end{figure}

In particular, in \cite{Hoehn:2014uua} the entire compatibility, complementarity and independence structure of any informationally complete set $\cq_{M_N}$ for arbitrary $N$ is derived, showing that the logical connective $*$ in (\ref{composite}) which can be used to build $\cq_{M_N}$ from subsystem questions must either be the XNOR or the (up to an overall negation) equivalent XOR. As a by-product, it is demonstrated how entanglement, monogamy, and the correlation structure for arbitrarily many qubits follow from principles \ref{lim} and \ref{unlim} alone. Furthermore, principles \ref{pres} and \ref{time}, together with elementary operational conditions, can be shown to entail (a) a linear {\it reversible} time evolution of the Bloch vector $\vec{r}=2\,\vec{y}-\vec{1}$ under a continuous one-parameter matrix group, and (b) a quadratic information measure
\ba
\alpha(y_i)=(2\,y_i-1)^2.\label{infomeasure}
\ea 
The total information $I_N(\vec{y})=|\vec{r}|^2$, quantifying $O$'s information about $S$ is thus the square norm of the Bloch vector \cite{Hoehn:2014uua}. This quadratic information measure was earlier proposed by Brukner and Zeilinger from a different perspective \cite{Brukner:1999qf,Brukner:ly,Brukner:2001ve,Brukner:ys,brukner2009information}. Finally, it is demonstrated in \cite{Hoehn:2014uua} how the conjunction of these results correctly yields the three-dimensional Bloch ball together with its isometry group $\SO(3)\simeq\rm{PSU}(2)$ as the state space $\Sigma_1$ and time evolution group $\ct_1$, respectively, for a single qubit (i.e., the $N=1$ case). It was also argued that $\cq_1=\mathbb{CP}^1$. But the reconstruction of $(\cq_N,\Sigma_N,\ct_N)$ was left open for $N>1$.

\subsection{Pure states}
With principle \ref{lim} at our disposal, we can define a notion of {\it pure state}: a pure state of $S_N$ (a composite system of $N$ gbits) is a state of maximal information (and thus of maximal length) in which $O$ knows the maximal amount of $N$ independent \texttt{bits} of information.\footnote{We emphasize the difference to reconstructions within the context of generalized probability theories \cite{masanes2011derivation,Mueller:2012ai,barrett2007information,Dakic:2009bh,Masanes:2011kx,de2012deriving,Masanes:2012uq} where pure states are simply defined to be the extremal states of the convex state space. } See \cite{Hoehn:2014uua} for a more in depth discussion of this informational notion of pure states. 

\section{Reconstruction}

These results will be exploited in the sequel to extend the reconstruction to arbitrary $N>1$ and thus to prove the claim given in the previous section. This will complete the work started in \cite{Hoehn:2014uua}.

\subsection{$N=2$ qubits}\label{sec_n2}

Principles \ref{lim}, \ref{unlim} and \ref{loc} imply that an informationally complete question set for two qubits is given by six individual questions $\{Q_{x_1},Q_{y_1},Q_{z_1},Q_{x_2},Q_{y_2},Q_{z_2}\}$ about qubit 1 and 2 and by nine `correlation questions' $\{Q_{xx},Q_{xy},Q_{xz},Q_{yx},Q_{yy},Q_{yz},Q_{zx},Q_{zy},Q_{zz}\}$, where, e.g., $Q_{xx}:=Q_{x_1}\leftrightarrow Q_{x_2}$ represents the question `are the answers to $Q_{x_1}$ and $Q_{x_2}$ the same?' and $\leftrightarrow$ denotes the XNOR connective.\footnote{We recall that $Q\leftrightarrow Q'=1$ if $Q=Q'$ and $Q\leftrightarrow Q'=0$ otherwise.} Proving this statement is quite non-trivial and takes a number of steps which, for reasons of space, we shall not summarize here. Instead, we shall simply use this result and refer the interested reader to Ref.\ \cite{Hoehn:2014uua} for a constructive proof of it. 

For example, for two spin-$\f{1}{2}$ particles $Q_{x_1},Q_{xx}$ could represent the questions `is the spin of qubit 1 up in $x$-direction?' and `are the spins of qubit 1 and 2 correlated in $x$-direction?', respectively. The compatibility, complementarity and correlation structure of these questions, ensuing from principles \ref{lim} and \ref{unlim}, is derived in \cite{Hoehn:2014uua} and is represented in terms of correlation triangles in figure \ref{fig_corr}.
\begin{figure}[hbt!]
\begin{center}
\psfrag{+}{$+$}
\psfrag{-}{\hspace*{-.1cm}$-$}
\psfrag{1}{$Q_{x_1}$}
\psfrag{2}{\hspace*{-.2cm}$Q_{y_1}$}
\psfrag{3}{$Q_{z_1}$}
\psfrag{1p}{$Q_{x_2}$}
\psfrag{2p}{$Q_{y_2}$}
\psfrag{3p}{$Q_{z_2}$}
\psfrag{11}{$Q_{xx}$}
\psfrag{22}{$Q_{yy}$}
\psfrag{12}{$Q_{xy}$}
\psfrag{33}{$Q_{zz}$}
\psfrag{13}{$Q_{xz}$}
\psfrag{21}{$Q_{yx}$}
\psfrag{23}{$Q_{yz}$}
\psfrag{31}{$Q_{zx}$}
\psfrag{32}{$Q_{zy}$}
\psfrag{i}{\hspace*{-.45cm}\footnotesize identify}
{\includegraphics[scale=.3]{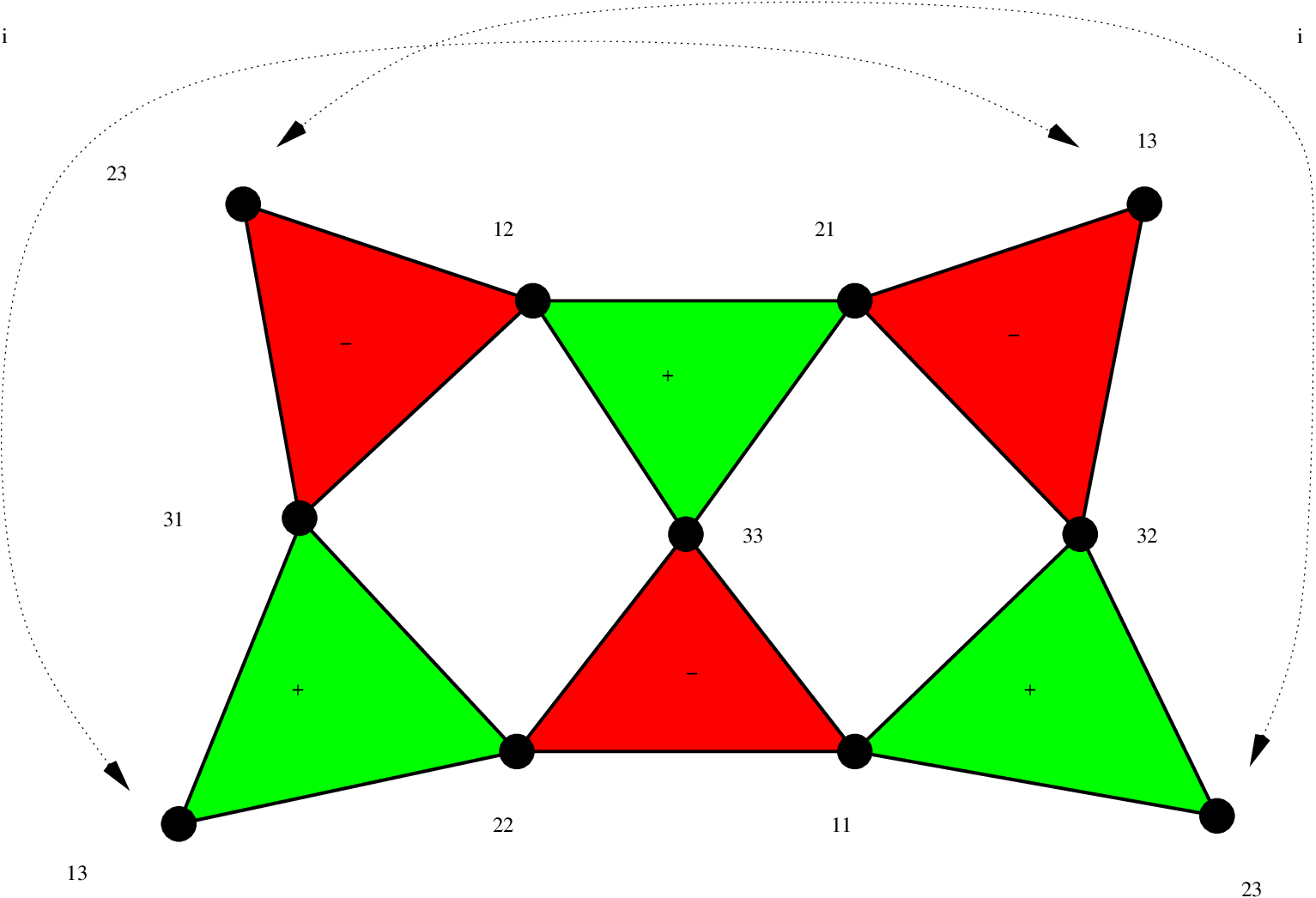}}\\\vspace*{1cm}
{\includegraphics[scale=.3]{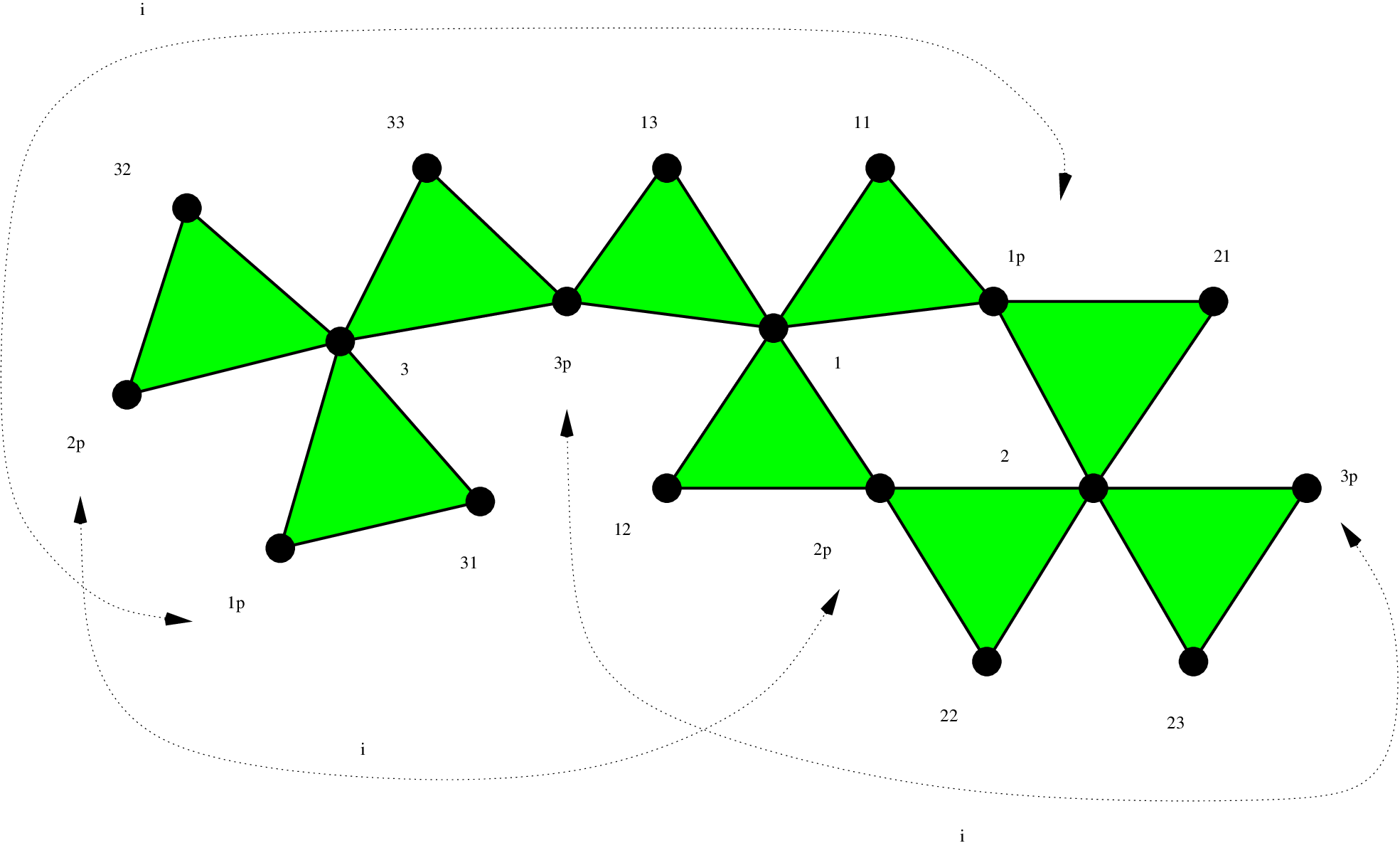}}
\caption{\small The compatibility, complementarity and correlation structure of an informationally complete set for {two qubits}. If two questions are connected by an edge, they are compatible. If two questions are {\it not} connected by an edge, they are complementary. Red triangles denote {\it odd} (or anti-)correlation; for instance, $Q_{zz}=\neg(Q_{xx}\leftrightarrow Q_{yy})$. Green triangles symbolize {\it even} correlation; for example, $Q_{zz}=Q_{xy}\leftrightarrow Q_{yx}$. Every question resides in exactly three triangles and is thereby compatible with six and complementary to eight other questions. (See \cite{Hoehn:2014uua} for further details.) }\label{fig_corr}
\end{center}
\end{figure}

For the sequel, it is important to note that we could have equally chosen to use the XOR instead of the XNOR connective to build up composite questions from individual questions (the XOR is up to an overall negation equivalent to the XNOR) \cite{Hoehn:2014uua}. For example, in that case and instead of $Q_{xx}$ as defined above, we would have $\tilde{Q}_{xx}:= \neg(Q_{x_1}\leftrightarrow Q_{x_2})$ as an``anti-correlation question", corresponding to ``are the answers to $Q_{x_1}$ and $Q_{x_2}$ different?." This would yield a logically equivalent representation of $O$'s experiences in his world, however, with flipped correlation structure (e.g., with odd correlation triangles in figure \ref{fig_corr} replaced by even ones, and vice versa). For $N>2$ gbits, different conventions of how to build up composite questions from the individual ones using the allowed XNOR or XOR can lead to many equivalent representations that will also arise in the reconstruction below. Therefore, to fix the representation, we shall henceforth make the convention that we build up composite questions from the individual ones for $N\geq2$ solely by the XNOR connective.

We note that a {\it pure state} as a state of maximal information will have length
\ba
I_{N=2}(\vec{r}_{\rm pure})=|\vec{r}_{\rm pure}|^2=3\,\texttt{bits},\nn
\ea
corresponding to $O$ knowing the answers to two independent and compatible questions with certainty (principle \ref{lim}) -- this yields two {\it independent} \texttt{bits} -- and, on account, of the XNOR properties also knowing the correlation of these questions -- this yields a third {\it dependent} \texttt{bit} \cite{Hoehn:2014uua}. For instance, if $O$ knows the answers to $Q_{x_1},Q_{x_2}$, he evidently knows the answer to $Q_{xx}$ too. By principle \ref{pres}, the time evolution image of any such state will feature the same length and thus constitutes a pure state too.

\subsubsection{Maximal mutually complementary sets of questions}

The three questions $\{Q_{x},Q_{y},Q_{z}\}$ form a single maximal mutually complementary set of questions for a single qubit. It is also useful to group the 15 questions for two qubits into {\it maximal mutually complementary sets} such that no further question can be added to such a set which would be complementary to all others in the same set too. This results in six complementarity sets, each containing five questions, which can be understood and represented conveniently in terms of question graphs
\ba\label{eq:pentineq}
\psfrag{x1}{\small $Q_{x_1}$}
\psfrag{y1}{\small $Q_{y_1}$}
\psfrag{z1}{\small $Q_{z_1}$}
\psfrag{x2}{\small $Q_{x_2}$}
\psfrag{y2}{\small $Q_{y_2}$}
\psfrag{z2}{\small $Q_{z_2}$}
\psfrag{xx}{\small $Q_{xx}$}
\psfrag{xy}{\small $Q_{xy}$}
\psfrag{xz}{\small $Q_{xz}$}
\psfrag{yx}{\small $Q_{yx}$}
\psfrag{yy}{\small $Q_{yy}$}
\psfrag{yz}{\small $Q_{yz}$}
\psfrag{zx}{\small $Q_{zx}$}
\psfrag{zy}{\small $Q_{zy}$}
\psfrag{zz}{\small $Q_{zz}$}
\includegraphics[scale=0.4]{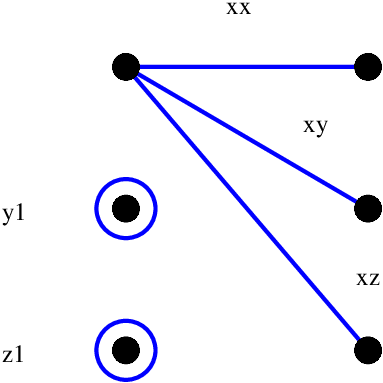} \q\q\q\q\q\q  \rbx{$\text{Pent}_1=\{Q_{xx},Q_{xy},Q_{xz},Q_{y_1},Q_{z_1}\},$}\nn\\\nn\\
\psfrag{x1}{\small $Q_{x_1}$}
\psfrag{y1}{\small $Q_{y_1}$}
\psfrag{z1}{\small $Q_{z_1}$}
\psfrag{x2}{\small $Q_{x_2}$}
\psfrag{y2}{\small $Q_{y_2}$}
\psfrag{z2}{\small $Q_{z_2}$}
\psfrag{xx}{\small $Q_{xx}$}
\psfrag{xy}{\small $Q_{xy}$}
\psfrag{xz}{\small $Q_{xz}$}
\psfrag{yx}{\small $Q_{yx}$}
\psfrag{yy}{\small $Q_{yy}$}
\psfrag{yz}{\small $Q_{yz}$}
\psfrag{zx}{\small $Q_{zx}$}
\psfrag{zy}{\small $Q_{zy}$}
\psfrag{zz}{\small $Q_{zz}$}
\includegraphics[scale=0.4]{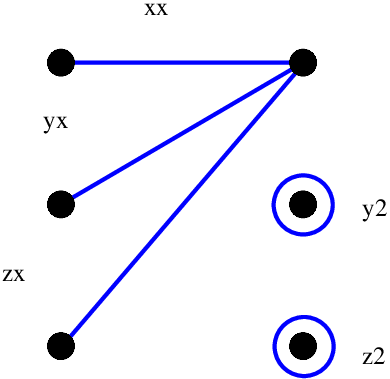} \q\q\q\q\q   \hspace*{-.1cm} \rbx{$\text{Pent}_2=\{Q_{xx},Q_{yx},Q_{zx},Q_{y_2},Q_{z_2}\},$} \nn\\\nn\\
\psfrag{x1}{\small $Q_{x_1}$}
\psfrag{y1}{\small $Q_{y_1}$}
\psfrag{z1}{\small $Q_{z_1}$}
\psfrag{x2}{\small $Q_{x_2}$}
\psfrag{y2}{\small $Q_{y_2}$}
\psfrag{z2}{\small $Q_{z_2}$}
\psfrag{xx}{\small $Q_{xx}$}
\psfrag{xy}{\small $Q_{xy}$}
\psfrag{xz}{\small $Q_{xz}$}
\psfrag{yx}{\small $Q_{yx}$}
\psfrag{yy}{\small $Q_{yy}$}
\psfrag{yz}{\small $Q_{yz}$}
\psfrag{zx}{\small $Q_{zx}$}
\psfrag{zy}{\small $Q_{zy}$}
\psfrag{zz}{\small $Q_{zz}$}
\includegraphics[scale=0.4]{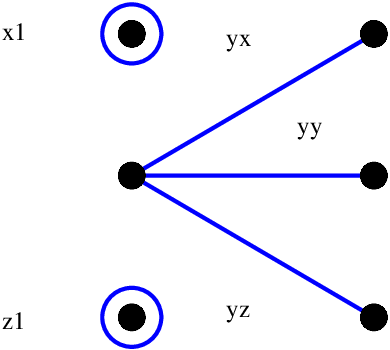} \q\q\q\q\q  \hspace*{.38cm}  \rbx{$\text{Pent}_3=\{Q_{yx},Q_{yy},Q_{yz},Q_{x_1},Q_{z_1}\},$}\nn\\\nn\\
\psfrag{x1}{\small $Q_{x_1}$}
\psfrag{y1}{\small $Q_{y_1}$}
\psfrag{z1}{\small $Q_{z_1}$}
\psfrag{x2}{\small $Q_{x_2}$}
\psfrag{y2}{\small $Q_{y_2}$}
\psfrag{z2}{\small $Q_{z_2}$}
\psfrag{xx}{\small $Q_{xx}$}
\psfrag{xy}{\small $Q_{xy}$}
\psfrag{xz}{\small $Q_{xz}$}
\psfrag{yx}{\small $Q_{yx}$}
\psfrag{yy}{\small $Q_{yy}$}
\psfrag{yz}{\small $Q_{yz}$}
\psfrag{zx}{\small $Q_{zx}$}
\psfrag{zy}{\small $Q_{zy}$}
\psfrag{zz}{\small $Q_{zz}$}
\includegraphics[scale=0.4]{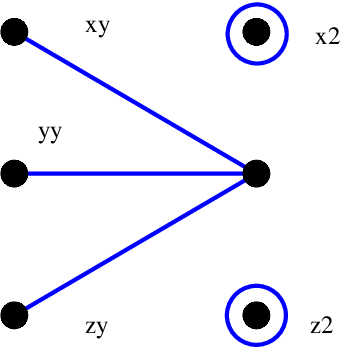} \q\q\q\q\q \hspace*{-.1cm}  \rbx{$\text{Pent}_4=\{Q_{xy},Q_{yy},Q_{zy},Q_{x_2},Q_{z_2}\},$} \nn\\\nn\\
\psfrag{x1}{\small $Q_{x_1}$}
\psfrag{y1}{\small $Q_{y_1}$}
\psfrag{z1}{\small $Q_{z_1}$}
\psfrag{x2}{\small $Q_{x_2}$}
\psfrag{y2}{\small $Q_{y_2}$}
\psfrag{z2}{\small $Q_{z_2}$}
\psfrag{xx}{\small $Q_{xx}$}
\psfrag{xy}{\small $Q_{xy}$}
\psfrag{xz}{\small $Q_{xz}$}
\psfrag{yx}{\small $Q_{yx}$}
\psfrag{yy}{\small $Q_{yy}$}
\psfrag{yz}{\small $Q_{yz}$}
\psfrag{zx}{\small $Q_{zx}$}
\psfrag{zy}{\small $Q_{zy}$}
\psfrag{zz}{\small $Q_{zz}$}
\includegraphics[scale=0.4]{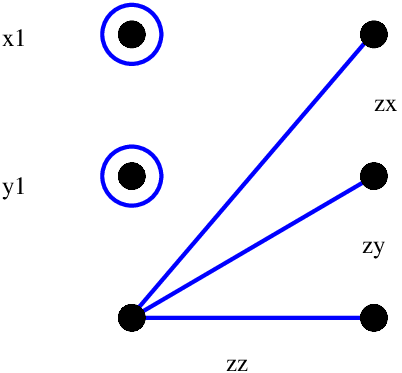} \q\q\q\q\q  \hspace*{.37cm}  \rbx{$\text{Pent}_5=\{Q_{zx},Q_{zy},Q_{zz},Q_{x_1},Q_{y_1}\},$}\nn\\\nn\\
\psfrag{x1}{\small $Q_{x_1}$}
\psfrag{y1}{\small $Q_{y_1}$}
\psfrag{z1}{\small $Q_{z_1}$}
\psfrag{x2}{\small $Q_{x_2}$}
\psfrag{y2}{\small $Q_{y_2}$}
\psfrag{z2}{\small $Q_{z_2}$}
\psfrag{xx}{\small $Q_{xx}$}
\psfrag{xy}{\small $Q_{xy}$}
\psfrag{xz}{\small $Q_{xz}$}
\psfrag{yx}{\small $Q_{yx}$}
\psfrag{yy}{\small $Q_{yy}$}
\psfrag{yz}{\small $Q_{yz}$}
\psfrag{zx}{\small $Q_{zx}$}
\psfrag{zy}{\small $Q_{zy}$}
\psfrag{zz}{\small $Q_{zz}$}
\includegraphics[scale=0.4]{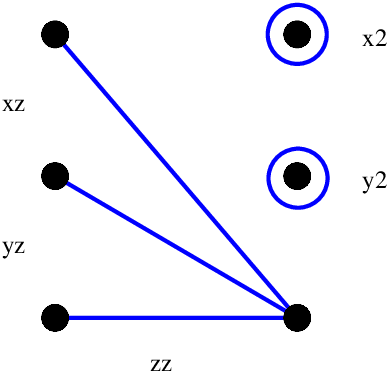}  \q\q\q\q\q \rbx{$\text{Pent}_6=\{Q_{xz},Q_{yz},Q_{zz},Q_{x_2},Q_{y_2}\}.$} 
\ea
The vertices correspond to individual questions while the edges connecting them represent the corresponding correlation questions. Vertices on the left correspond to qubit 1 and are compatible with the vertices on the right, corresponding to qubit 2, but not with each other. Vertices are compatible with edges if and only if they are vertices of the latter and edges are compatible if and only if they do not intersect in a vertex \cite{Hoehn:2014uua}. These complementarity relations are conveniently represented in figure \ref{fig:pentagon} in terms of a lattice of pentagons, where each pentagon corresponds to one of the six sets in (\ref{eq:pentineq}).
\begin{figure}[h!]
\begin{center}
\psfrag{x1}{\small $Q_{x_1}$}
\psfrag{y1}{\small $Q_{y_1}$}
\psfrag{z1}{\small $Q_{z_1}$}
\psfrag{x2}{\small $Q_{x_2}$}
\psfrag{y2}{\small $Q_{y_2}$}
\psfrag{z2}{\small $Q_{z_2}$}
\psfrag{xx}{\small $Q_{xx}$}
\psfrag{xy}{\small $Q_{xy}$}
\psfrag{xz}{\small $Q_{xz}$}
\psfrag{yx}{\small $Q_{yx}$}
\psfrag{yy}{\small $Q_{yy}$}
\psfrag{yz}{\small $Q_{yz}$}
\psfrag{zx}{\small $Q_{zx}$}
\psfrag{zy}{\small $Q_{zy}$}
\psfrag{zz}{\small $Q_{zz}$}
\psfrag{p1}{\small $1$}
\psfrag{p2}{\small $2$}
\psfrag{p3}{\small $3$}
\psfrag{p4}{\small $4$}
\psfrag{p5}{\small $5$}
\psfrag{p6}{\small $6$}
\includegraphics[scale=0.4]{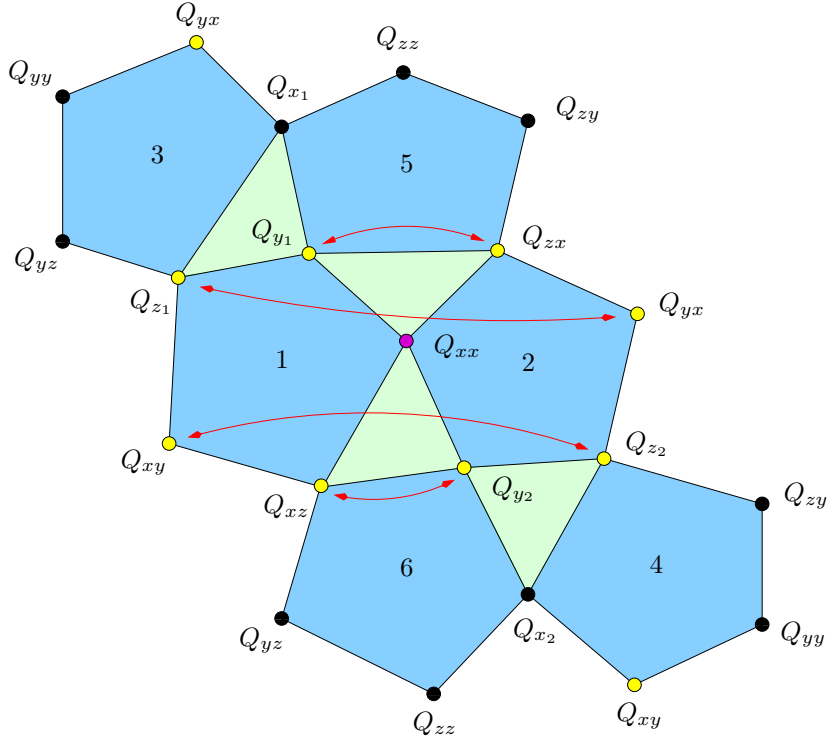}
\caption{\small The six maximal mutually complementary question sets (\ref{eq:pentineq}) represented as pentagons. In contrast to figure \ref{fig_corr}, if two questions lie in the same pentagon or are connected by an edge it means they are complementary (in all other cases they are compatible). Every question appears in precisely two pentagons such that every pentagon is connected to all other five. The green triangles are four of 20 maximal complementarity triangles (see appendix \ref{triangle-app}). The red arrows denote the information swap between pentagons 1 and 2 in (\ref{eq:pentswapq}) which leaves all pentagon equalities (\ref{pentin}) invariant and defines the time evolution generator (\ref{eq:pentswapT}).} \label{fig:pentagon}
 \end{center}\end{figure}
It can be easily checked, using such question graphs, that no other maximal complementarity sets of five or more questions exist. However, there also exist 20 maximal sets of three elements,  four of which are shown as green triangles in figure \ref{fig:pentagon}.  Since these 20 sets will only be employed for consistency checks of the complementarity inequalities (\ref{compstrong}) but not for the main flow of the arguments, we choose to display and explain them using the question graphs in appendix \ref{triangle-app}. There are no other maximal complementarity sets for two qubits.

\subsubsection{Constraints on the information distribution over the questions}\label{sec_infodist}

For pure states of a single qubit, the single maximal complementarity set carries precisely 1 \texttt{bit} of information, $I_{{N=1}}=\alpha_x+\alpha_y+\alpha_z=r_x^2+r_y^2+r_z^2=1$ \texttt{bit} ($r_i$ are the Bloch vector components) which, according to principle \ref{pres}, is a conserved `charge' of time evolution. This {\it defines} the unitary time evolution group $\rm{PSU}(2)$ and the Bloch sphere of pure states for a single qubit \cite{Hoehn:2014uua}. We shall now show the analogue for two qubits. 

Since every question is contained in precisely two pentagons, the sum of the information contained in each pentagon yields twice the total information of $O$ about the two qubits 
\ba
\sum_{a=1}^6\,I(\text{Pent}_a)=2\left(\sum_{i=x,y,z}(\alpha_{i_1}+\alpha_{i_2})+\sum_{i,j=x,y,z}\alpha_{ij}\right)=2\,I_{{N=2}}(\vec{r}),\label{totinfo}
\ea
where, thanks to (\ref{compstrong}), $0\,\text{\texttt{bits}}\leq I(\text{Pent}_a)=\sum_{i\in\text{Pent}_a}\alpha_i\leq 1\,\text{\texttt{bit}}$ is the sum of the information carried by the five questions in pentagon $a$. Since for pure states $I_{{N=2}}(\vec{r}_{\rm pure})=3$ \texttt{bits}, it follows that every pure state must satisfy what we shall call the {\it pentagon equalities}
\ba
\text{\bf pure states:}\q\q\q\q I(\text{Pent}_a)\equiv1\,\text{\texttt{bit}},\q\q a=1,\ldots,6.\label{pentin}
\ea
In analogy to the single qubit case, every pentagon therefore carries precisely one \texttt{bit} of information for every pure state. Hence, the pentagon equalities must also be conserved `informational charges' of time evolution. We shall see shortly in section \ref{sec_psu4} that these relations single out the unitary group for two qubits. There are no such conserved informational charges for the maximal complementarity sets consisting of only three elements (see appendix \ref{triangle-app}).

These identities are remarkable because the underlying probabilities $y_i$ in $\alpha_i=(2\,y_i-1)^2$ of the 15 questions are independent coordinates on $\Sigma_2$ and thus do not satisfy any linear identities for all pure states. This observation emphasises the strength of considering the information content in the questions in addition to their probabilities in quantum theory. In fact, writing $|\psi\rangle=\alpha|x_+x_+\rangle+\beta|x_-x_-\rangle+\gamma|x_+x_-\rangle+\delta|x_-x_+\rangle$ for an arbitrary two-qubit pure state, where $|\alpha|^2+|\beta|^2+|\gamma|^2+|\delta|^2=1$ and $x_\pm$ stands for `up/down' in $x$-direction, one can easily verify (using a computer programme) that quantum theory actually satisfies the pentagon equalities (\ref{pentin}) for the quadratic measure $\alpha_i=(2\,y_i-1)^2$ (where, e.g., $y_{x_1}$ is the probability that the spin of qubit 1 is `up', $y_{xx}$ is the probability that the spins of qubit 1 and 2 are correlated in $x$-direction, etc.).
For example, to put the pentagon identities (3.3) in the case of quantum theory into a more familiar language, the identity for $\text{Pent}_1$ reads for pure states
 \begin{equation}
 I(\text{Pent}_1) = \langle \sigma_y\otimes\mathds{1}\rangle^2+\langle \sigma_z\otimes\mathds{1}\rangle^2+\langle\sigma_x\otimes\sigma_x\rangle^2+\langle\sigma_x\otimes\sigma_y\rangle^2+\langle\sigma_x\otimes\sigma_z\rangle^2=1,\nonumber
\end{equation}
and similarly for the other pentagon identities. These informational pentagon identities (\ref{pentin}) seem to have previously gone unnoticed in quantum theory.

The pentagon equalities have two interesting consequences for pure states. Firstly, $I(\text{Pent}_1)+I(\text{Pent}_3)+I(\text{Pent}_5)-I(\text{Pent}_2)-I(\text{Pent}_4)-I(\text{Pent}_6)=0$ implies that $O$ knows as much individual information about qubit 1 as about qubit 2
\ba
\text{\bf pure states:}\q\q\q\q\alpha_{x_1}+\alpha_{y_1}+\alpha_{z_1}=\alpha_{x_2}+\alpha_{y_2}+\alpha_{z_2}.\nn
\ea
(Clearly, this identity cannot hold for all states of non-maximal information.) We exhibit further such identities in appendix \ref{triangle-app}. Secondly, the pentagon equalities entail that the amount of information carried by any question is determined by the amount of information carried by the six questions compatible with it -- and vice versa. In terms of the correlation triangles in figure \ref{fig_corr} this results in a `bulk/boundary' relation. For instance, for the three correlation triangles in figure \ref{fig:1bitquestion}, excised from figure \ref{fig_corr}, (\ref{totinfo}, \ref{pentin}) yield
\ba
\text{\bf pure states:}\q\q\q &&\alpha_{z_1}=\text{Boundary}_{z_1}-1\,\text{\texttt{bit}},\q\q\text{where}\label{bdry}\\
&&1\,\text{\texttt{bit}}\leq\text{Boundary}_{z_1}:=\alpha_{x_2}+\alpha_{zx}+\alpha_{y_2}+\alpha_{zy}+\alpha_{zz}+\alpha_{z_2}\leq 2\,\text{\texttt{bit}}.\nn
\ea
The special case $\alpha_{z_1}=1$ \texttt{bit} arises if and only if Boundary${}_{z_1}=2$ \texttt{bits} and the three triangles adjacent to $Q_{z_1}$ thus carry all $3$ \texttt{bits} of information. Analogous relations hold for any other question in figure \ref{fig_corr}.
\begin{figure}[h!]
\begin{center}
\psfrag{i}{}
\psfrag{+}{\small $+$}
\psfrag{z1}{\small $Q_{z_1}$}
\psfrag{x2}{\small $Q_{x_2}$}
\psfrag{y2}{\small $Q_{y_2}$}
\psfrag{z2}{\small $Q_{z_2}$}
\psfrag{zx}{\small $Q_{zx}$}
\psfrag{zy}{\small $Q_{zy}$}
\psfrag{zz}{\small $Q_{zz}$}
\psfrag{az1}{\small $\alpha_{z_1}=1$ \texttt{bit}}
\psfrag{azx}{\small \hspace*{.3cm}$\alpha_{zx}=\alpha_{x_2}$}
\psfrag{azy}{\small $\alpha_{zy}=\alpha_{y_2}$}
\psfrag{azz}{\small\hspace*{.3cm}$\alpha_{zz}=\alpha_{z_2}$}
\hspace{-1.5cm}\includegraphics[scale=0.4]{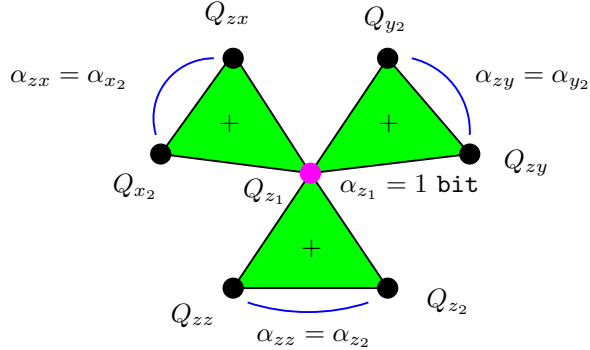}
\caption{\small If any question carries precisely 1 \texttt{bit}, the adjacent correlation triangles carry all remaining information -- also for states of non-maximal information. Moreover, within any correlation triangle, the information contained in the two other questions must be equal. } \label{fig:1bitquestion}
\end{center}
\end{figure}

It is easy to convince oneself, using that any question in a correlation triangle of figure \ref{fig_corr} is either the correlation or anti-correlation of the other two questions in the triangle, that whenever one question carries $1$ \texttt{bit} of information, the other two questions in the correlation triangle must carry equal amounts. For example, if the central vertex $Q_{z_1}$ in figure \ref{fig:1bitquestion} carries $\alpha_{z_1}=1$ \texttt{bit}, then $\alpha_{zz}=\alpha_{z_2}$, etc.\ as indicated.\footnote{E.g., if $O$ knew with certainty that $Q_{z_1}=$ `yes', he would know that the answers to $Q_{zz},Q_{z_2}$ are correlated, such that $y_{zz}=y_{z_2}$ and hence $\alpha_{zz}=\alpha_{z_2}$. (Note that $y_{zz}=y_{z_2}=\f{1}{2}$ is possible too, of course.)} While this must hold for states of non-maximal information too, for pure states it also follows directly from the pentagon identities (\ref{pentin}): e.g., inserting $\alpha_{z_1}=1$ \texttt{bit}, and thus $\alpha_i=0$ for any $Q_i$ complementary to $Q_{z_1}$, into $I(\text{Pent}_5)+I(\text{Pent}_6)-I(\text{Pent}_2)-I(\text{Pent}_4)=0$ implies directly $\alpha_{zz}=\alpha_{z_2}$. The analogous results can be similarly derived for all triangle neighbours of any $\alpha_i=1$ \texttt{bit}.

These observations will become valuable shortly.

\subsubsection{Derivation of the unitary group}\label{sec_psu4}

Any given time evolution acts linearly and continuously on the states (in between interrogations), $r_i(t)=T_{ij}(t)\,r_j(0)$, where $\vec{r}=2\,\vec{y}-\vec{1}\,\in\mathbb{R}^{15}$ is the generalized Bloch-vector, and constitutes a one-parameter subgroup of $\ct_2$ which itself is a group \cite{Hoehn:2014uua}. Principle \ref{pres} asserts that the total information is a `conserved charge' of time evolution, $I_{{N=2}}(T(t)\cdot\vec{r})=I_{{N=2}}(\vec{r})$. Since the total information is given by the square norm of the Bloch vector (\ref{infomeasure}), this implies that $\ct_2\subset\SO(15)$ (time evolution must be connected to the identity). In fact, $\ct_2$ must be a proper subgroup of $\SO(15)$ because the latter contains transformations that map all $3$ \texttt{bits} of information contained in any pure state into a single question, e.g., $\vec{r}=(1,1,1,0,\ldots,0)$ to $\vec{r}=(\sqrt{3},0,0,0,\ldots,0)$ -- which is illegal. 

In particular, every pure state evolves to a pure state. Therefore, the pentagon equalities (\ref{pentin}) are likewise `conserved charges', such that we must have $I(\text{Pent}_a(T(t)\cdot\vec{r}))=I(\text{Pent}_a(\vec{r}))$, $a=1,\ldots,6$. Given that $\ct_2\subset\SO(15)$ and $T(t_1+t_2)=T(t_1)\cdot T(t_2)=T(t_2)\cdot T(t_1)$, we may write $T(t)=\exp(t \,G)$ for some generator $G\in\so(15)$ which yields (to linear order in $t$)
\ba
 \sum_{i\in\text{Pent}_a,1\leq j\leq 15}r_i\,G_{ij}\,r_j=0,\q\q\q\q a=1,\ldots6, \label{eq:pentGeq}
\ea
where $G_{ij}=-G_{ji}$ since $G\in\so(15)$. This implies, in particular, conservation of the total information $I_{{N=2}}$.

(\ref{eq:pentGeq}) constitute restrictions on both the set of pure states and time evolution generators; {\it any} legal pure state must satisfy (\ref{eq:pentGeq}) for every legal time evolution generator $G$ and, vice versa, {\it any} legal time evolution generator must satisfy (\ref{eq:pentGeq}) for every legal pure state. Of course, at this stage, we neither know what the set of legal pure states nor what the time evolution group $\ct_2$ is. As we shall see shortly, however, the pentagon equalities (\ref{pentin}) and the conditions (\ref{eq:pentGeq}) are sufficient, together with the principles and background assumptions, to single out $\ct_2=\rm{PSU}(4)$ and the two qubit quantum state space. This is subject to the already employed convention to use only the XNOR connective $\leftrightarrow$ (rather than the XOR) for building multipartite questions from the individuals, e.g., $Q_{xx}=Q_{x_1}\leftrightarrow Q_{x_2}$.

To this end, we recall principle \ref{time} which implies that for {\it any} state the set of legal time evolutions is the maximal one compatible with the other principles. Given that the set of all time evolutions forms a group (which acts linearly and state independently on states), the principle thus requires the latter somehow to be `maximal'. In particular, we can check maximality for specific states that we know must be in $\Sigma_N$. Namely, for any set of mutually compatible questions, there must, by definition, exist a state in which these questions are simultaneously answered.
Furthermore, for every set of $N$ mutually compatible and independent questions (as in principle \ref{lim}) there must exist a state for every `yes/no'-answer configuration. For $N=2$ every such state must also respect the correlation structure of figure \ref{fig_corr}. This entails that the set of legal pure states must contain
\ba
\vec{r}=\vec{\delta}_{z_1}+\vec{\delta}_{z_2}+\vec{\delta}_{z_1z_2},\q\q\q \vec{r}=\vec{\delta}_{z_1}-\vec{\delta}_{z_2}-\vec{\delta}_{z_1z_2},\q\q\q \vec{r}=\vec{\delta}_{z_1}+\vec{\delta}_{z_2}-\vec{\delta}_{z_1z_2},\label{legalstates}
\ea
where $\vec{\delta}_i$ denotes a vector in $\mathbb{R}^{15}$ with the $i$-th component equal to $1$ and all others $0$. But (\ref{eq:pentGeq}) must, in particular, be satisfied for these three pure states which results in
\ba
G_{z1z1} + G_{z1z2} + G_{z1(z1z2)}&=&0,\q\q\q
G_{z1z1} - G_{z1z2} - G_{z1(z1z2)}=0,\\\nn
G_{z1z1} + G_{z1z2} - G_{z1(z1z2)}&=&0,\nn
\ea
and thus
\ba
G_{z1z1} = G_{z1z2} = G_{z1(z1z2)}=0.\nn
\ea
It is easy to convince oneself, by repeating the same argument with every correlation or anti-correlation triangle in figure \ref{fig_corr}, that {\it any} legal time evolution generator $G$ must feature
\ba
G_{ij}=0,\q\q\q\text{whenever $Q_i,Q_j$ are compatible.}\label{gij}
\ea

That is, {\it any} legal time evolution generator can only have non-zero components for pairs of indices corresponding to complementary questions. It follows from figure \ref{fig_corr} that every question is complementary to precisely eight questions from the informationally complete set. Since there are 15 questions, there are precisely $15\times8/2=60$ pairs of complementary questions. Thus, given the anti-symmetry $G_{ij}=-G_{ji}$, there could at most be $60$ linearly independent generators satisfying conditions (\ref{eq:pentGeq}) for every pure state. 

We shall now construct such a set of 60 linearly independent generators which satisfy (\ref{gij}) and have a clear operational meaning. However, as we shall see shortly, only 15 of such generators can be consistent with the principles at once. 


Since any two pentagons overlap in precisely one question, there is no transformation which redistributes the information only within a single pentagon and leaves all pentagon equalities invariant. However, for any pair of pentagons there exists a {\it unique} transformation which swaps the information from one pentagon to the other and leaves all other pentagons and all pentagon equalities (\ref{pentin}) invariant. Consider, e.g., pentagons $\text{Pent}_1$ and $\text{Pent}_2$ in figure \ref{fig:pentagon}. The red arrows denote the complete information swap ($\longleftrightarrow$ is not to be confused with the XNOR)
\begin{equation}
\alpha_{y_1}\longleftrightarrow \alpha_{zx} \ (\text{Pent}_5), \ \alpha_{z_1}\longleftrightarrow \alpha_{yx} \ (\text{Pent}_3), \ \alpha_{xy}\longleftrightarrow \alpha_{z_2} \ (\text{Pent}_4), \ \alpha_{xz}\longleftrightarrow \alpha_{y_2} \ (\text{Pent}_6) \label{eq:pentswapq}
\end{equation}
between the two pentagons which leaves the composite $\alpha_{xx}$ and all other questions invariant. Since each of the swaps in (\ref{eq:pentswapq}) occurs within precisely one of the remaining four pentagons, all pentagon equalities (\ref{pentin}) are preserved. Such a full swap of information between two pentagon sets is thus a good candidate for a legal time evolution. W.l.o.g.\ this swap transformation can be written as $T=\exp((\pi/2)\,G^{\text{Pent}_1,\text{Pent}_2})$ acting on $\vec{r}$ with
\begin{gather}
G_{ij}^{\text{Pent}_1,\text{Pent}_2}=\delta_{iy_1}\delta_{jzx}+s_1\,\delta_{iz_1}\delta_{jyx}+s_2\,\delta_{ixy}\delta_{jz_2}+s_3\,\delta_{ixz}\delta_{jy_2}-(i\longleftrightarrow j), \label{genansatz}
\end{gather}
where $s_1,s_2,s_3$ are three signs to be determined. Given that there are four linearly independent terms in the generator, one can produce precisely four linearly independent generators from (\ref{genansatz}) by changing the signs $s_1,s_2,s_3$. However, a legal time evolution generator must be consistent with the correlation structure in figure \ref{fig_corr} and the constraints on information distribution of section \ref{sec_infodist}. In appendix \ref{app_swap} it is shown that these constraints uniquely determine the generator candidate (up to an unimportant overall sign) to
\begin{gather}
G_{ij}^{\text{Pent}_1,\text{Pent}_2}=\delta_{iy_1}\delta_{jzx}-\delta_{iz_1}\delta_{jyx}+\delta_{ixy}\delta_{jz_2}-\delta_{ixz}\delta_{jy_2}-(i\longleftrightarrow j). \label{eq:pentswapT}
\end{gather}

For every pair of pentagons there exists such a unique information swap, resulting in $\binom{6}{2}=15$ transformations which are consistent with the correlation structure and the constraints on the information distribution. The form of their generators can be found similarly (see (\ref{swapgen1}, \ref{swapgen2}) in appendix \ref{app_swap}). There are nine swaps leaving a composite and six swaps leaving an individual question as the overlap of the pentagons invariant. As an example for the latter, the information swap between Pent${}_3$ and Pent${}_5$ leaves the individual $\alpha_{x_1}$ invariant and is generated by
\begin{gather}
G^{\text{Pent}_3,\text{Pent}_5}_{ij}=\delta_{iy_1}\delta_{jz_1}-\delta_{iyx}\delta_{jzx}-\delta_{iyy}\delta_{jzy}-\delta_{iyz}\delta_{jzz}-(i\longleftrightarrow j). \label{eq:TSO3x1}
\end{gather}

In Appendix \ref{app_swap}, it is shown that the various sign distributions over these 15 generators, as in (\ref{genansatz}), produce precisely 60 linearly independent generators satisfying (\ref{gij}). Regardless of the sign structure, each of these 60 linearly independent generators thus corresponds precisely to a complete information swap between two pentagon sets and for each pair of pentagon sets there are four linearly independent such swap generators. That is, whatever the resulting time evolution group consistent with (\ref{eq:pentGeq}) may be, it must be fully generated by complete information swaps between pentagons. Clearly, it cannot be generated by all 60 such generators as the only state which would satisfy (\ref{eq:pentGeq}) for all 60 generators is the state of no information $\vec{r}=0$. Indeed, requiring consistency with the correlation structure of figure \ref{fig_corr}, and thus consistency with the convention of only using the XNOR connective for building bipartite questions from individuals, results in one permissible generator candidate per pair of pentagon sets and in precisely the 15 candidate generators exhibited here and in appendix \ref{app_swap}. The time evolution group can thus not be generated by any other than these 15 surviving generator candidates; in fact, the remaining 45 possible sign distributions can be argued to correspond to different conventions (see appendix \ref{app_swap}).

Using a computer algebra program, one can easily check that these 15 surviving information swap generators (\ref{swapgen1}, \ref{swapgen2}) (see appendix \ref{app_swap})
\begin{itemize}
\item[(a)] satisfy the commutator algebra of $\su(4)\simeq\so(6)\simeq\mathfrak{psu}(4)\subset\so(15)$, and
\item[(b)] coincide exactly (in some cases up to an unimportant overall sign) with the adjoint representation 
\ba
(G^i)_{jk}:=f^{ijk}=\f{1}{4}\,\tr([\sigma_j,\sigma_k]\,\sigma_i)\nn
\ea
of the 15 fundamental generators of the unitary group $\SU(4)$. $f^{ijk}$ are the structure constants of $\SU(4)$, the indices $i,j,k$ take the 15 values $x_1,y_1,z_1,x_2,\ldots,xz,xy,\ldots,zz$ (as in our reconstruction) and $\sigma_{x_1}:=\sigma_x\otimes\mathds{1}$, ..., $\sigma_{x_2}:=\mathds{1}\otimes\sigma_x$, ..., $\sigma_{xx}:=\sigma_x\otimes\sigma_x$, ..., $\sigma_{zz}:=\sigma_z\otimes\sigma_z$ and $\sigma_x,\sigma_y,\sigma_z$ are the usual Pauli matrices. In particular, the ordering of coincidence is $G^i\equiv\pm G^{\text{Pent}_a,\text{Pent}_b}$ where $Q_i$ is the single question in $\text{Pent}_a\cap\text{Pent}_b$ which is left invariant by the swap; e.g., $G^{xx}\equiv G^{\text{Pent}_1,\text{Pent}_2}$, etc.

%

\end{itemize}

Next, we must check whether the {(image of any state under the)} full group $\ct_2'$ generated by exponentiating the 15 surviving swap generators (\ref{swapgen1}, \ref{swapgen2}) and their linear combinations is consistent with the principles and thus by principle \ref{time} {whether $\ct_2'$ is} contained in $\ct_2$. Clearly, $\ct_2'$ obeys principle \ref{pres} by construction and the only background assumption which it is not evidently consistent with are the complementarity inequalities (\ref{compstrong}). Similarly, the only structure entailed by principles \ref{lim} and \ref{unlim} that $\ct_2'$ is not evidently consistent with is the correlation structure of figure \ref{fig_corr}. We thus have to expose $\ct'_2$ to a few non-trivial consistency checks. In appendix \ref{app_A} it is shown that
\begin{itemize}
\item[(i)] (\ref{eq:pentGeq}) results in 15 independent conservation equations, one for each swap generator:
\ba
\sum_{i\in\text{Pent}_a,1\leq j\leq 15}r_i\,G_{ij}^{\text{Pent}_a,\text{Pent}_b}\,r_j=0,\q\q\q a<b,\q\q\q a,b=1,\ldots,6.\label{15cons}
\ea
All other combinations of the swap generators with the Bloch vector components of some pentagon lead via (\ref{eq:pentGeq}) to conservation equations which are either trivial or implied by (\ref{15cons}). (Appendix \ref{app_geneq})
\item[(ii)] Together with the six pentagon equalities (\ref{pentin}) these 15 conservation equations (\ref{15cons}) constitute 21 equations which define an invariant set under $\ct'_2$, i.e.\ for {\it any} Bloch vector $\vec{r}$ solving (\ref{pentin}, \ref{15cons}), $T(t)\cdot\vec{r}$ will again solve these 21 equations for all $T(t)\in\ct_2'$. In particular, writing $T(t)=\exp(t\,G)$ with $G$ in the lie algebra of $\ct_2'$, the pentagon equalities will be preserved to all orders in $t$ (recall that (\ref{eq:pentGeq}) was only the preservation condition to linear order in $t$). (Appendix \ref{app_geneq})
\item[(iii)] The complementarity inequalities (\ref{compstrong}) are preserved by $\ct'_2$ and all Bloch vectors $\vec{r}$ satisfying (\ref{pentin}, \ref{15cons}) also necessarily obey all complementarity inequalities. (Appendix \ref{app_comp})
\item[(iv)] $\ct_2'$ preserves the correlation structure of figure \ref{fig_corr} and, fixing the convention to only employ the XNOR for constructing multipartite questions from individuals, (\ref{pentin}, \ref{15cons}) implies unambiguously the correlation structure of figure \ref{fig_corr}. (Appendix \ref{app_corr})
\end{itemize}
Accordingly, $\ct'_2$ maps states satisfying principles \ref{lim}--\ref{pres}, all background assumptions and (\ref{pentin}, \ref{15cons}) to other such states. Principle \ref{time} {requires the existence of} {\it any} time evolution fulfilling these conditions such that we must indeed conclude $\ct'_2\subseteq\ct_2$.


But which group is $\ct_2'$? In (a) it was seen that the swap generators form the Lie algebra of $\su(4)\simeq\so(6)\simeq\mathfrak{psu}(4)$. $\SU(4)$ is a double cover of $\SO(6)$ which, in turn, is a double cover of $\rm{PSO}(6)\simeq\rm{PSU}(4)$. The exponentiation of the swap generators (\ref{swapgen1}, \ref{swapgen2}) can{\it not} result in a faithful representation of either $\SU(4)$ or $\SO(6)$ -- which feature a non-trivial centre --, because by Schur's lemma all centre elements read $c\cdot\mathds{1}$ with $c^{15}=1$ such that $c\equiv 1$ and the representation is centreless. The exponentiation will thus result in a faithful representation of $\rm{PSU}(4)$. Hence, $\rm{PSU}(4)\subseteq\ct_2\subset\SO(15)$. 

Can $\ct_2$ contain any additional transformations not contained in $\ct'_2$? Given that the 15 surviving swap generators (\ref{swapgen1}, \ref{swapgen2}) constitute a maximal set consistent with (\ref{eq:pentGeq}) and the correlation structure of figure \ref{fig_corr}, we must conclude that the answer is negative. In fact, in appendix \ref{app_max} it is further shown that $\rm{PSU}(4)$ is a {\it maximal subgroup}\footnote{A {maximal subgroup} $H$ of a group $G$ is a proper subgroup which is not contained in any other subgroup other than $H$ itself and the full group $G$.} of $\SO(15)$. Since $\ct_2$ must be a proper subgroup of $\SO(15)$, we conclude that 
\ba
\ct_2\simeq\rm{PSU}(4).\nn
\ea
This is the correct time evolution group for two qubits in quantum theory and, thanks to (b), we obtain it in the correct Bloch vector representation.\footnote{The adjoint action of $U\in\SU(4)$ in an evolution $\rho\mapsto U\,\rho\,U^\dag$ of a $4\times4$ density matrix is ignorant of the phase in $U$ and therefore yields a representation of $\rm{PSU}(4)$.}

It is interesting to note that the six generators (\ref{swapgen2}) of the information swaps between the pentagons which overlap in an individual question satisfy the commutator algebra of $\so(3)\oplus\so(3)$ and therefore generate the subgroup $\rm{PSU}(2)\times\rm{PSU}(2)\simeq\SO(3)\times\SO(3)$ of product unitaries corresponding to the Bloch sphere rotations of the two individual qubits. By contrast, the nine generators (\ref{swapgen1}) of the swaps between pentagons overlapping in a composite question generate the entangling unitaries in $\rm{PSU}(4)$ (see appendix \ref{app_swap}).

\subsubsection{State space reconstruction}\label{sec_statesn2}

Now that we have concluded that $\ct_2=\rm{PSU}(4)$ is the correct time evolution group, we are also in a position to determine $\Sigma_2$. The 21 equations (\ref{pentin}, \ref{15cons}) define a $\ct_2$-invariant set of Bloch vectors and every legal pure state must lie within it.
One may be worried that these 21 equations over-constrain the 15 components of the Bloch vector $\vec{r}$. However, the legitimate `product' states (\ref{legalstates}) satisfy all 21 equations and $\ct_2$ preserves these equations such that the set defined by (\ref{pentin}, \ref{15cons}) is clearly non-empty. In fact, in appendix \ref{app_product} it is shown, using the information distribution results of section \ref{sec_infodist}, that 
for {\it any} Bloch vector fulfilling (\ref{pentin}, \ref{15cons}) there exists a time evolution in $\ct_2$ which maps all information to the `product state' form $\alpha_{z_1}=\alpha_{z_2}=\alpha_{zz}=1$ \texttt{bit} and all other $\alpha_i=0$. This informational configuration has eight solutions in terms of the Bloch vector which can be divided into two mutually exclusive sets (all other $r_i=0$)
 \begin{eqnarray}
\cs_{\rm XNOR}:&& \ 1./2. \ r_{zz}=+1, \ r_{z_1}=\pm 1, r_{z_2}=\pm 1, \ \ \ 3./4. \ r_{zz}=-1, \ r_{z_1}=\pm 1, r_{z_2}=\mp 1, \nonumber\\
\cs_{\rm XOR}:&& \ 5./6. \ r_{zz}=-1, \ r_{z_1}=\mp 1, r_{z_2}=\mp 1, \ \ \ 7./8. \ r_{zz}=+1, \ r_{z_1}=\mp 1, r_{z_2}=\pm 1,\nn
\end{eqnarray}
the first of which is consistent with the XNOR conjunction $Q_{zz}=Q_{z_1}\leftrightarrow Q_{z_2}$, the second of which corresponds to the XOR connective $Q_{zz}=\neg(Q_{z_1}\leftrightarrow Q_{z_2})$. These are two perfectly consistent conventions for building up the composite questions (the information measure cannot distinguish between XNOR and XOR) \cite{Hoehn:2014uua}.

It can be easily verified that the four solutions in $\cs_{XNOR}$ are connected by elements of $\ct_2$, as are the four solutions in $\cs_{XNOR}$.\footnote{For instance, solutions 1 and 2  (or 5 and 6) are mapped to solutions 4 and 3 (or 8 and 7), respectively, by $T=\exp(\pi\,G^{\text{Pent}_3,\text{Pent}_5})$ or $T=\exp(\pi\,G^{\text{Pent}_1,\text{Pent}_5})$. Similarly, solutions 1 and 2  (or 5 and 6) are mapped to solutions 3 and 4 (or 7 and 8), respectively, by $T=\exp(\pi\,G^{\text{Pent}_4,\text{Pent}_6})$ or $T=\exp(\pi\,G^{\text{Pent}_2,\text{Pent}_6})$. (See appendix \ref{app_A} for the explicit representations of the swap generators and formulas for their exponentiation.)} However, the two sets of Bloch vectors generated by acting with $\ct_2$ on each of $\cs_{XNOR}$ and $\cs_{XOR}$ are {\it not} connected by time evolution since, using the time connectedness of each set, 
\ba
\ct_2(\cs_{XOR})&:=&\{T\cdot(-1)(\vec{\delta}_{z_1}+\vec{\delta}_{z_2}+\vec{\delta}_{z_1z_2})\,\big|\,T\in\rm\ct_2\}\nn\\
&=&-\{T\cdot(\vec{\delta}_{z_1}+\vec{\delta}_{z_2}+\vec{\delta}_{z_1z_2})\,\big|\,T\in\rm\ct_2\}=:-\ct_2(\cs_{XNOR})\nn
\ea
such that $\ct_2(\cs_{XOR})$ and $\ct_2(\cs_{XNOR})$ are related by a global multiplication with $-\mathds{1}_{15\times15}\notin\ct_2\subset\SO(15)$ which commutes with all elements in $\ct_2$. This corresponds precisely to a change of convention of building composite questions with XOR rather than XNOR.

In conclusion, the 21 equations (\ref{pentin}, \ref{15cons}) define exactly {\it two} isomorphic sets $\ct_2(\cs_{XOR})$ and $\ct_2(\cs_{XNOR})$ which are disconnected by time evolution, however, on each of which the time evolution group $\ct_2$ acts transitively.


It is well-known that, thanks to transitivity, $\ct_2\simeq\rm{PSU}(4)$ generates {\it all} two-qubit pure states of quantum theory by acting with all its elements on {\it any} legal pure state (in the Bloch or hermitian representation) \cite{Hardy:2001jk,masanes2011derivation,Masanes:2011kx}. The seed pure state $\vec{r}=\vec{\delta}_{z_1}+\vec{\delta}_{z_2}+\vec{\delta}_{z_1z_2}$ in $\cs_{XNOR}$, written in the basis defined by the informationally complete question set $\{Q_{x_1},\ldots,Q_{zz}\}$, coincides with the generalized Bloch vector representation of the two-qubit product state density matrix $\rho=1/4\,(\mathds{1}_{4\times4}+\sigma_z\otimes\mathds{1}+\mathds{1}\otimes\sigma_z+\sigma_z\otimes\sigma_z)$, written in the basis of the informationally complete Pauli operators $\mathds{1}\otimes\sigma_i,\sigma_j\otimes\mathds{1},\sigma_i\otimes\sigma_j$, $i,j=x,y,z$. We also recall from (b) in subsection \ref{sec_psu4} that the 15 swap generators (\ref{swapgen1}, \ref{swapgen2}), expressed in the question basis, coincide with the adjoint representation of the fundamental generators of the quantum time evolution group $\SU(4)$, written in the Pauli operator basis. It is thus clear that the orbit $\ct_2(\cs_{XNOR})$, expressed in the question basis, is {\it exactly} the set of two-qubit pure states of quantum theory, expressed in the Pauli operator basis.
\footnote{Indeed, it can be easily checked, using a computer algebra program, that {\it all} two-qubit pure states of quantum theory satisfy the 21 equations (\ref{pentin}, \ref{15cons}).}  Furthermore, since the seed states in $\cs_{XNOR}$ are legal pure states in $\Sigma_2$ and since the time evolution image of any legal state must again be legal, we conclude that $\ct_2(\cs_{XNOR})$ is fully contained in the set of pure states implied by the principles. Geometrically, this set of two-qubit pure states is $\ct_2(\cs_{XNOR})\simeq\mathbb{C}\mathbb{P}^3$ \cite{Bengtsson}, of which $\ct_2\simeq\rm{PSU}(4)$ is the isometry group.

%


Evidently, $\ct_2(\cs_{XOR})\simeq\mathbb{C}\mathbb{P}^3$ also defines a representation of the pure state space which is physically perfectly equivalent to $\ct_2(\cs_{XNOR})$. However, since it corresponds to the `XOR-convention' 
it leads to a correlation structure as in figure \ref{fig_corr}, except that the signs in all triangles would be flipped. 

Hence, adopting the convention, as we did so far, to build up composite questions from individuals solely by XNOR connectives, we conclude that the $N=2$ pure state space implied by the principles is precisely (one copy of) the pure state space for two qubits in quantum theory. \emph{Accordingly, upon fixing the XNOR convention, a Bloch vector $\vec{r}$ represents a pure two-qubit quantum state if and only if it satisfies the six pentagon equations (\ref{pentin}), which are ignorant of the correlation structure, and the 15 conservation equations (\ref{15cons}) which also encode the correlation structure (up to an overall XNOR vs.\ XOR ambiguity)}.\footnote{Clearly, the 21 equations cannot be fully independent. In fact, only nine of the 21 equations can be locally independent on $\mathbb{R}^{15}$ to produce a $15-9=6$-dimensional pure state space $\mathbb{C}\mathbb{P}^3$. It is not possible, however, thanks to pairwise independence of the questions in an informationally complete set, to globally parametrize the pure state space in terms of the probabilities (or Bloch vector components) of six fixed questions only.}

The pure states form the set of extremal Bloch vector length within the full state space $\Sigma_2$ which must be convex. Thus, clearly, the convex hull of the pure states is contained in $\Sigma_2$. But can there by any further legal extremal states? If there was another extremal state it could not be a state of maximal information and it could also not be a convex linear combination of pure states. In section \ref{sec_post}, we required that $O$ can prepare any extremal state in a single shot interrogation relative to the state of no information with questions from an informationally complete set -- and possibly a subsequent time evolution. However, it follows from our constraints on the state update rule in section \ref{sec_post} that any posterior state of a system of two qubits in such a single shot interrogation will be a quantum state\footnote{Any two questions in the informationally complete set are pairwise independent and either maximally complementary or maximally compatible. Given the two constraints of section \ref{sec_post} on the update rule ((1) questions are repeatable, and (2) independent compatible information is preserved), it is clear that any single shot interrogation on the prior state $\vec{r}=0$ with the questions of the informationally complete set will result in a posterior state $\vec{r}\,'$ with any component being one of $0,\pm1$. Any such posterior state must obey principle \ref{lim}, complementarity and the correlation structure in figure \ref{fig_corr} and thus has either precisely one or three components equal to $\pm1$ and the rest $0$. But any such state respecting the correlation structure corresponds to a quantum state. In particular, the $3$ \texttt{bit} states are legal pure states.} which is already contained in the convex hull of the pure states. Since the pure states are closed under all possible time evolutions, so is their convex hull. We thus conclude that there can be no further extremal states than the pure states. The Krein-Milman theorem \cite{Krein1940} states that a (compact) convex set is the closed convex hull of its extreme points. Hence, we find
\ba
\Sigma_2= \text{closed convex hull of } \mathbb{C}\mathbb{P}^3.\nn
\ea

$\Sigma_2$ contains the {\it state of no information}, $\vec{r}=0$, (e.g., multiply each of the four solutions in $\cs_{XNOR}$ with $\f{1}{4}$ and sum up) and indeed coincides with the set of unit trace density matrices over the two-qubit Hilbert space $\mathbb{C}^2\otimes\mathbb{C}^2$. From the fact that all pure states satisfy all complementarity inequalities (\ref{compstrong}) it follows that all convex mixtures of them will satisfy them too since the information measure (\ref{infomeasure}) satisfies $\alpha_i(\lambda\,\vec{r}_1+(1-\lambda)\,\vec{r}_2)<\max\{\alpha_i(\vec{r}_1),\alpha_i(\vec{r}_2)\}$ if $\lambda\in(0,1)$ and if the pure states $\vec{r}_1\neq\vec{r}_2$ are distinct \cite{Hoehn:2014uua}.

\subsection{$N>2$ qubits}\label{sec_n>2}

Principles \ref{lim}, \ref{unlim} and \ref{loc} imply that an informationally complete set for $N$ gbits contains $4^N-1$ questions $Q_{\mu_1\mu_2\cdots\mu_N}=Q_{\mu_z}\leftrightarrow Q_{\mu_2}\leftrightarrow\cdots\leftrightarrow Q_{\mu_N}$, $\mu_i=0,x,y,z$, where $Q_{0}=1$, such that the Bloch vector $\vec{r}$ is $(4^N-1)$-dimensional \cite{Hoehn:2014uua}. Pure state Bloch vectors have (squared) length $2^N-1$ \texttt{bits}, corresponding to having maximal information about $N$ mutually independent and compatible questions (principle \ref{lim}), as well as their (dependent) multipartite correlations.

\subsubsection{Derivation of the unitary group}\label{sec_psunl2}

Again, any given time evolution $T(t)$ acts linearly on the Bloch vector $r_i(t)=T_{ij}(t)\,r_j(0)$ and constitutes a one-parameter subgroup of $\ct_N$ \cite{Hoehn:2014uua}. For analogous reasons to the $N=2$ case, $\ct_N$ must be a proper subgroup of $\SO(4^N-1)$ for $N\geq2$.

We label the $N$ gbits by $1,\ldots,N$. Consider the gbit pair labeled by $(12)$. We shall say that this pair evolves as an isolated subsystem under $\ct^{(12)}_2=\rm{PSU}(4)$ (to avoid confusion, we label the copy of the two-gbit time evolution group by the pair of gbits) if the components of the $N$-gbit Bloch vector $\vec{r}\in\mathbb{R}^{4^N-1}$, 
\begin{description}
\item[$r_{\mu_1\mu_200\cdots0}$] corresponding to the 15 questions $Q_{\mu_1\mu_200\cdots0}$ (excluding $\mu_1=\mu_2=0$) forming an informationally complete set (see section \ref{sec_n2}) for the gbit pair $(12)$ evolve under $\ct^{(12)}_2$ as derived in section \ref{sec_psu4}, independently of the other components;\footnote{Note that these 15 Bloch vector components define an invariant subspace under $\ct^{(12)}_2$ of $\mathbb{R}^{4^N-1}$.} and
\item[$r_{00\mu_3\mu_4\cdots\mu_N}$] corresponding to all questions $Q_{00\mu_3\mu_4\cdots\mu_N}$ {\it not} involving gbits $(12)$ are left {\it invariant} under $\ct^{(12)}_2$. 
\end{description}
Recall that $\ct^{(12)}_2\supset\SO(3)\times\SO(3)$ contains the local qubit unitaries such that this definition also accounts for the isolated evolution of individual gbits.

Since $N$ gbits form a composite system, it must be physically possible for every pair of gbits to evolve in time together as an isolated subsystem, as derived in section \ref{sec_psu4}, and for any individual gbit to evolve isolated of the others, as described in  \cite{Hoehn:2014uua}, thus without affecting $O$'s information distribution over any other gbits. Accordingly, we shall require the time evolutions $\ct_2\simeq\rm{PSU}(4)$ for any pair of gbits and $\ct_1\simeq\SO(3)$ for any single gbit, respectively, to be contained in $\ct_N$. Of course, given three or more gbits, the different copies of $\rm{PSU}(4)$ cannot act simultaneously on all pairs due to monogamy of entanglement (which also naturally follows from the principles \cite{Hoehn:2014uua}). 

In appendix \ref{app_swapnl2}, it is shown that this requirement of isolated $\ct_2$- or $\ct_1$-evolution, together with the results of section \ref{sec_n2}, leads to an unambiguous promotion of the representation of the $\rm{PSU}(4)$ time evolution elements for every gbit pair from $\mathbb{R}^{15}$ to $\mathbb{R}^{4^N-1}$. In particular, the $\ct^{(12)}_2$-generators of the gbit pair $(12)$ take the form
\ba
G^{\text{Pent}^{(12)}_a,\text{Pent}_b^{(12)}}_{(\mu_1\mu_2\mu_3\mu_4\cdots\mu_N)(\nu_1\nu_2\nu_3\nu_4\cdots\nu_N)}=G^{\text{Pent}_a,\text{Pent}_b}_{(\mu_1\mu_2)(\nu_1\nu_2)}\,\delta_{\mu_3\nu_3}\delta_{\mu_4\nu_4}\cdots\delta_{\mu_N\nu_N},\label{psu4adj}
\ea
where $G^{\text{Pent}_a,\text{Pent}_b}_{(\mu_1\mu_2)(\nu_1\nu_2)}$ is the representation of the corresponding two-qubit swap generators on $\mathbb{R}^{15}$ of section \ref{sec_psu4} (and appendix \ref{app_swap})\footnote{In agreement with the more general notation of this section, we have exchanged the indices $i,j$ in $G^{\text{Pent}_a,\text{Pent}_b}_{ij}$ (\ref{swapgen1}, \ref{swapgen2}) for the equivalent $(\mu_1\mu_2)$ and $(\nu_1\nu_2)$ indices, respectively.} and we define $G^{\text{Pent}_a,\text{Pent}_b}_{(00)(\nu_1\nu_2)}:=0=:G^{\text{Pent}_a,\text{Pent}_b}_{(\mu_1\mu_2)(00)}$. In appendix \ref{app_QTswap}, it is furthermore shown that the generators (\ref{psu4adj}) coincide precisely with the adjoint representation of the fundamental generators $\sigma_i\otimes\mathds{1}\otimes\mathds{1}\otimes\cdots\mathds{1},\mathds{1}\otimes\sigma_i\otimes\mathds{1}\cdots\mathds{1},\sigma_i\otimes\sigma_j\otimes\mathds{1}\cdots\mathds{1}$ of pairwise unitaries in quantum theory. The ordering of coincidence is such that, firstly, $Q_{\mu_1\mu_20\cdots0}$ corresponds to $\sigma_{\mu_1\mu_20\cdots0}:=\sigma_{\mu_1}\otimes\sigma_{\mu_2}\otimes\mathds{1}\otimes\cdots\otimes\mathds{1}$ where $\sigma_0=\mathds{1}$ and, secondly, $G^{\text{Pent}_a^{(1,2)},\text{Pent}_b^{(1,2)}}$ coincides with the adjoint representation of $\sigma_{\mu_1\mu_20\cdots0}$ corresponding to the unique question $Q_{\mu_1\mu_20\cdots0}$ in $\text{Pent}_a^{(12)}\cap\text{Pent}_b^{(12)}$. For example, $G^{\text{Pent}_1^{(12)},\text{Pent}_2^{(12)}}$ coincides with the adjoint representation of $\sigma_x\otimes\sigma_x\otimes\mathds{1}\otimes\cdots\otimes\mathds{1}$. This coincidence holds analogously for arbitrary pairs among the $N$ gbits. Clearly, also the $\ct_N$ subgroups generated by these bipartite generators will have exactly the same form (at the Bloch vector level) as in quantum theory.

It is well-known that two-qubit unitaries $\rm{PSU}(4)$ (between any pair) and local evolutions $\SO(3)$ generate the full projective unitary group $\rm{PSU}(2^N)$ \cite{bremner2002practical,Harrow:2008aa}.\footnote{This universality result has also been used in other reconstructions \cite{masanes2011derivation,de2012deriving}.} Since all local evolutions and pairwise unitaries are required to be contained in $\ct_N$ and since these have the same representation as in quantum theory, we must conclude, abstractly, that $\rm{PSU}(2^N)\subseteq\ct_N\subset\rm{SO}(4^N-1)$ and, more explicitly, that the generated copy of $\rm{PSU}(2^N)$ appears in a Bloch vector representation, relative to the question basis, which is identical to the Bloch vector (or adjoint) representation of the quantum unitaries relative to the Pauli operator basis. As in the $N=2$ case, $\rm{PSU}(2^N)$ is a maximal subgroup of $\SO(4^N-1)$ (see appendix \ref{app_max}) and since $\ct_N$ must be a proper subgroup of the latter, we conclude
\ba
\ct_N\simeq\rm{PSU}(2^N)\nn
\ea
which is the correct time evolution group for $N$ qubits. The fact that we obtain the full group $\rm{PSU}(2^N)$ (rather than some of its subgroups) {follows from} the maximality requirement of principle \ref{time} which {demands} {\it every} time evolution compatible with the principles (and the background assumptions). As a consistency check, we show in appendix \ref{app_comppres} that $\rm{PSU}(2^N)$ indeed preserves all complementarity inequalities (\ref{compstrong}), as required.

\subsubsection{State space reconstruction}\label{sec_statesnl2}

We show in appendix \ref{app_prodnl2} that for every Bloch vector $\vec{r}$ which could be a legal $N$ gbit pure state there exists a time evolution in $\ct_N$ which transfers all $2^N-1$ \texttt{bits} to the `product state' form $\alpha_{z_1}=\cdots=\alpha_{z_N}=\alpha_{z_1z_2}=\alpha_{z_1z_3}=\cdots=\alpha_{z_1z_2z_3}=\cdots=\alpha_{z_1\cdots z_N}=1$ \texttt{bit} (and all other $\alpha_i=0$). This informational configuration has $2^{2^N-1}$ Bloch vector solutions $r_{z_1},\ldots,r_{z_1z_2},\ldots,r_{z_1\cdots z_N}\in\{-1,+1\}$ and the remaining $r_i=0$. Since by principle \ref{lim} only $N$ of the $2^N-1$ corresponding questions $Q_{z_1},\ldots,Q_{z_1\cdots z_N}$ are mutually independent, these Bloch vectors can be grouped into $2^{2^N-1}/2^N$ sets $\cs^{N}_1,\ldots,\cs^{N}_{2^{2^N-1-N}}$, each consistent with a specific convention of distributing XNOR or XOR connectives among the different individual gbit questions $Q_{\mu_1},\ldots,Q_{\mu_N}$ to build up multipartite questions -- in analogy to section \ref{sec_statesn2}. Evidently, only one of these sets agrees with our choice of employing solely the XNOR connective $\leftrightarrow$ to define multipartite questions $Q_{\mu_1\mu_2\cdots\mu_N}=Q_{\mu_1}\leftrightarrow Q_{\mu_2}\leftrightarrow\cdots\leftrightarrow Q_{\mu_N}$ from the individuals $Q_{\mu_i}$, namely the set of $2^N$ solutions defined by
\ba
\cs^{N}_{\rm XNOR}:=\Big\{(r_{z_1},\cdots,r_{z_N},r_{z_1z_2},\ldots,r_{z_1\cdots z_N})\,\Big|\,r_{z_1},\ldots,r_{z_N}\in\{-1,+1\},\nn\\
r_{z_{i_1}\cdots z_{i_m}}=\prod_{k=1}^{m\leq N} r_{z_{i_k}}, i_k\in\{1,\ldots,N\}, i_k<i_{k+1}\Big\}.\nn
\ea
It is not difficult to convince oneself that the $2^N$ Bloch vectors in any convention set $\cs^{N}_i$ are connected through the local rotations $\SO(3)\times\cdots\times\SO(3)\subset\ct_N$.\footnote{Namely, by the local unitaries which map $r_{z_i}=+1\longleftrightarrow r_{z_i}=-1$.}

We now focus on $\cs^N_{\rm XNOR}$. The state $\vec{r}_z:=\vec{\delta}_{z_1}+\cdots+\vec{\delta}_{z_1\cdots z_N}$ in $\cs^{N}_{\rm XNOR}$, 
coincides with the generalized Bloch vector representation of the $N$-qubit product state density matrix $\rho=(\mathds{1}_{2^N\times2^N}+\sigma_z\otimes\mathds{1}\otimes\cdots\otimes\mathds{1}+\cdots+\sigma_z\otimes\cdots\otimes\sigma_z)/2^N$ in quantum theory and is a legal pure state since $Q_{z_1},\ldots,Q_{z_N}$ are mutually compatible and independent \cite{Hoehn:2014uua}. 
It was shown in the previous section that the Bloch vector representation of $\ct_N$ is {\it exactly} the same as in quantum theory. 
As for $N=2$, $\ct_N$ acts transitively on the pure states of qubit quantum theory and therefore the {\it complete} pure quantum state space is generated when $\ct_N$ acts on any pure quantum state~\cite{Hardy:2001jk,masanes2011derivation,Masanes:2011kx}. 
Hence, the orbit $\ct_N(\cs^{N}_{\rm XNOR})$, expressed in our question basis, coincides exactly with the Bloch vector representation of the $N$-qubit pure state space of quantum theory, written in the Pauli operator basis. Since the time evolution image of any legal pure state must be again a legal pure state, we conclude that all of $\ct_N(\cs^{N}_{\rm XNOR})$ is contained in the set of pure states implied by the principles. 


But can there be more pure states? Since all other sets $\cs^N_i\neq\cs^N_{\rm XNOR}$ correspond to distinct conventions of building up composite questions from the individuals $Q_{\mu_i}$, the answer is negative. Indeed, the seed states in any $\cs^N_i\neq\cs^N_{\rm XNOR}$ are not legal quantum states, featuring a correlation structure distinct from quantum theory. (There are only $2^N$ pure quantum states with only $\pm1$ in the $z$-components and these precisely constitute $\cs^N_{\rm XNOR}$.) Hence, these sets are not connected via $\ct_N$ to our legal pure states $\ct_N(\cs^{N}_{\rm XNOR})$. Some of these other conventions will yield a distinct, but physically equivalent representation of the set of quantum pure states (e.g., as in the $N=2$ case the set corresponding to the convention of building up all composite questions with the XOR, rather than XNOR connective).



Consequently, adhering to our usual convention to build up all composite questions of an informationally complete set {\it only} with XNOR operations from the $Q_{\mu_i}$, implies that the set of all pure states allowed by the principles $\ct_N(\cs^{N}_{\rm XNOR})$ is precisely the set of pure quantum states. Geometrically, for $N$ qubits this space is given by $\ct_N(\cs^{N}_{\rm XNOR})\simeq\mathbb{C}\mathbb{P}^{2^N-1}$ \cite{Bengtsson} of which $\rm{PSU}(2^N)$ is the transitive isometry group. In complete analogy to the $N=2$ case in section \ref{sec_statesn2}, we thus obtain
\ba
\Sigma_N= \text{closed convex hull of } \mathbb{C}\mathbb{P}^{2^N-1}\nn
\ea
which contains the {\it state of no information} and coincides with the set of normalized density matrices over the $N$-fold tensor product of single qubit Hilbert spaces $\mathbb{C}^2\otimes\cdots\otimes\mathbb{C}^2$. For consistency, we show in appendix \ref{app_comppres} that all states in $\Sigma_N$ are compatible with the principles and, in particular, satisfy all complementarity inequalities (\ref{compstrong}).

In conclusion, we arrive precisely at the correct state spaces and time evolution groups for arbitrarily many qubits.

\subsection{The set of allowed questions $\cq_N$ and the Born rule}\label{sec_qn}

The reconstruction of the time evolution groups $\ct_N$ and the state spaces $\Sigma_N$ did not require the derivation of the precise structure of $\cq_N$, but only of the structure of an informationally complete set $\cq_{M_N}\subset\cq_N$. But what is the structure of the question set $\cq_N$? And what is the action of $\ct_N$ on $\cq_N$? To answer these questions, we invoke principle \ref{Q} which did not yet come into play.

\subsubsection{Characterization of the question set}\label{sec_Qreq}\label{sec_tran}

We begin by phrasing the derived probability rule (\ref{ansatz}) in terms of Bloch vectors, $y(Q|\vec{r}):=Y(Q|\vec{y}=\f{1}{2}\vec{r}+\vec{1})$; the probability for $Q=$ `yes', given the state $\vec{r}$, then reads
\ba
y(Q|\vec{r})=y(\vec{q}|\vec{r})=\f{1}{2}\left(1+\vec{q}\cdot\vec{r}\right).\label{born1}
\ea
The structure of the landscape in section \ref{sec_post} implies that to every $Q\in\cq_N$ there corresponds, via (\ref{born1}), a question vector $\vec{q}\in\mathbb{R}^{4^N-1}$ such that $y(\vec{q}|\vec{r})\in[0,1]$ $\forall\,\vec{r}\in\Sigma_N$ and a $1$ \texttt{bit} state $\vec{r}_Q$ such that $y(\vec{q}|\vec{r}_Q)=1$, i.e., such that $O$ `knows' that $Q=$ `yes' if $S$ is in the state $\vec{r}_Q$. Conversely, principle \ref{Q} asserts that each such vector $\vec{q}$ corresponds to a question $Q\in\cq_N$. In appendix \ref{app_B0}, we show that any such vector $\vec{q}$ is a $1$ \texttt{bit} quantum state, in fact, coinciding with $\vec{r}_Q=\vec{q}$. We thus arrive at the following question vector characterization:

\begin{cons}{\bf(Question vector characterization)} A vector $\vec{q}\in\mathbb{R}^{4^N-1}$ corresponds to $Q\in\cq_N$ if and only if it is a quantum state with $|\vec{q}|^2=1$ \texttt{bit} and $y(\vec{q}|\vec{r})\in[0,1]$ $\forall\,\vec{r}\in\Sigma_N$.
\end{cons}

Given that every question vector corresponds to a unique $Q\in\cq_N$ and vice versa (see section \ref{sec_post}), we immediately have
\ba
\cq_N\simeq\{\vec{q}\in\mathbb{R}^{4^N-1}\,|\, y(\vec{q}|\vec{r})\in[0,1]\,\forall\,\vec{r}\in\Sigma_N\q\text{and}\q\vec{q}\q\text{is a $1$ \texttt{bit} quantum state}\}\label{Qcharacter}
\ea

Among other things, operationally this means that every $Q\in\cq_N$ is in one-to-one correspondence with a {unique} $1$ \texttt{bit} state $\vec{r}_Q\in\Sigma_N$ which represents the truth value $Q=$ `yes' and which does not represent the truth value `yes' for any other question in $\cq_N$. This state $\vec{r}_Q$ encodes the situation that $O$ has asked {\it only} the single question $Q$ to $S$ in the state of no information $\vec{r}=\vec{0}$ and received a `yes' answer (i.e., $\vec{r}_Q$ is the updated state after receiving $Q=$`yes' relative to $\vec{r}=\vec{0}$). 
For $N=1$ each question in $\cq_N$ will therefore be described by the {\it pure} state in which it is answered with `yes', while for $N>1$ each question is represented by the {\it mixed} state in which it is answered with `yes' by $S$. We also note that $\neg Q\in\cq_N$ iff $Q\in\cq_N$ and that $\neg Q$ will be described by a distinct question vector. 




Thus, the full set of legitimate $1$ \texttt{bit} question vectors, corresponding to $\cq_N$, coincides with a subset of the $1$ \texttt{bit} quantum states in $\Sigma_N$. Firstly, notice that not every Bloch vector of length $1$ \texttt{bit} represents a legal state in $\Sigma_N$ for $N>1$. For instance, consider $N=2$ qubits and the vector $\vec{r}_{\rm ill}=\f{1}{\sqrt{2}}(1,1,0,\ldots,0)$ which naively could be interpreted as $O$ having half a \texttt{bit} of information about each of $Q_{x_1}$ and $Q_{x_2}$. But this would specify the probabilities that $O$ receives `yes' answers to the latter two questions as $y_{x_1}=y_{x_2}=(r_{x_1}+1)/2=(1+1/\sqrt{2})/2>0.85$. In this case it is impossible that the probability $y_{xx}$ that $Q_{xx}$ gives `yes' is $1/2$. Accordingly, $r_{xx}$ must be larger than $0$ and $\vec{r}_{\rm ill}$ is an illegal state. In fact, one can convince oneself that $\vec{r}_{\rm ill}$ is {\it not} a convex combination of pure states and that this Bloch vector would produce a non-positive density matrix.\footnote{It must hold $r_{xx}\geq\sqrt{2}-1$ in order for the state to be positive.} 
We conclude that, for $N>1$, not all vectors of length $1$ \texttt{bit} can correspond to questions in $\cq_N$.

Secondly, we proceed with the observation that also not every legal $1$ \texttt{bit} mixed state corresponds to a `yes' answer of a question in $\cq_N$. For example, for any pure state $\vec{r}_{\rm pure}$, the rescaling $\vec{r}_{\rm pure}/\sqrt{2^N-1}$ corresponds to a convex sum of the original pure state and the state of no information and thus yields a legal $1$ \texttt{bit} mixed state.\footnote{We note that such a state is {\it not} connected via time evolution to the $1$ \texttt{bit} states corresponding to the questions in an informationally complete set. For example, $\vec{r}_{\rm pure}/\sqrt{2^N-1}$ cannot be time-connected to $\vec{q}_{x_1}=\vec{\delta}_{x_1}$, corresponding to $Q_{x_1}$, for this would be equivalent to $\vec{r}_{\rm{pure}}$ being time-connected to $\sqrt{2^N-1}\,\vec{\delta}_{x_1}$ which is impossible for $N>1$. Thus, there are subsets of $1$ \texttt{bit} mixed states for $N>1$ which cannot be related via time evolution.} This state cannot correspond to a question vector of any $Q\in\cq_N$ because (\ref{born1}) implies that the probability for measuring a 'yes' outcome for $Q$ in the state $\vec{r}_{\rm pure}$ would be larger than one, $y(Q|\vec{r}_{\rm pure})=(1+\sqrt{2^N-1})/2>1$ for $N>1$.

\subsubsection{The Born rule for projective measurements}\label{sec_Born}

As an interlude, we note that (\ref{born1}) coincides precisely with the Born rule of quantum theory for projective measurements onto the Pauli operators $\vec{n}\cdot \vec{\sigma}$,\footnote{Pauli operators are those hermitian operators on $\mathbb{C}^{2^N}$ which have two eigenvalues $\pm1$ with equal dimensionality of the corresponding eigenspaces. These are exactly the hermitian, traceless operators $\sigma$ satisfying $\sigma^2=\mathds{1}$ (see section \ref{sec_tran} below, but also \cite{nielsen2010quantum,lawrence2002mutually}). As we shall see in appendix \ref{app_paulitime}, not every $\vec{n}\in\mathbb{R}^{4^N-1}$ with $|\vec{n}|=1$ yields a Pauli operator.} where $\vec{n}\in\mathbb{R}^{4^N-1}$ with $|\vec{n}|=1$.  
Namely, it can be easily checked that the projector onto the $+1$ eigenspace of a Pauli operator $\sigma_{\mu_1\cdots\mu_N}$ is given by $P_{\mu_1\cdots\mu_N}=\f{1}{2}(\mathds{1}+\sigma_{\mu_1\cdots\mu_N})$. Indeed, using $\sigma_{\mu_1\cdots\mu_N}^2=\mathds{1}$ it follows that $P_{\mu_1\cdots\mu_N}^2=P_{\mu_1\cdots\mu_N}$ and $P_{\mu_1\cdots\mu_N}\,\rho_{\mu_1\cdots\mu_N}=\rho_{\mu_1\cdots\mu_N}$ where $\rho_{\mu_1\cdots\mu_N}=\f{1}{2^N}(\mathds{1}+\sigma_{\mu_1\cdots\mu_N})$ is the density matrix corresponding to only $\sigma_{\mu_1\cdots\mu_N}$ being measured with $+1$ and all other $\sigma_{\nu_1\cdots\nu_N}$ unknown. Using that all Pauli operators are connected by unitary conjugation (see appendix \ref{app_paulitime}), one finds that $P_{\vec{n}}=\f{1}{2}(\mathds{1}+\vec{n}\cdot \vec{\sigma})$ constitutes the projector onto the $+1$ eigenspace of the Pauli operator $\vec{n}\cdot\vec{\sigma}$. But then for all permitted $\vec{n}$ and all density matrices we find
\ba
\tr(P_{\vec{n}}\,\rho)=\f{1}{2}\left(1+\vec{n}\cdot\vec{r}\right)\nn
\ea
in agreement with (\ref{born1}) under the identification $\vec{n}=\vec{q}$. 

We have thus reconstructed the Born rule of quantum theory for projective measurements onto Pauli operators. Next, we show that $\cq_N$ also coincides with the set of projective measurements onto the Pauli operators.

\subsubsection{Questions as projective measurements onto Pauli operators  }

For a single qubit, (\ref{Qcharacter}) immediately implies $\cq_1\simeq\{\vec{q}\in\mathbb{R}^3\,|\,\,|\vec{q}|=1\}\simeq\mathbb{CP}^1\simeq S^2$ such that $\cq_1$ is isomorphic to the set of pure states. 
This has two consequences. (1) It induces a transitive action of the time evolution group $\ct_1\simeq\SO(3)$ on $\cq_1$: if the Bloch vector $\vec{r}$ ($|\vec{r}|=1$) incarnates the `yes' answer to $Q$, represented by $\vec{q}$, then $T\cdot\vec{r}$ is the `yes' answer to the question $T(Q)$, represented by $T\cdot\vec{q}$, for any $T\in\ct_1$ (we can imagine the `time evolution' of a question to correspond to a rotation of the measurement device by means of which $O$ asks the questions). (2) $\cq_1$ is isomorphic to the set of projective measurements on single qubit Pauli operators, $\vec{n}\cdot\vec{\sigma}$, $\vec{\sigma}=(\sigma_x,\sigma_y,\sigma_z)$, which likewise are parametrized by $\vec{n}\in\mathbb{R}^3$, $|\vec{n}|=1$.

For $N>1$ the situation is more intricate. However, in appendix \ref{app_B} we derive the analogous results also for $N>1$. Firstly, on the quantum side, we show the following in appendix \ref{app_paulitime}:
\begin{itemize}
\item[(a)] The Pauli operators on an $N$-qubit Hilbert space $\mathbb{C}^{2^N}$ can be written as $\vec{n}\cdot\vec{\sigma}$, where 
\ba
\vec{\sigma}=(\sigma_{x_1}\otimes\mathds{1}\otimes\cdots\otimes\mathds{1}, \mathds{1}\otimes\sigma_{x_2}\otimes\cdots\otimes\mathds{1},\ldots,\sigma_{z_1}\otimes\sigma_{z_2}\otimes\cdots\otimes\sigma_{z_N})\label{pauli}
\ea
constitutes a basis of Pauli operators and the set of permissible unit vectors $\vec{n}$ is the orbit $\{T\cdot\vec{\delta}_{z_1}\,|\,T\in\ct_N\}\simeq\mathbb{C}\mathbb{P}^{2^N-1}$ which is thus isomorphic to the set of quantum pure states. (Note that this set of permissible $\vec{n}$ is a strict subset of the unit sphere for $N>1$.) In particular, $\ct_N=\rm{PSU}(2^N)$ acts transitively on the unit vectors $\vec{n}$ defining the Pauli operators. Equivalently, for any Pauli operator $\vec{n}\cdot\vec{\sigma}$ there exists $U\in\SU(2^N)$ such that $\sigma_{z_1}=U\,(\vec{n}\cdot\vec{\sigma})\,U^\dag$, where $\sigma_{z_1}:=\sigma_z\otimes\mathds{1}\times\cdots\otimes\mathds{1}$. The set of Pauli operators accounts for all traceless hermitian operators on $\mathbb{C}^{2^N}$ with $\pm1$ eigenvalues because all diagonal operators on $\mathbb{C}^{2^N}$ featuring equally many $\pm 1$ along their diagonals are related to $\sigma_{z_1}$ by conjugation with permutation matrices lying in $\SU(2^N)$. 
\end{itemize}
Second, on the reconstruction side, we establish in appendix \ref{app_qn} the below consequences of (\ref{Qcharacter}):
\begin{itemize}
\item[(b)] In its $1$ \texttt{bit} vector representation, the question set $\cq_N$ inherits an action of the time evolution group $\ct_N$ from the states $\Sigma_N$ and $\ct_N$ acts transitively on $\cq_N$. In particular, the basis question vectors $\vec{q}_{x_1}=\vec{\delta}_{x_1},\ldots,\vec{q}_{x_N}=\vec{\delta}_{x_N},\ldots,\vec{q}_{z_1\cdots z_N}=\vec{\delta}_{z_1\cdots z_N}$, corresponding to the informationally complete set $\cq_{M_N}=\{Q_{x_1},\ldots,Q_{x_N},\ldots,Q_{z_1\cdots z_N}\}$, are connected by time evolution and no question in $\cq_N$ exists whose question vector is not connected by time evolution to these basis vectors.

\item[(c)] Under the identification $\vec{q}\equiv\vec{n}$, $\cq_N$ is isomorphic to the set of Pauli operators on an $N$ qubit Hilbert space. Hence, $\cq_N\simeq\mathbb{CP}^{2^N-1}$ and the set of allowed questions is thanks to (a) therefore isomorphic to the pure state space also for $N>1$. 

\end{itemize}

We conclude that the set of binary questions $\cq_N$, which we have restricted $O$ to, corresponds to a strict subset of all possible $N$-qubit observables -- the Pauli operators. In fact, any $\vec{n}\in\mathbb{R}^{4^N-1}$ produces a hermitian operator $\vec{n}\cdot\vec{\sigma}$ on $\mathbb{C}^{2^N}$ and thus legitimate $N$-qubit observable. However, these operators can feature $2^N$ arbitrary real eigenvalues, corresponding to many different measurement outcomes per observable such that the latter cannot be represented by a single binary question. These observables are not captured by $\cq_N$.

The above results have strong implications for the question set. In particular, under the identification $\vec{n}\equiv\vec{q}$, we ultimately obtain the correspondence
\ba
Q_{\mu_1}\leftrightarrow Q_{\mu_2}\leftrightarrow\cdots\leftrightarrow Q_{\mu_N} \q\Leftrightarrow\q P_{\mu_1\cdots\mu_N}:=\f{1}{2}(\mathds{1}+\sigma_{\mu_1}\otimes\sigma_{\mu_2}\otimes\cdots\otimes\sigma_{\mu_N})\nn
\ea
where $P_{\mu_1\cdots{\mu_N}}$ is the projector onto the $+1$ eigenspace of $\sigma_{\mu_1}\otimes\sigma_{\mu_2}\otimes\cdots\otimes\sigma_{\mu_N}$. Indeed, $Q_{\mu_1\cdots\mu_N}$ yields $1$ or $0$ if an even or odd number of $Q_{\mu_i}$ is $0$, respectively, and thus corresponds to the question ``is the product of the spin projections of $\sigma_{\mu_1}\otimes\sigma_{\mu_2}\otimes\cdots\otimes\sigma_{\mu_N}$ $+1$?'' (see also \cite{lawrence2002mutually} for a related discussion of Pauli operators). We thus see that the XNOR connective $\leftrightarrow$ at the question level corresponds to the tensor product at the operator level. 

In the remainder of this section we shall discuss further consequences of (a)--(c).

\subsubsection{The dual time evolution of questions: Heisenberg vs.\ Schr\"odinger}

We just observed that the set of permissible questions $\cq_N$ inherits an action of the time evolution group $\ct_N$ from the states $\Sigma_N$. Specifically, any two legal question vectors are connected by a time evolution element and the time evolution of a legal question always yields another legal question.\footnote{For example, let $T_1\in\ct_1$ be a local rotation of qubit 1 and let $\tilde{T}\in\ct_1^{(1)}$ be its product representation within $\ct_N$. Denote by $T(Q)$ the action of some $T\in\ct_N$ on a question $Q\in\cq_N$ (understood at the question vector level). Since $T_1(Q_{\mu_1})$ is a legal question on qubit 1, so must be
\ba
T_1(Q_{\mu_1})\leftrightarrow Q_{\mu_2}\leftrightarrow\cdots\leftrightarrow Q_{\mu_N}=\tilde{T}(Q_{\mu_1}\leftrightarrow Q_{\mu_2}\leftrightarrow\cdots\leftrightarrow Q_{\mu_N})\nn.\ea
} At this point, this might be taken as just a mathematical observation. However, we might as well interpret the action of the time evolution group on the questions as transformations (e.g., rotations) of the measurement device(s) by means of which $O$ interrogates the systems.



The evolution of questions is dual to the evolution of states. Namely, the Born rule (\ref{born1}) implies $y(Q|T\cdot\vec{r})=(1+\vec{q}\cdot(T\cdot\vec{r}))/2=(1+(T^t\cdot\vec{q})\cdot\vec{r})/2=y(T^{-1}(Q)|\vec{r})$. That is, we may describe $O$'s interrogation of a system of $N$ qubits as it evolves in time in two equivalent ways: (1) the state vector $\vec{r}(t)$ evolves in time while the questions are time independent, or (2) the state vector is time independent and the questions $\vec{q}(t)$ evolve under the inverse of the time evolution. In particular, if both the state and question are evolved simultaneously, the probability remains invariant. (1) corresponds to the usual Schr\"odinger picture of quantum theory, while (2) parallels the Heisenberg picture; our reconstruction thus admits these dual interpretations of qubit quantum theory. 

Importantly, for the Heisenberg picture, the time evolution invariance of the Born rule (\ref{born1}) immediately implies that the compatibility and independence structure of the questions is invariant if time evolved simultaneously. Indeed, using that the question vectors are identical to $1$ \texttt{bit} states in which only the corresponding question is positively answered, we can express the independence relations of two arbitrary question $Q_1,Q_2\in\cq_N$ via $y(Q_1|\vec{q}_2)$ and clearly it holds
\ba
y(T(Q_1)|T\cdot\vec{q}_2)=y(Q_1|\vec{q}_2)\nn.
\ea
By similar arguments, using the Born rule with respect to states, it follows that also their compatibility relations remain invariant.

Finally, this also entails that every question $Q\in\cq_N$ is indeed contained in an informationally complete set, a mutually complementary set, and a maximal set of compatible questions. Namely, consider some set of mutually complementary $\{\vec{q}_{1},\vec{q}_{2}\,',\ldots,\vec{q}_{k}\, '\}$ and another of mutually compatible $\{\vec{q}_{1},\vec{q}_{2}\, '',\ldots,\vec{q}_{j}\,''\}$ questions. Since for any $Q\in\cq_N$ there is a $T\in\ct_N$ such that $\vec{q}=T\cdot\vec{q}_1$, the following time evolved sets $\{\vec{q},T\cdot\vec{q}_{i_2}\, ',\ldots,T\cdot\vec{q}_{i_k}\, '\}$ and  $\{\vec{q},T\cdot\vec{q}_{i_2}\, '',\ldots,T\cdot\vec{q}_{i_j}\, ''\}$ constitute a mutually complementary and a compatible set of questions, respectively, both of which contain $\vec{q}$. In the sense of compatibility and independence relations, no question in $\cq_N$ is special. 

\subsubsection{(Non-)uniqueness of pure state decompositions in terms of questions}\label{sec_decomp}


Every pure state can be decomposed in terms of a sum of $2^N-1$ mutually compatible question vectors. The reason is that, thanks to the transitivity of $\ct_N$ on the set of pure states, every pure state $\vec{r}_{\text{pure}}$ can be written as $\vec{r}_{\text{pure}}=T\cdot(\vec{\delta}_{z_1}+\vec{\delta}_{z_2}+\cdots+\vec{\delta}_{z_1\cdots z_N})$ for some $T\in\ct_N$. The vectors $T\cdot\vec{\delta}_{z_1},\ldots, T\cdot\vec{\delta}_{z_1\cdots z_N}$ within the decomposition are time connected to the question vectors $\vec{q}_{z_1},\ldots,\vec{q}_{z_1\cdots z_N}$ and are therefore themselves legal question vectors, featuring the same compatibility and independence relations. The Born rule (\ref{born1}) implies that the probability for each of these $2^N-1$ questions in the pure state decomposition to be answered by $S_N$ with `yes' equals one in this state. In fact, (b) above implies that, by running through all elements $T$ in $\ct_N$, all question vectors will appear in some pure state. This raises the question whether such a question decomposition of a pure state 
is unique or not and, in consequence, whether $S_N$, prepared in a pure state, answers a unique set of questions in $\cq_N$ with `yes'.

For $N=1$ this is trivially the case since every pure state vector is also a legal question vector. For $N>1$ the situation, however, turns out to be less trivial. More precisely, in appendix \ref{app_decomp}, we demonstrate the peculiar fact that
\begin{itemize}
\item \emph{The decomposition of a pure state vector $\vec{r}_{\text{pure}}=\vec{q}_1+\cdots+\vec{q}_{2^N-1}$ in terms of question vectors $\vec{q}_i$ for $Q_i\in\cq_N$ is \emph{unique} for $N=1,2$ and \emph{non-unique} for $N\geq3$.}
\end{itemize}
This is a consequence of the fact that the isotropy subgroup $\rm{PSU}(2^N-1)$ of $\ct_N=\rm{PSU}(2^N)$ on $\mathbb{C}\mathbb{P}^{2^N-1}$ corresponding to a pure state $\vec{r}_{\rm pure}$ contains elements for $N\geq3$ which are not part of the isotropy subgroups associated to every question vector $\vec{q}_i$ in the decomposition.

In other words, for $N=1,2$, $S_N$, prepared in any pure state, answers a unique set of $2^N-1$ questions from $\cq_N$ positively. For $N\geq3$, $S_N$ answers in every pure state multiple distinct sets of $2^N-1$ questions from $\cq_N$ simultaneously with `yes'. However, for $N\geq3$ the total information contained in one of these sets of $2^N-1$ questions is evidently equivalent to that carried by any other such set, even though a question in the first set might be (partially) independent from all questions in any other set.



\subsection{The von Neumann evolution equation}\label{sec_Neum}

For completeness, we discuss briefly how the von Neumann evolution equation of density matrices follows from the reconstruction. 

After having established coincidence between $\Sigma_N$ and the set of $N$-qubit density matrices, nothing stops us from passing from the Bloch vector representation of states to the equivalent hermitian representation in terms of density matrices on $\mathbb{C}^{2^N}$
\ba
\rho=\f{1}{2^N}\left(\mathds{1}_{2^N\times2^N}+\vec{r}\cdot\vec{\sigma}\right),\nn
\ea
where $\vec{r}$ is the Bloch vector and $\vec{\sigma}$ is given in (\ref{pauli}). We have seen (e.g., in appendix \ref{app_nl2}) that the linear evolution $\vec{r}(t)=T(t)\,\vec{r}(0)$ with $T(t)=e^{t\,G}\in\rm{PSU}(2^N)$ is equivalent to the adjoint action of $U(t)=e^{-i\,H\,t}\in\SU(2^N)$ on its Lie algebra
\ba
\rho(t)=U(t)\,\rho(0)\,U^\dag(t),\label{unitary}
\ea
for some hermitian operator $H$ on $\mathbb{C}^{2^N}$ \cite{Bengtsson}. In particular, using that $\Tr(\sigma_i\,\sigma_j)=2^N\,\delta_{ij}$ \cite{lawrence2002mutually}, $T_{ij}(t)=\f{1}{2^N}\,\Tr\left(\sigma_i\,U(t)\,\sigma_j\,U^\dag(t)\right)$. This yields a relation between time evolution generators $G\in\mathfrak{psu}(2^N)$ at the Bloch vector level and a `Hamiltonian' $H$ on a Hilbert space. But (\ref{unitary}) is equivalent to $\rho(t)$ satisfying the von Neumann evolution equation 
\ba
i\f{\p\,\rho}{\p t}=[H,\rho]\nn
\ea
which, in turn, is well-known to be equivalent to the Schr\"odinger equation for pure states.

\section{Discussion and conclusions}\label{sec_con}

We have shown that one can derive qubit quantum theory from transparent rules on an observer's acquisition of information about an observed system. These rules constitute a set of physical statements, equivalent to the usual textbook axioms, characterizing the quantum formalism. This manuscript, together with \cite{Hoehn:2014uua}, thereby offers a solution to a longstanding problem and completes related informational reconstruction ideas put forward in the context of Rovelli's {\it relational quantum mechanics} \cite{Rovelli:1995fv} and the Brukner-Zeilinger informational interpretation of quantum theory \cite{zeilinger1999foundational,Brukner:ys} for the case of qubit systems. (It also can be regarded as a completion of some ideas put forward by Spekkens in his epistemic toy model \cite{spekkens2007evidence} which, however, relies on ontic states.) One of the salient conclusions to be drawn from the present reconstruction is that it is sufficient to speak about the information that an observer has access to through measurement. This information is associated to the relation between the observer and the system, established through interaction; the state represents the observer's `catalogue of knowledge' about the system and it is not necessary to consider the notion of intrinsic state of the system. This highlights that quantum theory {\it may} be understood as an inference framework governing an observer's acquisition of information and pertaining to what the observer can say about Nature, rather than to how Nature `really' is (but clearly does not preclude other interpretations).

In addition, the reconstruction provides new structural insights into qubit quantum theory which were previously unnoticed. Specifically, we have derived new constraints on the distribution of information over the various questions in an informationally complete set (orthonormal basis of Pauli operators) of $N$ qubits. This employs the quadratic information measure derived from the principles in \cite{Hoehn:2014uua} and earlier proposed from a different perspective in \cite{Brukner:1999qf,Brukner:2001ve,Brukner:ys,brukner2009information}. Most importantly, we have shown for two qubits that the maximal mutually complementary question sets each carry precisely $1$ \texttt{bit} of information for pure states, constituting six {\it conserved informational charges} of time evolution for two qubits. These six equalities define the unitary group and, together with 15 conservation equalities, fully characterize the pure state space. This generalizes the single qubit case where a similar statement holds. While it was not necessary for the completion of the reconstruction, it is tempting to conjecture that this is a general property, namely that the unitary group and pure states are characterized by maximal mutually complementary sets carrying precisely $1$ \texttt{bit} of information for arbitrarily many qubits. We leave this as an open question.

These observations highlight information as a `charge of quantum theory' in the sense of providing the conserved quantities of the unitary groups. In analogy to charges in other areas of physics which can be transferred without loss among different carrier systems, these informational charges can be redistributed without loss among different questions in between measurements.

{Such conserved charges thus form part of the invariant structure that observers in distinct reference frames should agree on. As such, they might be useful, say, in a quantum communication protocol as in \cite{Hoehn:2014vua} which permits distinct observers, who have never met before but can communicate, to efficiently agree on their respective descriptions of quantum states. In this manner, one can derive the appropriate reference frame transformation group operationally from the structure of the communicated physical objects, rather than imposing it on the theory `by hand'. For instance, depending on the conditions on such a quantum communication protocol, one can show that either the rotation group $\SO(3)$ or the orthochronous Lorentz group $\rm O^+(3,1)$ constitutes the dictionary among distinct observer's quantum descriptions -- without presupposing any specific spacetime structure \cite{Hoehn:2014vua}.}

We have also derived the Born rule for projective measurements and shown that the time evolution of states implied by the principles is equivalent to the von Neumann evolution equation. While it was not necessary for us to derive the Born rule for state transition probabilities, this could presumably be accomplished by using arguments similar to the ones in \cite{Hardy:2001jk,Dakic:2009bh,masanes2011derivation}. We emphasize that it was also not necessary to fully specify (or derive) the precise state update rule in order to arrive at the structure of quantum theory. We shall similarly leave the full clarification of this update rule as an open matter.

The binary question framework in its present form is limited to reconstructing qubit (and rebit \cite{Hoehn:2014uua,hw2}) quantum theory and requires a generalization in order to be applicable to arbitrary $n$-level quantum systems. A treatment of {\it mechanical} systems may even necessitate an entirely novel approach.

\section*{Acknowledgements}
We are indebted to Markus M\"uller for numerous discussions, helpful comments and the proof of lemma \ref{lem_markus}. We would also like to thank Lucien Hardy and Rob Spekkens for many useful discussions and comments. PH, furthermore, thanks \v{C}aslav Brukner, Sylvain Carrozza, Borivoje Daki\'c and Carlo Rovelli for discussion. CW is grateful for the hospitality of Perimeter Institute during his four visits, where a large part of this work was done. PH similarly thanks the Demokritos Institute in Athens for hospitality during a visit and both authors also appreciate the hospitality of the Institute for Theoretical Physics of Universiteit Utrecht during three visits. Research at Perimeter Institute is supported by the Government of Canada through Industry Canada and by the Province of Ontario through the Ministry of Research and Innovation. The project leading to this publication has also received funding from the European Union's Horizon 2020 research and innovation programme under the Marie Sklodowska-Curie grant agreement No 657661 (awarded to PH). During the main part of this research project, CW worked at NCSR ``Demokritos" and was supported by the Research Funding Program ARISTEIA, HOCTools (co-financed by the European Union (European Social Fund ESF) and Greek national funds through the Operational Program "Education and Lifelong Learning" of the National Strategic Reference Framework (NSRF)).

\appendix

\section{Affine-linearity of the probability function}\label{app_Born}

In section \ref{sec_post}, we argued operationally that the probability function is at least convex linear
\ba
Y(Q|\lambda\,\vec{y}_1+(1-\lambda)\,\vec{y}_2)=\lambda\,Y(Q|\vec{y}_1)+(1-\lambda)\,Y(Q|\vec{y}_2), \q\q\q0\leq\lambda\leq1\nn.
\ea
This holds for all $Q\in\cq$ and $\vec{y}_i\in\Sigma$. More generally, this means that
\ba
Y(Q|\sum_i\lambda_i\,\vec{y}_i)=\sum_i\lambda_i\,Y(Q|\vec{y}_i), \q\q\q \lambda_i\geq0,\q\q\q \sum_i\lambda_i=1.\label{appborn1}
\ea
It will be convenient to parametrize states equivalently by $\vec{r}_i=2\,\vec{y}_i-\vec{1}$. Setting $y(Q|\vec{r}):= Y(Q|\vec{y}=1/2\,\vec{r}+\vec{1})$, it is clear that (\ref{appborn1}) if and only if
\ba
y(Q|\sum_i\lambda_i\,\vec{r}_i)=\sum_i\lambda_i\,y(Q|\vec{r}_i), \q\q\q \lambda_i\geq0,\q\q\q \sum_i\lambda_i=1\label{appborn2a}.
\ea

We shall now show that this implies that $y(Q|\vec{r})$ is affine-linear as claimed in (\ref{ansatz}). This result and its following proof are a twisted version of linearity results in Hardy's \cite{Hardy:2001jk} and Barrett's \cite{barrett2007information}. The biggest twist is that here we work with `normalized' states only and with $y(Q|\vec{0})\neq0$.

Setting one of the $\vec{r}_i$ in (\ref{appborn2a}) equal to the state of no information, $\vec{r}=\vec{0}$, it is easy to see that for $\vec{r}_i\in\Sigma$ also
\ba
y(Q|\sum_i\lambda_i\,\vec{r}_i)=\sum_i\lambda_i\,y(Q|\vec{r}_i)+(1-\sum_i\lambda_i)\,y(Q|\vec{0}), \q\q \lambda_i\geq0,\q\q \sum_i\lambda_i\leq1\label{appborn2}.
\ea
Suppose now that $\sum_i\lambda_i=\gamma<1$. It is then not difficult to convert (\ref{appborn2}) to
\ba
y(Q|\sum_i\tilde{\lambda}_i\,\vec{\tilde{r}}_i)=\sum_i\tilde{\lambda}_i\,y(Q|\vec{\tilde{r}}_i)+(1-\sum_i\tilde{\lambda}_i)\,y(Q|\vec{0}),\nn
\ea
where $\tilde{\lambda}_i=\lambda_i/\gamma^2$ and $\vec{\tilde{r}}_i=\gamma\,\vec{r}_i$. Since $\sum_i\tilde{\lambda}_i=\gamma^{-1}>1$ and $\vec{\tilde{r}}_i\in\Sigma$ if $\vec{r}_i\in\Sigma$ ($\vec{\tilde{r}}_i$ is a convex combination of $\vec{r}_i$ and $\vec{0}$), we have even more generally
\ba
y(Q|\sum_i\lambda_i\,\vec{r}_i)=\sum_i\lambda_i\,y(Q|\vec{r}_i)+(1-\sum_i\lambda_i)\,y(Q|\vec{0}), \q\q\q \lambda_i\geq0,\label{appborn3}
\ea
as long as $\vec{r}_i,\sum_i\lambda_i\,\vec{r}_i\in\Sigma$. A special case of this is
\ba
y(Q|\lambda\,\vec{r})=\lambda\,y(Q|\vec{r})+(1-\lambda)\,y(Q|\vec{0}),\q\q\q \lambda>0,\q\q\q\vec{r},\lambda\vec{r}\in\Sigma.\nn
\ea
For $\lambda\,\vec{r}\notin\Sigma$ this equation is {\it a priori} not defined. Since this does not correspond to a physical situation, we are free to demand that also
\ba
y(Q|\lambda\,\vec{r})=\lambda\,y(Q|\vec{r})+(1-\lambda)\,y(Q|\vec{0}),\q\q\q \forall\,\vec{r}\in\Sigma,\lambda\geq0.\label{appborn4}
\ea

Consider now the set $\Sigma_+:=\{\lambda\,\vec{r}\,|\,\forall\,\vec{r}\in\Sigma,\lambda\geq0\}$. As $\Sigma$ is convex, $\Sigma_+$ is a convex cone. Using (\ref{appborn3}, \ref{appborn4}), one easily shows that this implies 
\ba
y(Q|\sum_i\lambda_i\,\vec{r}_i)=\sum_i\lambda_i\,y(Q|\vec{r}_i)+(1-\sum_i\lambda_i)\,y(Q|\vec{0}), \q\q\q \forall\,\vec{r}\in\Sigma_+,\lambda_i\geq0.\label{appborn5}
\ea
Next, suppose 
\ba
\vec{r}=\sum_i t_i\,\vec{r}_i, \q\q\q\text{where}, \q\q\q \vec{r}, \vec{r}_i\in\Sigma_+,\q\q\q t_i\in\mathbb{R}.
\ea
Split the indices according to $i\in A_-$ if $t_i<0$ and $i\in A_+$ if $t_i\geq0$. Then we have
\ba
\vec{r}+\sum_{i\in A_-}|t_i|\,\vec{r}_i=\sum_{i\in A_+}t_i\,\vec{r}_i\nn
\ea
and each side of the equation is in $\Sigma_+$ such that (\ref{appborn5}) holds. Upon reorganization, this yields
\ba
y(Q|\sum_i\lambda_i\,\vec{r}_i)=\sum_i\lambda_i\,y(Q|\vec{r}_i)+(1-\sum_i\lambda_i)\,y(Q|\vec{0}), \q\q\q \forall\,\lambda_i\in\mathbb{R}.\label{appborn6}
\ea
This may be linearly extended to any vector that lies in the span of $\Sigma_+$; $y(Q|\vec{r})$ on the rest of $\mathbb{R}^D$ is arbitrary and may be freely chosen to be of the affine form (\ref{appborn6}) as well.

Finally, setting $f_Q(\vec{r}):=y(Q|\vec{r})-y(Q|\vec{0})$, (\ref{appborn6}) implies that $f_Q(\vec{r})$ is {\it linear}
\ba
f_Q(\sum_it_i\,\vec{r}_i)=\sum_it_i\,f_Q(\vec{r}_i),\q\q\q t_i\in\mathbb{R}\nn,
\ea
such that
\ba
y(Q|\vec{r})=\vec{f}_Q\cdot\vec{r}+y(Q|\vec{0})\nn
\ea
for some $\vec{f}_Q\in\mathbb{R}^D$. Remembering that by definition $y(Q|\vec{r}=\vec{0})=1/2$ for all $Q\in\cq$ in the state of no information $\vec{r}=\vec{0}$, and setting $\vec{f}_Q=1/2\,\vec{q}$, we immediately have
\ba
y(Q|\vec{r})=\f{1}{2}(\vec{q}\cdot\vec{r}+1)\label{appborn7},
\ea
where $\vec{q}\in\mathbb{R}^D$ is a vector depending on $Q\in\cq$. Hence, $Y(Q|\vec{y})=1/2(\vec{q}\cdot(2\vec{y}-\vec{1})+1)$.

\section{Reconstruction of the unitary group and state spaces}\label{app_A}

In order to present a flowing text in the main part of the paper, some proofs, derivations and other statements were left out. These are collected in this appendix.

\subsection{Maximal mutually complementary triangle sets for $N=2$ qubits}\label{triangle-app}

The maximal mutually complementary pentagon sets are important for the derivation of the time evolution group since their information contents constitute conserved charges under time evolution. In addition to the maximal pentagon sets, there are also maximal mutually complementary sets which contain three questions:
\begin{gather}
\text{Tri}_1=\{Q_{xx},Q_{xy},Q_{z_2}\}, \ \text{Tri}_2=\{Q_{xx},Q_{xz},Q_{y_2}\}, \ \text{Tri}_3=\{Q_{xx},Q_{yx},Q_{z_1}\}, \nonumber\\ \text{Tri}_4=\{Q_{xx},Q_{zx},Q_{y_1}\}, \text{Tri}_5=\{Q_{xy},Q_{xz},Q_{x_2}\}, \ \text{Tri}_6=\{Q_{xy},Q_{yy},Q_{z_1}\}, \nonumber\\ \text{Tri}_7=\{Q_{xy},Q_{zy},Q_{y_1}\}, \ \text{Tri}_8=\{Q_{xz},Q_{zz},Q_{y_1}\}, \ \text{Tri}_9=\{Q_{xz},Q_{yz},Q_{z_1}\}, \nonumber\\ \text{Tri}_{10}=\{Q_{yx},Q_{yy},Q_{z_2}\}, \ \text{Tri}_{11}=\{Q_{yx},Q_{yz},Q_{y_2}\}, \ \text{Tri}_{12}=\{Q_{yx},Q_{zx},Q_{x_1}\}, \nonumber\\
\text{Tri}_{13}=\{Q_{yy},Q_{yz},Q_{x_2}\}, \ \text{Tri}_{14}=\{Q_{yy},Q_{zy},Q_{x_1}\}, \ \text{Tri}_{15}=\{Q_{yz},Q_{zz},Q_{x_1}\}, \nonumber\\ 
\text{Tri}_{16}=\{Q_{zx},Q_{zy},Q_{z_2}\}, \ \text{Tri}_{17}=\{Q_{zx},Q_{zz},Q_{y_2}\}, \ \text{Tri}_{18}=\{Q_{zy},Q_{zz},Q_{x_2}\}, \nonumber\\
\text{Tri}_{19}=\{Q_{x_1},Q_{y_1},Q_{z_1}\}, \ \text{Tri}_{20}=\{Q_{x_2},Q_{y_2},Q_{z_2}\}. \label{eq:trineq}
\end{gather}
Similarly as for the pentagon sets, they can be represented by question graphs as given in (\ref{eq:triangles}). Again, the vertices correspond to individual questions and edges connecting them represent the corresponding correlation questions. However, contrary to the maximal pentagon sets, for pure states the total information carried by each triangle set is not necessarily equal to the $1$ \texttt{bit} bound in (\ref{compstrong}) (e.g., for entangled states the information content in $\text{Tri}_{19},\text{Tri}_{20}$ is $0$ \texttt{bits}) and is furthermore not conserved under time evolution. The pentagon and triangle sets are the only maximal mutually complementary sets for $N=2$ qubits.
\begin{align}\label{eq:triangles}
&\rbx{$\text{Tri}_1=$} \q \psfrag{x1}{\small $Q_{x_1}$}\psfrag{y1}{\small $Q_{y_1}$}\psfrag{z1}{\small $Q_{z_1}$}\psfrag{x2}{\small $Q_{x_2}$}\psfrag{y2}{\small $Q_{y_2}$}\psfrag{z2}{\small $Q_{z_2}$}\psfrag{xx}{\small $Q_{xx}$}\psfrag{xy}{\small $Q_{xy}$}\psfrag{xz}{\small $Q_{xz}$}\psfrag{yx}{\small $Q_{yx}$}\psfrag{yy}{\small $Q_{yy}$}\psfrag{yz}{\small $Q_{yz}$}\psfrag{zx}{\small $Q_{zx}$}\psfrag{zy}{\small $Q_{zy}$}\psfrag{zz}{\small $Q_{zz}$} \includegraphics[scale=0.4]{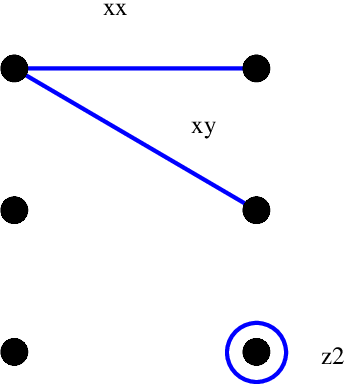} &
&\rbx{$\text{Tri}_2=$} \q \psfrag{x1}{\small $Q_{x_1}$}\psfrag{y1}{\small $Q_{y_1}$}\psfrag{z1}{\small $Q_{z_1}$}\psfrag{x2}{\small $Q_{x_2}$}\psfrag{y2}{\small $Q_{y_2}$}\psfrag{z2}{\small $Q_{z_2}$}\psfrag{xx}{\small $Q_{xx}$}\psfrag{xy}{\small $Q_{xy}$}\psfrag{xz}{\small $Q_{xz}$}\psfrag{yx}{\small $Q_{yx}$}\psfrag{yy}{\small $Q_{yy}$}\psfrag{yz}{\small $Q_{yz}$}\psfrag{zx}{\small $Q_{zx}$}\psfrag{zy}{\small $Q_{zy}$}\psfrag{zz}{\small $Q_{zz}$} \includegraphics[scale=0.4]{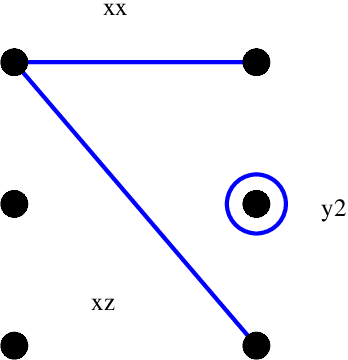} &
&\rbx{$\text{Tri}_3=$} \q \psfrag{x1}{\small $Q_{x_1}$}\psfrag{y1}{\small $Q_{y_1}$}\psfrag{z1}{\small $Q_{z_1}$}\psfrag{x2}{\small $Q_{x_2}$}\psfrag{y2}{\small $Q_{y_2}$}\psfrag{z2}{\small $Q_{z_2}$}\psfrag{xx}{\small $Q_{xx}$}\psfrag{xy}{\small $Q_{xy}$}\psfrag{xz}{\small $Q_{xz}$}\psfrag{yx}{\small $Q_{yx}$}\psfrag{yy}{\small $Q_{yy}$}\psfrag{yz}{\small $Q_{yz}$}\psfrag{zx}{\small $Q_{zx}$}\psfrag{zy}{\small $Q_{zy}$}\psfrag{zz}{\small $Q_{zz}$} \includegraphics[scale=0.4]{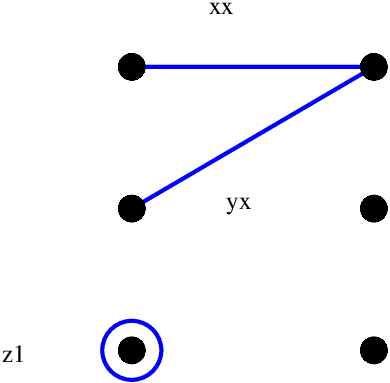} \nn\\\nn\\
&\rbx{$\text{Tri}_4=$} \q \psfrag{x1}{\small $Q_{x_1}$}\psfrag{y1}{\small $Q_{y_1}$}\psfrag{z1}{\small $Q_{z_1}$}\psfrag{x2}{\small $Q_{x_2}$}\psfrag{y2}{\small $Q_{y_2}$}\psfrag{z2}{\small $Q_{z_2}$}\psfrag{xx}{\small $Q_{xx}$}\psfrag{xy}{\small $Q_{xy}$}\psfrag{xz}{\small $Q_{xz}$}\psfrag{yx}{\small $Q_{yx}$}\psfrag{yy}{\small $Q_{yy}$}\psfrag{yz}{\small $Q_{yz}$}\psfrag{zx}{\small $Q_{zx}$}\psfrag{zy}{\small $Q_{zy}$}\psfrag{zz}{\small $Q_{zz}$} \includegraphics[scale=0.4]{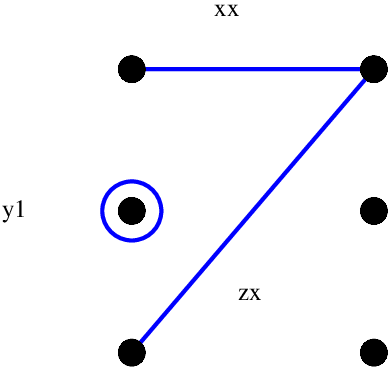} &
&\rbx{$\text{Tri}_5=$} \q \psfrag{x1}{\small $Q_{x_1}$}\psfrag{y1}{\small $Q_{y_1}$}\psfrag{z1}{\small $Q_{z_1}$}\psfrag{x2}{\small $Q_{x_2}$}\psfrag{y2}{\small $Q_{y_2}$}\psfrag{z2}{\small $Q_{z_2}$}\psfrag{xx}{\small $Q_{xx}$}\psfrag{xy}{\small $Q_{xy}$}\psfrag{xz}{\small $Q_{xz}$}\psfrag{yx}{\small $Q_{yx}$}\psfrag{yy}{\small $Q_{yy}$}\psfrag{yz}{\small $Q_{yz}$}\psfrag{zx}{\small $Q_{zx}$}\psfrag{zy}{\small $Q_{zy}$}\psfrag{zz}{\small $Q_{zz}$} \includegraphics[scale=0.4]{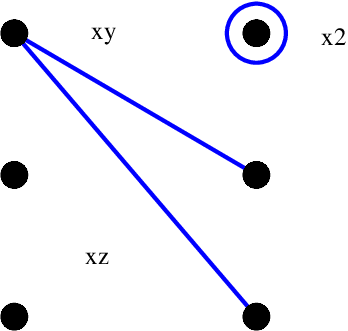} &
&\rbx{$\text{Tri}_6=$} \q \psfrag{x1}{\small $Q_{x_1}$}\psfrag{y1}{\small $Q_{y_1}$}\psfrag{z1}{\small $Q_{z_1}$}\psfrag{x2}{\small $Q_{x_2}$}\psfrag{y2}{\small $Q_{y_2}$}\psfrag{z2}{\small $Q_{z_2}$}\psfrag{xx}{\small $Q_{xx}$}\psfrag{xy}{\small $Q_{xy}$}\psfrag{xz}{\small $Q_{xz}$}\psfrag{yx}{\small $Q_{yx}$}\psfrag{yy}{\small $Q_{yy}$}\psfrag{yz}{\small $Q_{yz}$}\psfrag{zx}{\small $Q_{zx}$}\psfrag{zy}{\small $Q_{zy}$}\psfrag{zz}{\small $Q_{zz}$} \includegraphics[scale=0.4]{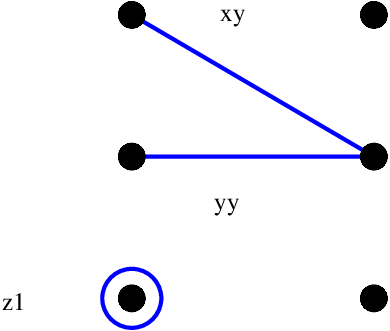} \nn\\\nn\\
&\rbx{$\text{Tri}_7=$} \q \psfrag{x1}{\small $Q_{x_1}$}\psfrag{y1}{\small $Q_{y_1}$}\psfrag{z1}{\small $Q_{z_1}$}\psfrag{x2}{\small $Q_{x_2}$}\psfrag{y2}{\small $Q_{y_2}$}\psfrag{z2}{\small $Q_{z_2}$}\psfrag{xx}{\small $Q_{xx}$}\psfrag{xy}{\small $Q_{xy}$}\psfrag{xz}{\small $Q_{xz}$}\psfrag{yx}{\small $Q_{yx}$}\psfrag{yy}{\small $Q_{yy}$}\psfrag{yz}{\small $Q_{yz}$}\psfrag{zx}{\small $Q_{zx}$}\psfrag{zy}{\small $Q_{zy}$}\psfrag{zz}{\small $Q_{zz}$} \includegraphics[scale=0.4]{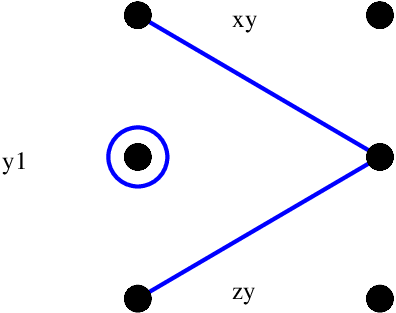} &
&\rbx{$\text{Tri}_8=$} \q \psfrag{x1}{\small $Q_{x_1}$}\psfrag{y1}{\small $Q_{y_1}$}\psfrag{z1}{\small $Q_{z_1}$}\psfrag{x2}{\small $Q_{x_2}$}\psfrag{y2}{\small $Q_{y_2}$}\psfrag{z2}{\small $Q_{z_2}$}\psfrag{xx}{\small $Q_{xx}$}\psfrag{xy}{\small $Q_{xy}$}\psfrag{xz}{\small $Q_{xz}$}\psfrag{yx}{\small $Q_{yx}$}\psfrag{yy}{\small $Q_{yy}$}\psfrag{yz}{\small $Q_{yz}$}\psfrag{zx}{\small $Q_{zx}$}\psfrag{zy}{\small $Q_{zy}$}\psfrag{zz}{\small $Q_{zz}$} \includegraphics[scale=0.4]{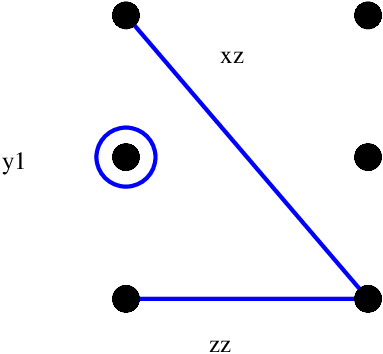} &
&\rbx{$\text{Tri}_9=$} \q \psfrag{x1}{\small $Q_{x_1}$}\psfrag{y1}{\small $Q_{y_1}$}\psfrag{z1}{\small $Q_{z_1}$}\psfrag{x2}{\small $Q_{x_2}$}\psfrag{y2}{\small $Q_{y_2}$}\psfrag{z2}{\small $Q_{z_2}$}\psfrag{xx}{\small $Q_{xx}$}\psfrag{xy}{\small $Q_{xy}$}\psfrag{xz}{\small $Q_{xz}$}\psfrag{yx}{\small $Q_{yx}$}\psfrag{yy}{\small $Q_{yy}$}\psfrag{yz}{\small $Q_{yz}$}\psfrag{zx}{\small $Q_{zx}$}\psfrag{zy}{\small $Q_{zy}$}\psfrag{zz}{\small $Q_{zz}$} \includegraphics[scale=0.4]{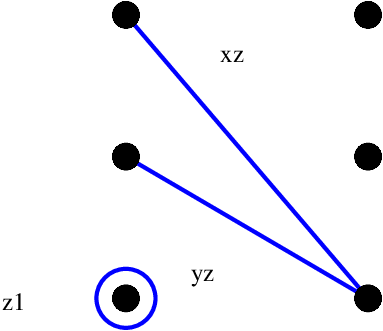} \nn\\\nn\\
&\rbx{$\text{Tri}_{10}=$} \q \psfrag{x1}{\small $Q_{x_1}$}\psfrag{y1}{\small $Q_{y_1}$}\psfrag{z1}{\small $Q_{z_1}$}\psfrag{x2}{\small $Q_{x_2}$}\psfrag{y2}{\small $Q_{y_2}$}\psfrag{z2}{\small $Q_{z_2}$}\psfrag{xx}{\small $Q_{xx}$}\psfrag{xy}{\small $Q_{xy}$}\psfrag{xz}{\small $Q_{xz}$}\psfrag{yx}{\small $Q_{yx}$}\psfrag{yy}{\small $Q_{yy}$}\psfrag{yz}{\small $Q_{yz}$}\psfrag{zx}{\small $Q_{zx}$}\psfrag{zy}{\small $Q_{zy}$}\psfrag{zz}{\small $Q_{zz}$} \includegraphics[scale=0.4]{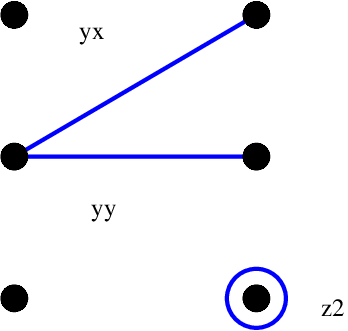} &
&\rbx{$\text{Tri}_{11}=$} \q \psfrag{x1}{\small $Q_{x_1}$}\psfrag{y1}{\small $Q_{y_1}$}\psfrag{z1}{\small $Q_{z_1}$}\psfrag{x2}{\small $Q_{x_2}$}\psfrag{y2}{\small $Q_{y_2}$}\psfrag{z2}{\small $Q_{z_2}$}\psfrag{xx}{\small $Q_{xx}$}\psfrag{xy}{\small $Q_{xy}$}\psfrag{xz}{\small $Q_{xz}$}\psfrag{yx}{\small $Q_{yx}$}\psfrag{yy}{\small $Q_{yy}$}\psfrag{yz}{\small $Q_{yz}$}\psfrag{zx}{\small $Q_{zx}$}\psfrag{zy}{\small $Q_{zy}$}\psfrag{zz}{\small $Q_{zz}$} \includegraphics[scale=0.4]{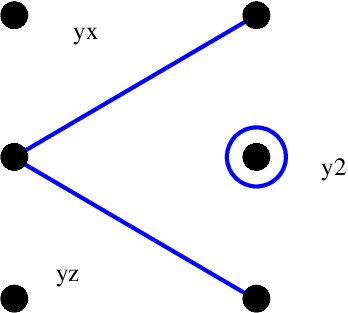} &
&\rbx{$\text{Tri}_{12}=$} \q \psfrag{x1}{\small $Q_{x_1}$}\psfrag{y1}{\small $Q_{y_1}$}\psfrag{z1}{\small $Q_{z_1}$}\psfrag{x2}{\small $Q_{x_2}$}\psfrag{y2}{\small $Q_{y_2}$}\psfrag{z2}{\small $Q_{z_2}$}\psfrag{xx}{\small $Q_{xx}$}\psfrag{xy}{\small $Q_{xy}$}\psfrag{xz}{\small $Q_{xz}$}\psfrag{yx}{\small $Q_{yx}$}\psfrag{yy}{\small $Q_{yy}$}\psfrag{yz}{\small $Q_{yz}$}\psfrag{zx}{\small $Q_{zx}$}\psfrag{zy}{\small $Q_{zy}$}\psfrag{zz}{\small $Q_{zz}$} \includegraphics[scale=0.4]{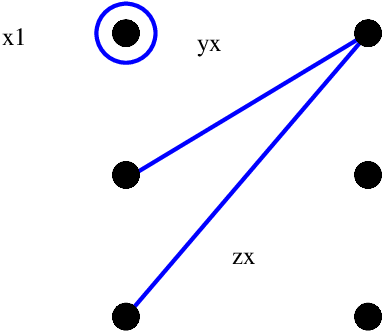} \nn\\\nn\\
&\rbx{$\text{Tri}_{13}=$} \q \psfrag{x1}{\small $Q_{x_1}$}\psfrag{y1}{\small $Q_{y_1}$}\psfrag{z1}{\small $Q_{z_1}$}\psfrag{x2}{\small $Q_{x_2}$}\psfrag{y2}{\small $Q_{y_2}$}\psfrag{z2}{\small $Q_{z_2}$}\psfrag{xx}{\small $Q_{xx}$}\psfrag{xy}{\small $Q_{xy}$}\psfrag{xz}{\small $Q_{xz}$}\psfrag{yx}{\small $Q_{yx}$}\psfrag{yy}{\small $Q_{yy}$}\psfrag{yz}{\small $Q_{yz}$}\psfrag{zx}{\small $Q_{zx}$}\psfrag{zy}{\small $Q_{zy}$}\psfrag{zz}{\small $Q_{zz}$} \includegraphics[scale=0.4]{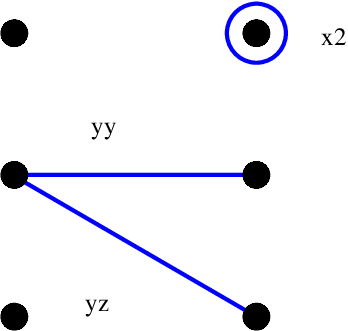} &
&\rbx{$\text{Tri}_{14}=$} \q \psfrag{x1}{\small $Q_{x_1}$}\psfrag{y1}{\small $Q_{y_1}$}\psfrag{z1}{\small $Q_{z_1}$}\psfrag{x2}{\small $Q_{x_2}$}\psfrag{y2}{\small $Q_{y_2}$}\psfrag{z2}{\small $Q_{z_2}$}\psfrag{xx}{\small $Q_{xx}$}\psfrag{xy}{\small $Q_{xy}$}\psfrag{xz}{\small $Q_{xz}$}\psfrag{yx}{\small $Q_{yx}$}\psfrag{yy}{\small $Q_{yy}$}\psfrag{yz}{\small $Q_{yz}$}\psfrag{zx}{\small $Q_{zx}$}\psfrag{zy}{\small $Q_{zy}$}\psfrag{zz}{\small $Q_{zz}$} \includegraphics[scale=0.4]{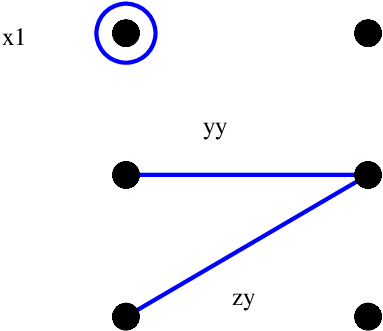} &
&\rbx{$\text{Tri}_{15}=$} \q \psfrag{x1}{\small $Q_{x_1}$}\psfrag{y1}{\small $Q_{y_1}$}\psfrag{z1}{\small $Q_{z_1}$}\psfrag{x2}{\small $Q_{x_2}$}\psfrag{y2}{\small $Q_{y_2}$}\psfrag{z2}{\small $Q_{z_2}$}\psfrag{xx}{\small $Q_{xx}$}\psfrag{xy}{\small $Q_{xy}$}\psfrag{xz}{\small $Q_{xz}$}\psfrag{yx}{\small $Q_{yx}$}\psfrag{yy}{\small $Q_{yy}$}\psfrag{yz}{\small $Q_{yz}$}\psfrag{zx}{\small $Q_{zx}$}\psfrag{zy}{\small $Q_{zy}$}\psfrag{zz}{\small $Q_{zz}$} \includegraphics[scale=0.4]{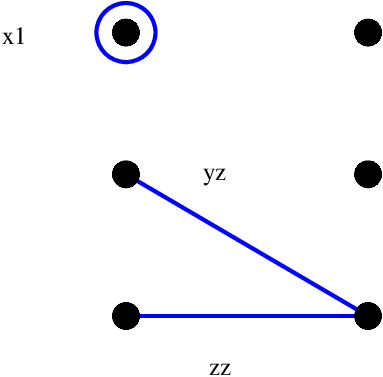} \nn\\\nn\\
&\rbx{$\text{Tri}_{16}=$} \q \psfrag{x1}{\small $Q_{x_1}$}\psfrag{y1}{\small $Q_{y_1}$}\psfrag{z1}{\small $Q_{z_1}$}\psfrag{x2}{\small $Q_{x_2}$}\psfrag{y2}{\small $Q_{y_2}$}\psfrag{z2}{\small $Q_{z_2}$}\psfrag{xx}{\small $Q_{xx}$}\psfrag{xy}{\small $Q_{xy}$}\psfrag{xz}{\small $Q_{xz}$}\psfrag{yx}{\small $Q_{yx}$}\psfrag{yy}{\small $Q_{yy}$}\psfrag{yz}{\small $Q_{yz}$}\psfrag{zx}{\small $Q_{zx}$}\psfrag{zy}{\small $Q_{zy}$}\psfrag{zz}{\small $Q_{zz}$} \includegraphics[scale=0.4]{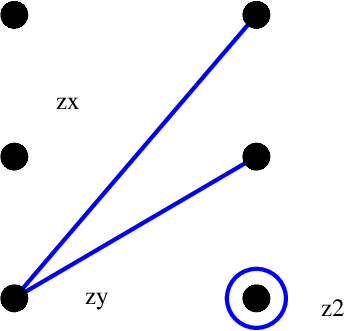} &
&\rbx{$\text{Tri}_{17}=$} \q \psfrag{x1}{\small $Q_{x_1}$}\psfrag{y1}{\small $Q_{y_1}$}\psfrag{z1}{\small $Q_{z_1}$}\psfrag{x2}{\small $Q_{x_2}$}\psfrag{y2}{\small $Q_{y_2}$}\psfrag{z2}{\small $Q_{z_2}$}\psfrag{xx}{\small $Q_{xx}$}\psfrag{xy}{\small $Q_{xy}$}\psfrag{xz}{\small $Q_{xz}$}\psfrag{yx}{\small $Q_{yx}$}\psfrag{yy}{\small $Q_{yy}$}\psfrag{yz}{\small $Q_{yz}$}\psfrag{zx}{\small $Q_{zx}$}\psfrag{zy}{\small $Q_{zy}$}\psfrag{zz}{\small $Q_{zz}$} \includegraphics[scale=0.4]{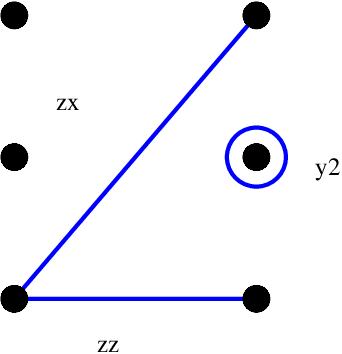} &
&\rbx{$\text{Tri}_{18}=$} \q \psfrag{x1}{\small $Q_{x_1}$}\psfrag{y1}{\small $Q_{y_1}$}\psfrag{z1}{\small $Q_{z_1}$}\psfrag{x2}{\small $Q_{x_2}$}\psfrag{y2}{\small $Q_{y_2}$}\psfrag{z2}{\small $Q_{z_2}$}\psfrag{xx}{\small $Q_{xx}$}\psfrag{xy}{\small $Q_{xy}$}\psfrag{xz}{\small $Q_{xz}$}\psfrag{yx}{\small $Q_{yx}$}\psfrag{yy}{\small $Q_{yy}$}\psfrag{yz}{\small $Q_{yz}$}\psfrag{zx}{\small $Q_{zx}$}\psfrag{zy}{\small $Q_{zy}$}\psfrag{zz}{\small $Q_{zz}$} \includegraphics[scale=0.4]{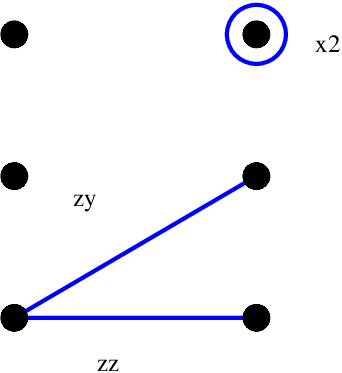} \nn\\\nn\\
&\rbx{$\text{Tri}_{19}=$} \q \psfrag{x1}{\small $Q_{x_1}$}\psfrag{y1}{\small $Q_{y_1}$}\psfrag{z1}{\small $Q_{z_1}$}\psfrag{x2}{\small $Q_{x_2}$}\psfrag{y2}{\small $Q_{y_2}$}\psfrag{z2}{\small $Q_{z_2}$}\psfrag{xx}{\small $Q_{xx}$}\psfrag{xy}{\small $Q_{xy}$}\psfrag{xz}{\small $Q_{xz}$}\psfrag{yx}{\small $Q_{yx}$}\psfrag{yy}{\small $Q_{yy}$}\psfrag{yz}{\small $Q_{yz}$}\psfrag{zx}{\small $Q_{zx}$}\psfrag{zy}{\small $Q_{zy}$}\psfrag{zz}{\small $Q_{zz}$} \includegraphics[scale=0.4]{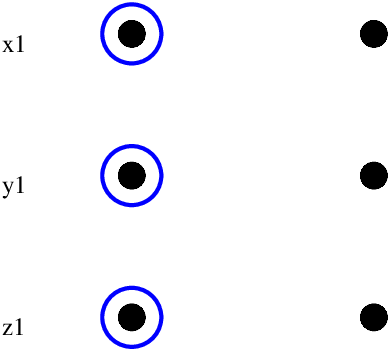} &
&\rbx{$\text{Tri}_{20}=$} \q \psfrag{x1}{\small $Q_{x_1}$}\psfrag{y1}{\small $Q_{y_1}$}\psfrag{z1}{\small $Q_{z_1}$}\psfrag{x2}{\small $Q_{x_2}$}\psfrag{y2}{\small $Q_{y_2}$}\psfrag{z2}{\small $Q_{z_2}$}\psfrag{xx}{\small $Q_{xx}$}\psfrag{xy}{\small $Q_{xy}$}\psfrag{xz}{\small $Q_{xz}$}\psfrag{yx}{\small $Q_{yx}$}\psfrag{yy}{\small $Q_{yy}$}\psfrag{yz}{\small $Q_{yz}$}\psfrag{zx}{\small $Q_{zx}$}\psfrag{zy}{\small $Q_{zy}$}\psfrag{zz}{\small $Q_{zz}$} \includegraphics[scale=0.4]{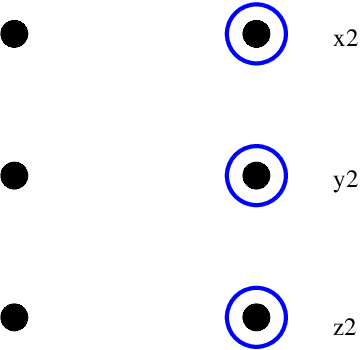} &
\end{align}
We note that Tri${}_2$, Tri${}_4$, Tri${}_{19}$ and Tri$_{20}$ are represented as green triangles in the pentagon lattice of figure \ref{fig:pentagon}.

These triangle sets define via (\ref{compstrong}) complementarity inequalities $0\leq I(\text{Tri}_i)\leq1$ \texttt{bit}, where $I(\text{Tri}_i)$ is the information contained in triangle set $i$. Together with the pentagon equalities (\ref{pentin}), these triangle complementarity inequalities define all independent complementarity inequalities which pure states have to satisfy as there are no other maximal mutually complementary sets. That is, any set of mutually complementary questions among the informationally complete set will be contained in either the pentagons or the triangles.

It is easy to show, however, that for pure states the 20 complementarity inequalities following from the triangle sets (\ref{eq:triangles}) are not all independent. In fact, the pentagon equalities (\ref{pentin}) imply that
\ba
I(\text{Tri}_1):=\alpha_{z_2}+\alpha_{xx}+\alpha_{xy}&=&\alpha_{x_1}+\alpha_{yz}+\alpha_{zz}=:I(\text{Tri}_{15}),\nn\\
I(\text{Tri}_2):=\alpha_{y_2}+\alpha_{xx}+\alpha_{xz}&=&\alpha_{x_1}+\alpha_{yy}+\alpha_{zy}=:I(\text{Tri}_{14}),\nn\\
I(\text{Tri}_3):=\alpha_{z_1}+\alpha_{xx}+\alpha_{yx}&=&\alpha_{x_2}+\alpha_{zy}+\alpha_{zz}=:I(\text{Tri}_{18}),\nn\\
I(\text{Tri}_4):=\alpha_{y_1}+\alpha_{xx}+\alpha_{zx}&=&\alpha_{x_2}+\alpha_{yy}+\alpha_{yz}=:I(\text{Tri}_{13}),\nn\\
I(\text{Tri}_5):=\alpha_{x_2}+\alpha_{xy}+\alpha_{xz}&=&\alpha_{x_1}+\alpha_{yx}+\alpha_{zx}=:I(\text{Tri}_{12}),\nn\\
I(\text{Tri}_6):=\alpha_{z_1}+\alpha_{yy}+\alpha_{xy}&=&\alpha_{y_2}+\alpha_{zx}+\alpha_{zz}=:I(\text{Tri}_{17}),\nn\\
I(\text{Tri}_7):=\alpha_{y_1}+\alpha_{zy}+\alpha_{xy}&=&\alpha_{y_2}+\alpha_{yx}+\alpha_{yz}=:I(\text{Tri}_{11}),\nn\\
I(\text{Tri}_8):=\alpha_{y_1}+\alpha_{xz}+\alpha_{zz}&=&\alpha_{z_2}+\alpha_{yx}+\alpha_{yy}=:I(\text{Tri}_{10}),\nn\\
I(\text{Tri}_9):=\alpha_{z_1}+\alpha_{yz}+\alpha_{xz}&=&\alpha_{z_2}+\alpha_{zx}+\alpha_{zy}=:I(\text{Tri}_{16}),\nn\\
I(\text{Tri}_{19}):=\alpha_{x_1}+\alpha_{y_1}+\alpha_{z_1}&=&\alpha_{x_2}+\alpha_{y_2}+\alpha_{z_2}=:I(\text{Tri}_{20}).\nn
\ea
Note the symmetry pattern of these relations in terms of the graphical representation of the triangle sets in (\ref{eq:triangles}); the encircled individual question of the triangle set on the left hand side is the vertex where the two correlation questions of the triangle set on the right hand side meet and vice versa.

\subsection{The swap generators for $N=2$ qubits}\label{app_swapmain}

In this section we discuss the swap generators defining the group $\ct_2=\rm{PSU}(4)$, their exponentiation, the pentagon preservation equations and the consistency conditions arising from the complementarity inequalities (\ref{compstrong}) and the correlation structure of figure \ref{fig_corr}.

\subsubsection{Derivation of the swap generators}\label{app_swap}

We shall present the derivation of the 15 swap generators of section \ref{sec_psu4} which are consistent with the correlation structure of figure \ref{fig_corr}. Subsequently, we shall argue that by varying the relative signs in these swap generators one accounts for all possible 60 linearly independent generators which could satisfy (\ref{eq:pentGeq}, \ref{gij}).

At the Bloch vector level, the swap transformation (\ref{eq:pentswapq}) between $\text{Pent}_1$ and $\text{Pent}_2$ is of the form:
\begin{equation}
r_{y_1}\longleftrightarrow \pm r_{zx} \ (\text{Pent}_5), \ r_{z_1}\longleftrightarrow \pm r_{yx} \ (\text{Pent}_3), \ r_{xy}\longleftrightarrow \pm r_{z_2} \ (\text{Pent}_4), \ r_{xz}\longleftrightarrow \pm r_{y_2} \ (\text{Pent}_6).\nn \label{eq:pentswapr}
\end{equation}
Writing the transformation as $\vec{r}\,'=T\,\vec{r}=\exp((\pi/2)\,G^{\text{Pent}_1,\text{Pent}_2})\,\vec{r}$, the corresponding generator is, without loss of generality, of the following form:
\begin{equation}
G_{ij}^{\text{Pent}_1,\text{Pent}_2}=\delta_{iy_1}\delta_{jzx}+s_1\, \delta_{iz_1}\delta_{jyx}+s_2\,\delta_{ixy}\delta_{jz_2}+s_3\,\delta_{ixz}\delta_{jy_2}-(i\longleftrightarrow j), \nn\label{eq:pentswapTold}
\end{equation}
where $s_1,s_2,s_3\in\{-1,+1\}$ are relative signs which must be determined. As can be easily checked, (\ref{eq:pentGeq}) is trivially satisfied for $i\in$ Pent${}_k$ with $k\neq1,2$ and $G_{ij}=G_{ij}^{\text{Pent}_1,\text{Pent}_2}$ thanks to symmetry/anti-symmetry. For both the two swapped pentagons $k=1,2$, on the other hand, the conservation equations (\ref{eq:pentGeq}) with $G^{\text{Pent}_1,\text{Pent}_2}$ are equivalent to
\ba
r_{y_1}\,r_{zx}+s_1\,r_{z_1}\,r_{yx}+s_2\,r_{xy}\,r_{z_2}+s_3\,r_{xz}\,r_{y_2}=0.\label{sign}
\ea

The sign structure of the generator $G_{ij}^{\text{Pent}_1,\text{Pent}_2}$ can be derived by considering three separate information distributions, all of which correspond to legal states:
\begin{eqnarray}
&&\text{Configuration 1:} \ \alpha_{x_2}=1\,\texttt{bit} \,\q\Rightarrow \,\q \alpha_{y_1}=\alpha_{yx}, \ \alpha_{z_1}=\alpha_{zx}, \ \alpha_{xy}=\alpha_{z_2}=\alpha_{xz}=\alpha_{y_2}=0, \nonumber\\
&&\text{Configuration 2:} \ \alpha_{x_1}=1\,\texttt{bit} \,\q\Rightarrow \,\q \alpha_{y_2}=\alpha_{xy}, \ \alpha_{z_2}=\alpha_{xz}, \ \alpha_{y_1}=\alpha_{zx}=\alpha_{z_1}=\alpha_{yx}=0, \nonumber\\
&&\text{Configuration 3:} \ \alpha_{zz}=1 \,\texttt{bit} \,\q\Rightarrow \,\q \alpha_{z_1}=\alpha_{z_2}, \ \alpha_{xy}=\alpha_{yx}, \ \alpha_{y_1}=\alpha_{zx}=\alpha_{xz}=\alpha_{y_2}=0. \nn\label{eq:pmconf}
\end{eqnarray}
On the right hand sides, we have made use of the constraints on the information distribution at the end of section \ref{sec_infodist} -- in particular, figure \ref{fig:1bitquestion} -- and complementarity (e.g., $Q_{x_2},Q_{xy}$ being complementary implies that $\alpha_{x_2}=1$ \texttt{bit} necessitates $\alpha_{xy}=0$, etc.). 

In configuration 1, (\ref{sign}) reduces to 
\ba
r_{y_1}\,r_{zx}+s_1\,r_{z_1}\,r_{yx}=0.\label{sign2}
\ea
The relevant correlation triangles are represented in figure \ref{fig:conf1}. Since both triangles represent {\it even} correlations, we have that $r_{x_2}=\pm1$ implies $r_{z_1}=\pm r_{zx}$, $r_{y_1}=\pm r_{yx}$ (in this sign order). Accordingly, (\ref{sign2}) requires $s_1=-1$ in order to be satisfied\footnote{There must be a state which has all the Bloch components $r_{z_1},r_{zx},r_{y_1}$ and $r_{yx}$ being non-zero. If this was not the case it would imply that whenever the observer $O$ knows the answer to $Q_{x_2}$ completely, $O$ would then also know the answer to either the pairs $Q_{z_1},Q_{zx}$ or $Q_{y_1},Q_{yx}$ completely as well. However this is not possible since then the question pairs $Q_{z_1},Q_{zx}$ or $Q_{y_1},Q_{yx}$ would be (partially) dependent on $Q_{x_2}$, which contradicts the fact that they are part of an informationally complete set.}.
\begin{figure}[h!]
\psfrag{+}{\small $+$}
\psfrag{A}{\small $Q_{x_2}$}
\psfrag{B}{\small $Q_{y_1}$}
\psfrag{C}{\small $Q_{yx}$}
\psfrag{D}{\small $Q_{z_1}$}
\psfrag{E}{\small $Q_{zx}$}
\psfrag{a}{\small $\alpha_{x_2}=1$ \texttt{bit}}
\psfrag{b}{\small $\alpha_{y_1}=\alpha_{yx}$}
\psfrag{c}{\small $\alpha_{z_1}=\alpha_{zx}$}
\begin{subfigure}[b]{.22\textwidth}
\centering
\includegraphics[scale=0.3]{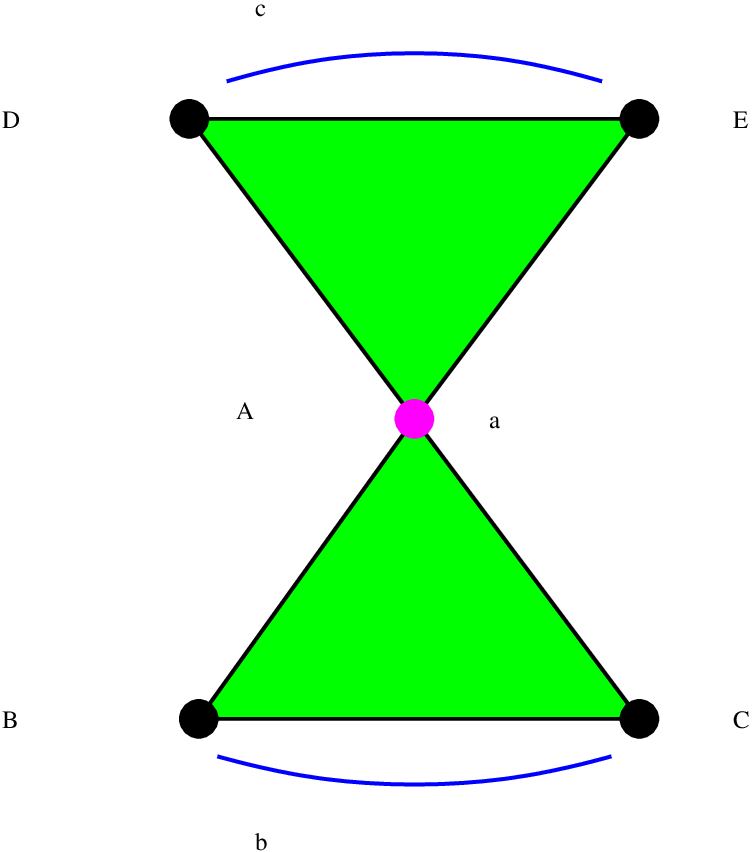}
\centering
\caption{\small }\label{fig:conf1}\end{subfigure}
\hspace*{1.8cm}
\begin{subfigure}[b]{.22\textwidth}
\psfrag{+}{\small $+$}
\psfrag{A}{\small $Q_{x_1}$}
\psfrag{B}{\small $Q_{z_2}$}
\psfrag{C}{\small $Q_{xz}$}
\psfrag{D}{\small $Q_{y_2}$}
\psfrag{E}{\small $Q_{xy}$}
\psfrag{a}{\small $\alpha_{x_1}=1$ \texttt{bit}}
\psfrag{b}{\small $\alpha_{z_2}=\alpha_{xz}$}
\psfrag{c}{\small $\alpha_{y_2}=\alpha_{xy}$}
\centering
\includegraphics[scale=0.3]{figures/conf1.eps}
\caption{\small }\label{fig:conf2}
\end{subfigure}
\hspace*{1.8cm}
\begin{subfigure}[b]{.22\textwidth}
\psfrag{+}{\small $+$}
\psfrag{A}{\small $Q_{zz}$}
\psfrag{B}{\small $Q_{xy}$}
\psfrag{C}{\small $Q_{yx}$}
\psfrag{D}{\small $Q_{z_1}$}
\psfrag{E}{\small $Q_{z_2}$}
\psfrag{a}{\small $\alpha_{zz}=1$ \texttt{bit}}
\psfrag{b}{\small $\alpha_{xy}=\alpha_{yx}$}
\psfrag{c}{\small $\alpha_{z_1}=\alpha_{z_2}$}
\centering
\includegraphics[scale=0.3]{figures/conf1.eps}
\caption{\small }\label{fig:conf3}
\end{subfigure}
\caption{\small The relevant correlation triangles of figure \ref{fig_corr} for (\ref{sign}) in configurations (a) 1; (b) 2; and (c) 3. } 
\end{figure}
Using configuration 2 and figure \ref{fig:conf2}, one shows similarly that $s_2\cdot s_3=-1$ and, finally, employing configuration 3 and figure \ref{fig:conf3}, one easily verifies that (\ref{sign}) requires $s_2=+1$ and, hence, $s_3=-1$. This yields $G^{\text{Pent}_1,\text{Pent}_2}$ in the form (\ref{eq:pentswapT}).

The generators of the eight other swaps between pentagons in figure \ref{fig:pentagon} sharing composite questions are derived similarly. As will become clear shortly, together with $G^{\text{Pent}_1,\text{Pent}_2}$ these constitute the
\newpage
\begin{gather}
\textbf{9 generators of entangling unitaries:}\nn\\
G_{ij}^{\text{Pent}_1,\text{Pent}_2}=\delta_{iy_1}\delta_{jzx}+\delta_{ixy}\delta_{jz_2}-\delta_{iz_1}\delta_{jyx}-\delta_{ixz}\delta_{jy_2}-(i\longleftrightarrow j), \nonumber\\
G_{ij}^{\text{Pent}_1,\text{Pent}_4}=\delta_{iz_1}\delta_{jyy}+\delta_{ixx}\delta_{jz_2}-\delta_{iy_1}\delta_{jzy}-\delta_{ixz}\delta_{jx_2}-(i\longleftrightarrow j), \nonumber\\
G_{ij}^{\text{Pent}_1,\text{Pent}_6}=\delta_{iy_1}\delta_{jzz}+\delta_{ixx}\delta_{jy_2}-\delta_{iz_1}\delta_{jyz}-\delta_{ixy}\delta_{jx_2}-(i\longleftrightarrow j), \nonumber\\
G_{ij}^{\text{Pent}_2,\text{Pent}_3}=\delta_{iy_2}\delta_{jyz}+\delta_{izx}\delta_{jx_1}-\delta_{iz_2}\delta_{jyy}-\delta_{ixx}\delta_{jz_1}-(i\longleftrightarrow j), \nonumber\\
G_{ij}^{\text{Pent}_2,\text{Pent}_5}=\delta_{iy_2}\delta_{jzz}+\delta_{ixx}\delta_{jy_1}-\delta_{iz_2}\delta_{jzy}-\delta_{iyx}\delta_{jx_1}-(i\longleftrightarrow j), \nonumber\\
G_{ij}^{\text{Pent}_3,\text{Pent}_4}=\delta_{iyz}\delta_{jx_2}+\delta_{iz_1}\delta_{jxy}-\delta_{iyx}\delta_{jz_2}-\delta_{ix_1}\delta_{jzy}-(i\longleftrightarrow j), \nonumber\\
G_{ij}^{\text{Pent}_3,\text{Pent}_6}=\delta_{iyy}\delta_{jx_2}+\delta_{ix_1}\delta_{jzz}-\delta_{iyx}\delta_{jy_2}-\delta_{iz_1}\delta_{jxz}-(i\longleftrightarrow j), \nonumber\\
G_{ij}^{\text{Pent}_4,\text{Pent}_5}=\delta_{ix_2}\delta_{jzz}+\delta_{iyy}\delta_{jx_1}-\delta_{iz_2}\delta_{jzx}-\delta_{ixy}\delta_{jy_1}-(i\longleftrightarrow j), \nonumber\\
G_{ij}^{\text{Pent}_5,\text{Pent}_6}=\delta_{izy}\delta_{jx_2}+\delta_{iy_1}\delta_{jxz}-\delta_{izx}\delta_{jy_2}-\delta_{ix_1}\delta_{jyz}-(i\longleftrightarrow j). \label{swapgen1}
\end{gather}
Note from the index structure that these generators always swap information between a pair of an individual and a composite question, thus transferring information from composite to individual questions and vice versa -- as appropriate for an entangling transformation.

Next, we shall briefly explain how to derive the specific form of the swap generators for pentagon pairs in figure \ref{fig:pentagon} overlapping in an individual question. For example, for the swap between Pent${}_3$ and Pent${}_5$, overlapping in $Q_{x_1}$, one arrives in analogy to above at
\ba
G^{\text{Pent}_3,\text{Pent}_5}_{ij}=\delta_{iy_1}\delta_{jz_1}+s_1'\,\delta_{iyx}\delta_{jzx}+s_2'\,\delta_{iyy}\delta_{jzy}+s_3'\,\delta_{iyz}\delta_{jzz}-(i\longleftrightarrow j)\nn
\ea
such that (\ref{eq:pentGeq}) for $k=3,5$ (again, the latter is trivially satisfied for $k\neq3,5$ and $G^{\text{Pent}_3,\text{Pent}_5}$) is equivalent to
\ba
r_{y_1}\,r_{z_1}+s_1'\,r_{yx}\,r_{zx}+s_2'\,r_{yy}\,r_{zy}+s_3'\,r_{yz}\,r_{zz}=0.\label{sign3}
\ea
The sign structure can be determined by considering the information distributions
\begin{eqnarray}
&&\text{Configuration 1':} \ \alpha_{x_2}=1\,\texttt{bit} \,\q\Rightarrow \,\q \alpha_{y_1}=\alpha_{yx}, \ \alpha_{z_1}=\alpha_{zx}, \ \alpha_{yy}=\alpha_{zy}=\alpha_{yz}=\alpha_{zz}=0, \nonumber\\
&&\text{Configuration 2':} \ \alpha_{y_2}=1\,\texttt{bit} \,\q\Rightarrow \,\q \alpha_{y_1}=\alpha_{yy}, \ \alpha_{z_1}=\alpha_{zy}, \ \alpha_{yx}=\alpha_{zx}=\alpha_{zz}=\alpha_{yz}=0, \nonumber\\
&&\text{Configuration 3':} \ \alpha_{z_2}=1 \,\texttt{bit} \,\q\Rightarrow \,\q \alpha_{z_1}=\alpha_{zz}, \ \alpha_{y_1}=\alpha_{yz}, \ \alpha_{yx}=\alpha_{zx}=\alpha_{yy}=\alpha_{zy}=0. \nn
\end{eqnarray}
The relevant correlation triangles for configurations 1'--3' are represented in figures \ref{fig:conf1p}--\ref{fig:conf3p}.
\begin{figure}[h!]
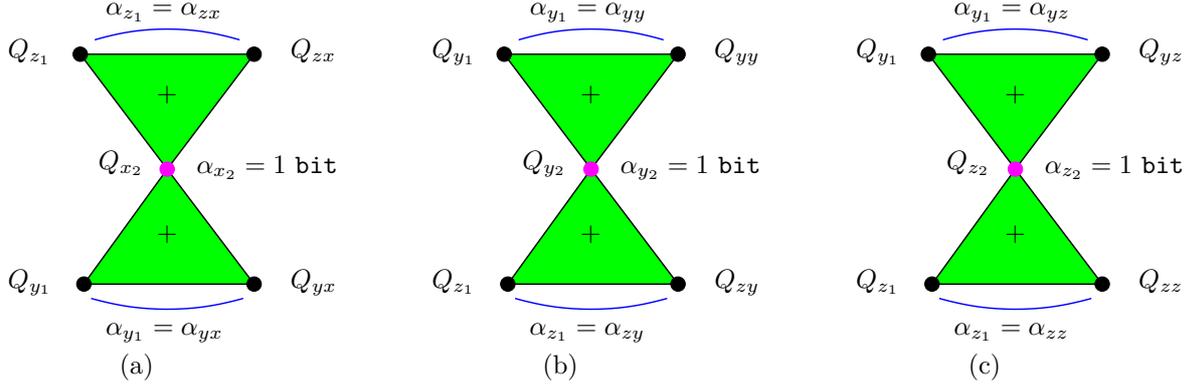

\psfrag{+}{\small $+$}
\psfrag{A}{\small $Q_{x_2}$}
\psfrag{B}{\small $Q_{y_1}$}
\psfrag{C}{\small $Q_{yx}$}
\psfrag{D}{\small $Q_{z_1}$}
\psfrag{E}{\small $Q_{zx}$}
\psfrag{a}{\small $\alpha_{x_2}=1$ \texttt{bit}}
\psfrag{b}{\small $\alpha_{y_1}=\alpha_{yx}$}
\psfrag{c}{\small $\alpha_{z_1}=\alpha_{zx}$}
\begin{subfigure}[b]{.22\textwidth}
\centering
\includegraphics[scale=0.3]{figures/conf1.eps}
\centering
\caption{\small }\label{fig:conf1p}\end{subfigure}
\hspace*{1.8cm}
\begin{subfigure}[b]{.22\textwidth}
\psfrag{+}{\small $+$}
\psfrag{A}{\small $Q_{y_2}$}
\psfrag{B}{\small $Q_{z_1}$}
\psfrag{C}{\small $Q_{zy}$}
\psfrag{D}{\small $Q_{y_1}$}
\psfrag{E}{\small $Q_{yy}$}
\psfrag{a}{\small $\alpha_{y_2}=1$ \texttt{bit}}
\psfrag{b}{\small $\alpha_{z_1}=\alpha_{zy}$}
\psfrag{c}{\small $\alpha_{y_1}=\alpha_{yy}$}
\centering
\includegraphics[scale=0.3]{figures/conf1.eps}
\caption{\small }\label{fig:conf2p}
\end{subfigure}
\hspace*{1.8cm}
\begin{subfigure}[b]{.22\textwidth}
\psfrag{+}{\small $+$}
\psfrag{A}{\small $Q_{z_2}$}
\psfrag{B}{\small $Q_{z_1}$}
\psfrag{C}{\small $Q_{zz}$}
\psfrag{D}{\small $Q_{y_1}$}
\psfrag{E}{\small $Q_{yz}$}
\psfrag{a}{\small $\alpha_{z_2}=1$ \texttt{bit}}
\psfrag{b}{\small $\alpha_{z_1}=\alpha_{zz}$}
\psfrag{c}{\small $\alpha_{y_1}=\alpha_{yz}$}
\centering
\includegraphics[scale=0.3]{figures/conf1.eps}
\caption{\small }\label{fig:conf3p}
\end{subfigure}
\caption{\small The relevant correlation triangles of figure \ref{fig_corr} for (\ref{sign3}) in configurations (a) 1'; (b) 2'; and (c) 3'. } 
\end{figure}
Now one proceeds as before, using that all relevant triangles represent {\it even} correlations, to show that $s_1'=s_2'=s_3'=-1$. The different sign structure (three, compared to the two minus signs for the entangling swaps) results from the fact that for {\it all} configurations 1'--3' the sign is determined by relating the last three terms in (\ref{sign3}) to the first (signless) term $r_{y_1}r_{z_1}$. By contrast, e.g., configuration 2 for the swap between Pent${}_1$ and Pent${}_2$ relates the last two terms with signs in (\ref{sign}) against each other.

This yields $G^{\text{Pent}_3,\text{Pent}_5}$ in the form (\ref{eq:TSO3x1}) and, in analogy, the full set of swap generators for pentagon pairs overlapping in an individual. As will be discussed below, these are the 
\begin{gather}
\textbf{6 generators of product unitaries:}\nn\\
G_{ij}^{\text{Pent}_1,\text{Pent}_3}=\delta_{ix_1}\delta_{jy_1}-\delta_{iyz}\delta_{jxz}-\delta_{iyy}\delta_{jxy}-\delta_{iyx}\delta_{jxx}-(i\longleftrightarrow j), \nonumber\\
G_{ij}^{\text{Pent}_1,\text{Pent}_5}=\delta_{iz_1}\delta_{jx_1}-\delta_{ixz}\delta_{jzz}-\delta_{ixy}\delta_{jzy}-\delta_{ixx}\delta_{jzx}-(i\longleftrightarrow j), \nonumber\\
G_{ij}^{\text{Pent}_2,\text{Pent}_4}=\delta_{iy_2}\delta_{jx_2}-\delta_{izx}\delta_{jzy}-\delta_{iyx}\delta_{jyy}-\delta_{ixx}\delta_{jxy}-(i\longleftrightarrow j), \nonumber\\
G_{ij}^{\text{Pent}_2,\text{Pent}_6}=\delta_{iz_2}\delta_{jx_2}-\delta_{izx}\delta_{jzz}-\delta_{iyx}\delta_{jyz}-\delta_{ixx}\delta_{jxz}-(i\longleftrightarrow j), \nonumber\\
G_{ij}^{\text{Pent}_3,\text{Pent}_5}=\delta_{iz_1}\delta_{jy_1}-\delta_{iyz}\delta_{jzz}-\delta_{iyy}\delta_{jzy}-\delta_{iyx}\delta_{jzx}-(i\longleftrightarrow j), \nonumber\\
G_{ij}^{\text{Pent}_4,\text{Pent}_6}=\delta_{iz_2}\delta_{jy_2}-\delta_{izy}\delta_{jzz}-\delta_{iyy}\delta_{jyz}-\delta_{ixy}\delta_{jxz}-(i\longleftrightarrow j). \label{swapgen2}
\end{gather}
It can be easily checked that the six generators in (\ref{swapgen2}) satisfy the commutator algebra of $\so(3)\oplus\so(3)$. Note from the index structure that these six generators always swap information between pairs of individual questions or pairs of composite questions -- as appropriate for the generators of the product unitaries.

With a computer algebra program one may check that, remarkably, the 15 generators (\ref{swapgen1}, \ref{swapgen2}) coincide exactly (in some cases up to an unimportant overall sign) with the adjoint representation of the 15 fundamental generators of the Lie group $\SU(4)$
\ba
(G^i)_{jk}:=f^{ijk}=\f{1}{4}\,\tr([\sigma_j,\sigma_k]\,\sigma_i), \label{eq:comalg}
\ea
where $f^{ijk}$ are the structure constants of $\SU(4)$, the indices $i,j,k$ take the 15 values $x_1,y_1,z_1,x_2,\ldots,xz,xy,\ldots,zz$ (as in our reconstruction) and $\sigma_{x_1}:=\sigma_x\otimes\mathds{1}$, ..., $\sigma_{x_2}:=\mathds{1}\otimes\sigma_x$, ..., $\sigma_{xx}:=\sigma_x\otimes\sigma_x$, ..., $\sigma_{zz}:=\sigma_z\otimes\sigma_z$ and $\sigma_x,\sigma_y,\sigma_z$ are the usual Pauli matrices. In particular, the ordering of coincidence is $G^i\equiv\pm G^{\text{Pent}_a,\text{Pent}_b}$ where $Q_i$ is the single question in $\text{Pent}_a\cap\text{Pent}_b$ which is left invariant by the swap; e.g., $G^{xx}\equiv G^{\text{Pent}_1,\text{Pent}_2}$, etc. This ultimately also clarifies that indeed (\ref{swapgen1}) constitute the generators of entangling unitaries, while (\ref{swapgen2}) are the generators of the product unitaries. Clearly, the 15 swap generators thus satisfy the commutator algebra of $\su(4)\simeq\so(6)\simeq\mathfrak{psu}(4)$
\begin{equation}
[G^i,G^j]=f^{ijk}\,G^k, \nn
\end{equation}
with $f^{ijk}$ given by (\ref{eq:comalg}).

Let us now explain how the full information swaps account for all 60 linearly independent generators which {\it could} solve (\ref{eq:pentGeq}, \ref{gij}). While deriving the 15 generators (\ref{swapgen1}, \ref{swapgen2}) we have made use of the correlation structure in figure \ref{fig_corr} in order to fix the relative signs in the generators, e.g.\ in (\ref{sign}, \ref{sign3}). It is clear, however, that by varying these relative signs, one can produce four linearly independent generators from each generator in (\ref{swapgen1}, \ref{swapgen2}) since each such generator contains four linearly independent components. By inspection, the reader may also verify that each of the 60 distinct pairs of complementary questions is encoded in precisely one of the 15 generators in terms of a non-vanishing component, corresponding to the pair of indices representing the pair of complementary questions. This immediately entails that by varying the relative signs in the 15 generators (\ref{swapgen1}, \ref{swapgen2}), one obtains precisely the maximal amount of 60 linearly independent generators which satisfy (\ref{gij}). The relative sign structure only affects the correlation structure but not the fact that each of these 60 linearly independent generators represents a full swap of information between a pair of pentagons sets. As evident from the derivation in this section, however, it is only the 15 generators in (\ref{swapgen1}, \ref{swapgen2}) which are consistent with the correlation structure in figure \ref{fig_corr} and which thus are legitimate candidates for time evolution generators in our reconstruction.

As an aside, let us briefly note that the sign structure of the 15 generators would be exactly the same, had we instead followed the alternative convention to build composite questions with the XOR connective, e.g., $\tilde{Q}_{xx}:=\neg(Q_{x_1}\leftrightarrow Q_{x_2})$, etc., rather than the XNOR as done thus far (see also \cite{Hoehn:2014uua} on this). $\tilde{Q}_{xx}$ represents an {\it anti-correlation} question ``are the answers to $Q_{x_1}$ and $Q_{x_2}$ anti-correlated?". In this case, the correlation structure for the XOR composites would coincide with the one in figure \ref{fig_corr} except that all even correlation triangles would be replaced by odd ones and vice versa. However, this would leave the relative sign structure, determined via figures \ref{fig:conf1}--\ref{fig:conf3p} invariant. This has to be expected, of course, since both conventions are physically equivalent. 

However, we note that the 15 generators for mirror quantum theory \cite{Hoehn:2014uua,Dakic:2009bh,Masanes:2011kx}, obtained by swapping the assignment `yes' $\leftrightarrow$ `no' of a single individual question and adhering to the convention of building composite questions with the XNOR as before, would be distinct. Indeed, the swap of the answer assignment for, say, $Q_{x_1}$ is equivalent to $Q_{x_1}\mapsto\neg Q_{x_1}$ (a partial transpose at the density matrix level) and $Q_{xx},Q_{xy},Q_{xz}\mapsto\neg Q_{xx},\neg Q_{xy},\neg Q_{xz}$. This produces a flip of the sign of the correlation triangles in {\it only} the upper graph in figure \ref{fig_corr} (involving {\it only} the correlation questions), while leaving the lower graph invariant (see \cite{Hoehn:2014uua} for details). This has the consequence that figures \ref{fig:conf1p}--\ref{fig:conf3p} and, more generally, the six product generators in (\ref{swapgen2}) remain invariant. However, the nine generators (\ref{swapgen1}) of the entangling transformations change their sign structure. In particular, figure \ref{fig:conf3} involves an {\it even} correlation triangle of the composites $Q_{xy},Q_{yx},Q_{zz}$, which would be replaced by an {\it odd} triangle for mirror quantum theory. This would result in $s_2=-1$ and thus $s_3=+1$ for mirror quantum theory (and analogously for the other generators). Mirror quantum theory thus has distinct entangling Hamiltonians (corresponding to the partial transpose relating it to standard quantum theory). Nevertheless, mirror quantum theory is physically perfectly equivalent to standard quantum theory and just employs a distinct convention for `yes/no' outcomes of questions \cite{Hoehn:2014uua,Masanes:2011kx}. The example of mirror quantum theory thus demonstrates that those swap generators among the 60 linearly independent ones mentioned above which differ in their relative sign structure from (\ref{swapgen1}, \ref{swapgen2}) simply correspond to distinct conventions.

\subsubsection{Exponentiation of the generators}

For the reasons mentioned in section \ref{sec_psu4}, the exponentiation of the swap generators results in the connected, simple Lie group $\ct_2'=\PSU(4)\simeq \PSO(6)$ (rather than in $\SU(4)$ or $\SO(6)$). 

The exponential of any single generator $G^a$ acts as $2\times 2$-rotation matrices on the planes spanned by each pair of swapped questions. For example, $T^{\text{Pent}_1,\text{Pent}_2}(t)=\exp(t\, G^{\text{Pent}_1,\text{Pent}_2})$ acts as rotations of angle $\pm t$ in the planes $(r_{y_1},r_{zx})$, $(r_{z_1},r_{yx})$, $(r_{xy},r_{z_2})$, $(r_{xz},r_{y_2})$, where the signs are fixed by the swap generator (\ref{eq:pentswapT}). Furthermore, as one can easily convince oneself, the six generators (\ref{swapgen2}) exponentiate to the $\SO(3)\times\SO(3)\simeq\PSU(2)\times\PSU(2)$ product unitaries. For instance, $T^{\text{Pent}_3,\text{Pent}_5}(t)=\exp(t\, G^{\text{Pent}_3,\text{Pent}_5})$, generated by the swap which leaves $\alpha_{x_1}$ invariant, describes rotations of the Bloch vector around the $r_{x_1}$ axis (leaving $r_{x_1},r_{xx},r_{xy},r_{xz}$ and $r_{x_2},r_{y_2},r_{z_2}$ invariant). Similarly, the other generators in (\ref{swapgen2}) generate rotations around the Bloch vector axis corresponding to the individual question which constitute the overlap of the respective pentagon pairs.

In general, the exponential of any single swap generator $G^{\text{Pent}_i,\text{Pent}_j}$ is of the form:
\begin{gather}
T^{\text{Pent}_a,\text{Pent}_b}(t)=\exp(t\, G^{\text{Pent}_a,\text{Pent}_b})=(\cos(t)-1)\tilde{\mathds{1}}^{\text{Pent}_a,\text{Pent}_b}+\sin(t)G^{\text{Pent}_a,\text{Pent}_b}+\mathds{1}, \label{eq:pentexpT} \\
\tilde{\mathds{1}}^{\text{Pent}_a,\text{Pent}_b}_{kl}=\sum_{\bar{k},\bar{l}\in (\text{Pent}_a\cup \text{Pent}_b)\setminus(\text{Pent}_a\cap \text{Pent}_b)}\delta_{k\bar{k}}\delta_{l\bar{l}}.\nonumber
\end{gather}
The matrix $\tilde{\mathds{1}}^{\text{Pent}_a,\text{Pent}_b}_{kl}\sim (G^{\text{Pent}_a,\text{Pent}_b})^2_{kl}$ is the diagonal matrix with ones at the positions of the eight questions which are swapped by $G^{\text{Pent}_a,\text{Pent}_b}$ and otherwise zeros.

\subsubsection{Pentagon conservation equations for $N=2$ qubits}\label{app_geneq}

Every swap generator $G=G^{\text{Pent}_a,\text{Pent}_b}$ puts constraints on the potential pure states according to Eq.\ (\ref{eq:pentGeq}). These equalities follow from the requirement that for pure states the total information in each pentagon set $\text{Pent}_c$ must be a conserved charge under the time evolution group $\ct_2$. Eq.\ (\ref{eq:pentGeq}) defines the pentagon conservation equations to linear order in $t$. However, clearly, if the group $\ct_2'=\rm{PSU}(4)$ generated by (\ref{swapgen1}, \ref{swapgen2}) did constitute the correct time evolution group $\ct_2$, then acting with an arbitrary $T(t)\in\ct'_2$ on a legal pure state state $\vec{r}$ must produce another pure state $\vec{r}\,':=T(t)\cdot \vec{r}$ which satisfies the pentagon equalities (\ref{pentin}) to all orders in $t$. In this section we shall show that the linear order conservation conditions (\ref{eq:pentGeq}) are, in fact, sufficient to guarantee preservation of the pentagon equalities to all orders in $t$ and for all $T\in\ct_2$.

To this end, we firstly consider the action of the exponential map (\ref{eq:pentexpT}) for an arbitrary of the 15 generators on some pure state $\vec{r}$. Surely, $\vec{r}\,'=T^{\text{Pent}_a,\text{Pent}_b}(t) \, \vec{r}$ must again satisfy the pentagon equalities (time evolution preserves the total information and must map states to states), i.e.\ we must have
\begin{gather}
1=\sum_{l\in\text{Pent}_c} r_l^2\overset{!}{=}\sum_{l\in\text{Pent}_c} (r'_l)^2,\q\q\q c=1,\ldots,6.\label{pentpres}
\end{gather}
We shall now show in lemma \ref{le:N2checkA} below that, if the first equation in (\ref{pentpres}) is satisfied, then
\begin{gather}
\sum_{l\in\text{Pent}_c} (r'_l)^2=\sum_{l\in\text{Pent}_c} r_l^2+2\sin(t)\cos(t)\hspace{-0.8 cm}\sum_{l\in\text{Pent}_c,1\leq m\leq 15}r_l \ G^{\text{Pent}_a,\text{Pent}_b}_{lm} \ r_m \nn
\end{gather}
such that (\ref{pentpres}) is satisfied for arbitrary $t$ iff the linear constraints (\ref{eq:pentGeq}) holds
\begin{gather}
 \sum_{l\in\text{Pent}_c,1\leq m\leq 15}r_l \ G^{\text{Pent}_a,\text{Pent}_b}_{lm} \ r_m=0. \label{pentrTreq}
\end{gather}
For this purpose we introduce the projector $P^{\text{Pent}_a}_{kl}:=\sum_{\bar{k},\bar{l}\in \text{Pent}_a}\delta_{k\bar{k}}\delta_{l\bar{l}}$ onto the Bloch vector components corresponding to $\text{Pent}_a$ and the symmetric matrix $R_{kl}:=(\vec{r}\cdot \vec{r}\,{}^t)_{kl}=r_kr_l$. In the following we choose the short-hand notation $P^a:=P^{\text{Pent}_a},G^{ab}:=G^{\text{Pent}_a,\text{Pent}_b}$ with $a<b$. The pentagon equalities and generator constraints (\ref{pentrTreq}) can now be equivalently expressed as
\begin{gather}
\text{Pentagon eq.}: \text{tr}[P^aR]=1, \text{for all} \ 1\leq a\leq 6, \label{pure state}\\
\text{Generator eq.}: \text{tr}[P^aG^{bc}R]=\frac{1}{2}\text{tr}[[P^a,G^{bc}]R]=0, \text{for all} \ 1\leq a\leq 6 \ \text{and} \ 1\leq b< c\leq 6. \nonumber
\end{gather}
Before we show the above mentioned result, we firstly require a few identities. Using the explicit expressions (\ref{swapgen1}, \ref{swapgen2}) for the 15 generators, one can check that the following statements are valid for any $1\leq a,b,c,d\leq 6$
\begin{enumerate}[(a)]
\item $[P^a,G^{bc}]=(\delta_{ab}+\delta_{ac})G^{bc}(\mathds{1}-2P^a)$, which also implies $\{P^a,G^{ab}\}=\{P^b,G^{ab}\}=G^{ab}$.
\item $[P^a,G^{ab}]=-[P^b,G^{ab}]$.
\item $[G^{ab},G^{cd}]=0$, whenever $G^{ab}$ and $G^{cd}$ swap different pentagons, i.e.\ $a,b$ are both different from $c,d$.
\end{enumerate}
Note that (a) and (b) imply that there are only 15 independent pentagon conservation equations arising from (\ref{eq:pentGeq}, \ref{pentrTreq}). These are exhibited in (\ref{15cons}). Furthermore, (c) corresponds to the vanishing of the structure constants $f^{(ab)(cd)(eg)}$ of $\PSU(4)$ whenever $G^{ab}$ and $G^{cd}$ swap different pentagons or similarly to the commutation of $\PSO(6)$ rotations in the different planes $(ab)$ and $(cd)$. Throughout the derivation we will also use the relation 
\ba
\tilde{\mathds{1}}^{ab}:=\tilde{\mathds{1}}^{\text{Pent}_a,\text{Pent}_b}=P^a+P^b-2P^aP^b.\nn
\ea
\begin{lemma}
\leavevmode
Define $\vec{r}\,'=\exp(t G^{ab})\cdot \vec{r}$ where $G^{ab}$ is any of the 15 swap generators (\ref{swapgen1}, \ref{swapgen2}). If $\text{tr}[P^cR]=1$ for all $1\leq c\leq 6$, then $\text{tr}[P^cR']=\text{tr}[P^cR]+2\sin(t)\cos(t)\text{tr}[P^cG^{ab}R]$.
\label{le:N2checkA}
\end{lemma}
\begin{proof}
By using the fact that the diagonal matrices $\tilde{\mathds{1}}^{ab}$ and $P^c$ commute, together with $(\tilde{\mathds{1}}^{ab})^2=\tilde{\mathds{1}}^{ab}$,$\tilde{\mathds{1}}^{ab}\cdot G^{ab}=G^{ab}$, the properties of the trace $\text{tr}[M]=\text{tr}[M^t]$,$\text{tr}[MN]=\text{tr}[NM]$ and further straightforward trigonometry we can show
\begin{gather*}
\text{tr}[P^cR']=\text{tr}[((\cos(t)-1)\tilde{\mathds{1}}^{ab}-\sin(t)G^{ab}+\mathds{1})P^cR((\cos(t)-1)\tilde{\mathds{1}}^{ab}+\sin(t)G^{ab}+\mathds{1})] \nonumber\\
=\text{tr}[P^cR]+2\sin(t)\cos(t)\text{tr}[P^cG^{ab}R]-\sin^2(t)(\text{tr}[P^c\tilde{\mathds{1}}^{ab}R]+\text{tr}[P^cG^{ab}RG^{ab}]),
\end{gather*}
where we denoted by $R'=\vec{r}\,'\cdot \vec{r}\,'{}^t=\exp(-t G^{ab})\vec{r}\cdot\vec{r}\,{}^t\exp(t G^{ab})=\exp(-t G^{ab})R\exp(t G^{ab})$. The last term $\sim\text{tr}[P^c\tilde{\mathds{1}}^{ab}R]+\text{tr}[P^cG^{ab}RG^{ab}]$ on the second line above vanishes. To see this we use (a), together with $\tilde{\mathds{1}}^{ab}=P^a+P^b-2P^aP^b$, $(G^{ab})^2=-\tilde{\mathds{1}}^{ab}$ to get
\begin{gather*}
\text{tr}[P^cG^{ab}RG^{ab}]+\text{tr}[P^c\tilde{\mathds{1}}^{ab}R]=\text{tr}[P^c(G^{ab})^2R]-\text{tr}[[P^c,G^{ab}]G^{ab}R]+\text{tr}[P^c\tilde{\mathds{1}}^{ab}R] \nonumber\\
=(\delta_{ca}+\delta_{cb})\text{tr}[(\mathds{1}-2P^c)\tilde{\mathds{1}}^{ab}R]=(\delta_{ca}+\delta_{cb})\text{tr}[(\mathds{1}-2P^c)(P^a+(\mathds{1}-2P^a)P^b)R].
\end{gather*}
If $c\neq a,b$ the above vanishes because of the $\delta$'s. Choosing $c=a$ without loss of generality, it vanishes again because of $P^c\cdot P^c=P^c$ and thus $(\mathds{1}-2P^c)^2=1$, which implies $\text{tr}[(\mathds{1}-2P^a)(P^a+(\mathds{1}-2P^a)P^b)R]=\text{tr}[((\mathds{1}-2P^a)P^a+P^b)R]=\text{tr}[(P^b-P^a)R]=\text{tr}[P^bR]-\text{tr}[P^aR]=1-1=0$.
\end{proof}

Using this result, we can now move on to show that, in fact, the 21 equations (\ref{pure state}) define an $\ct_2'$-invariant set, where $\ct_2'=\rm{PSU}(4)$ is the {\it full} group generated by exponentiating the 15 generators (\ref{swapgen1}, \ref{swapgen2}) and their linear combinations. That is, not only the pentagon, but also the pentagon conservation equations are preserved by $\ct_2'$.\footnote{We note that the pentagon equalities (\ref{pentin}) alone are generally not preserved under $\ct_2'$ without the generator conservation equations in (\ref{pure state}). For instance, the information distribution $\alpha_{xy}=\alpha_{z_2}=\alpha_{xz}=\alpha_{y_2}=\f{1}{2}$ \texttt{bit} and $\alpha_{x_1}=1$ \texttt{bit} (and all other $\alpha_i=0$) satisfies all pentagon equalities, however, violates the generator conservation equations. Under a finite evolution with $T^{\text{Pent}_4\text{Pent}_6}(t)$ in (\ref{eq:pentexpT}) (this is a rotation around the $x_2$-axis) this can be evolved to $\alpha_{xy}+\alpha_{z_2}>1$ \texttt{bit}, thereby violating the equality for $\text{Pent}_4$.}

\begin{lemma}
\leavevmode
If $\vec{r}$ satisfies (\ref{pure state}), then so does $\vec{r}\,'=T\cdot \vec{r}$ for any $T\in \mathcal{T}_2'$. 
\label{le:N2checkB}
\end{lemma}
\begin{proof}
We start by showing that every time evolution $T\in\ct_2'=\PSU(4)$ can be written as a product of exponentials, i.e.\ $T=\prod_{ab}\exp(t_{ab} G^{ab})$ where always a single generator $G^{ab}$ (from a given basis) appears in every exponent. First note that {\it any} matrix $T\in GL(\mathbb{R},n^2)$ lying in $\SO(n)$ can be expressed as a product of rotation matrices $\exp(tG_F^{lm})$, each in some plane $(lm)$ \cite{Raffenetti:1969,Raffenetti:1972} 
by the use of generalized Euler angles, where $G_F^{lm}$ are the anti-symmetric generators of the {\it fundamental} representation of $\SO(n)$, i.e.\ $(G_F^{lm})_{ij}=\delta_{li}\delta_{mj}-\delta_{lj}\delta_{mi}$. This statement is true for the entire equivalence class of generators, where the equivalence relation amounts to similarity transformations of the fundamental generators. That is, all of the choices of Lie algebra bases in that equivalence class have the same structure constants. The statement that any group element can be written as products of exponentials of single generators (of a basis from this equivalence class) can also be understood abstractly at the manifold level of the Lie group and hence must be true for any representation (of the equivalence class of bases). 

The same therefore holds true for the fundamental generators of $\PSU(4)\simeq \PSO(6)$. Our $15\times 15$ swap generator matrices $G^{ab}$ are exactly in a one-to-one correspondence with the fundamental $6\times 6$ generator matrices $(G_F^{lm})_{ij}=\delta_{li}\delta_{mj}-\delta_{lj}\delta_{mi}$ of $\SO(6)$. This has also been explicitly checked for the matrices corresponding to (\ref{swapgen1}) and (\ref{swapgen2}). In other words, also in the adjoint representation all the $\PSU(4)\simeq \PSO(6)$ group elements generated by the generators (\ref{swapgen1}, \ref{swapgen2}) are expressible as products of {\it single} exponentials of our swap generators $G^{ab}$, where only {\it one} swap generator appears in each exponent as in (\ref{eq:pentexpT}). For this reason it suffices to consider $T=\exp(t G^{ab})$ in the following and then the case of a general $T=\prod_{ab}\exp(t_{ab} G^{ab})\in\mathcal{T}'_2\simeq \PSU(4)$ follows by induction.

Consider $\vec{r}\,'=\exp(t G^{ab})\cdot \vec{r}$ where $\vec{r}$ satisfies (\ref{pure state}), i.e.\ $\text{tr}[P^cR]=1$ (pentagon equalities) and $\text{tr}[P^cG^{de}R]=0$ (generator equalities) for all $1\leq c,d,e\leq 6$. From lemma \ref{le:N2checkA} it follows that $\text{tr}[P^cR']=\text{tr}[P^cR]+2\sin(t)\cos(t)\text{tr}[P^cG^{ab}R]=\text{tr}[P^cR]=1$ and thus $\vec{r}\,'$ also satisfies the pentagon equalities. It remains to show that $\vec{r}\,'$ satisfies the generator equalities $\text{tr}[P^cG^{de}R']=0$ as well. Note that if $d=a$ and $e=b$, then $\text{tr}[P^cG^{de}R']\sim\text{tr}[P^c\exp(-t'G^{ab})R'\exp(t'G^{ab})]-\text{tr}[P^cR']=\text{tr}[P^c\exp(-(t'+t)G^{ab})R\exp((t'+t)G^{ab})]-\text{tr}[P^cR]=2\sin(t'+t)\cos(t'+t)\text{tr}[P^cG^{ab}R]=0$ because of lemma \ref{le:N2checkA}. Therefore, we should only consider the case where $a\neq d,e$ and/or $b\neq d,e$. Using the explicit expression for $\exp(t G^{ab})$ in (\ref{eq:pentexpT}), one finds
\begin{gather*}
\text{tr}[P^cG^{de}R']=\text{tr}[((\text{c}(t)-1)\tilde{\mathds{1}}^{ab}-\text{s}(t)G^{ab}+\mathds{1})P^cG^{de}R((\text{c}(t)-1)\tilde{\mathds{1}}^{ab}+\text{s}(t)G^{ab}+\mathds{1})] \nonumber\\
=\frac{1}{2}((\text{c}(t)-1)^2M_1+2\text{s}(t)(\text{c}(t)-1)M_2-\text{s}^2(t)M_3+2(\text{c}(t)-1)M_4+2\text{s}(t)M_5), \nonumber\\
\text{c}(t):=\cos(t), \text{s}(t):=\sin(t), \ M_1=\text{tr}[[P^c,G^{de}](\tilde{\mathds{1}}^{ab}R\tilde{\mathds{1}}^{ab})], \ M_2=\text{tr}[[P^c,G^{de}](G^{ab}R\tilde{\mathds{1}}^{ab})], \nonumber\\
M_3=\text{tr}[[P^c,G^{de}](G^{ab}RG^{ab})], \ M_4=\text{tr}[[P^c,G^{de}](R\tilde{\mathds{1}}^{ab})], \ M_5=\text{tr}[[P^c,G^{de}](G^{ab}R)].
\end{gather*}

We will now show that $M_1=-M_3=M_4,M_2=M_5=0$, such that $\text{tr}[P^cG^{de}R']=\frac{1}{2}((\text{c}(t)-1)^2+\text{s}^2(t)+2(\text{c}(t)-1))M_4=0$. Because of (a), we will take without loss of generality $c=d$ throughout the derivation and use $\text{tr}[[G^{ab},G^{de}]P^gR]=f^{(ab)(de)(cg)}\text{tr}[G^{cg}P^gR]=\text{tr}[P^gG^{cg}R]=0$, and thus (d): $\text{tr}[G^{ab}G^{de}P^gR]=\text{tr}[G^{ab}G^{de}RP^g]$. From (d) it follows that $M_5=\text{tr}[[P^d,G^{de}](G^{ab}R)]=-\text{tr}[G^{de}[P^d,G^{ab}]R]$. Furthermore, using (b) and then (d) implies as well $M_5=-\text{tr}[[P^e,G^{de}](G^{ab}R)]=\text{tr}[G^{de}[P^e,G^{ab}]R]$. Note that $[P^d,G^{ab}]=0$ or $[P^e,G^{ab}]=0$ because of (a) and also $a\neq d,e$ and/or $b\neq d,e$ and therefore $M_5=0$. For showing the three remaining equalities $M_1=-M_3=M_4,M_2=0$ we consider two separate cases, $c_1$: $a=d$ and $b\neq d,e$, $c_2$: $a\neq d,e$ and $b\neq d,e$. The symmetric case $a=e$ and $b\neq d,e$ is also captured because of (b). 

Let us start with the case $c_1$, for which $[G^{de},P^b]=0$ because of (a). Then for $M_2$
\begin{gather*}
M_2=\text{tr}[[P^d,G^{de}](G^{db}R\tilde{\mathds{1}}^{db})]=\text{tr}[[P^d,G^{de}]P^bG^{db}R]+\text{tr}[P^d[P^d,G^{de}]((\mathds{1}-2P^b)G^{db}R)]\nonumber\\
=\frac{1}{2}(\text{tr}[[P^d,G^{de}][P^b,G^{db}]R]-\text{tr}[[G^{db},[P^d,G^{de}]]P^bR])+\text{tr}[P^d[P^d,G^{de}][G^{db},P^b]R]\nonumber\\
=\frac{1}{2}(-\text{tr}[(\mathds{1}-2P^d)[P^d,G^{de}][P^d,G^{db}]R]-\frac{1}{2}\text{tr}[[\{G^{db},P^b\},[P^d,G^{de}]]R]) \nonumber\\
=\frac{1}{2}(\text{tr}[G^{de}[P^d,G^{db}]R]-\frac{1}{2}\text{tr}[[G^{db},[P^d,G^{de}]]R])=0-0=0.
\end{gather*}
Similarly, for $M_3$ we can show
\begin{gather*}
M_3=\text{tr}[[P^d,G^{de}](G^{db}RG^{db})]=-\text{tr}[[P^e,G^{de}](G^{db}RG^{db})] \nonumber\\
=-(\text{tr}[[P^e,[G^{de},G^{db}]]RG^{db}]+\text{tr}[[P^e,G^{de}]R(G^{db})^2]) \nonumber\\
=-(f^{(de)(db)(rs)}\text{tr}[[P^e,G^{rs}]RG^{db}]-\text{tr}[[P^e,G^{de}](R\tilde{\mathds{1}}^{db})]) \nonumber\\
=\text{tr}[[P^e,G^{de}](R\tilde{\mathds{1}}^{db})]=-\text{tr}[[P^d,G^{de}](R\tilde{\mathds{1}}^{db})]=-M_4.
\end{gather*}
Finally for $M_1$ it also follows
\begin{gather*}
M_1=\text{tr}[[P^d,G^{de}](\tilde{\mathds{1}}^{db}R\tilde{\mathds{1}}^{db})]=\text{tr}[[P^d,G^{de}]P^dR\tilde{\mathds{1}}^{db}(\mathds{1}-2P^b)]+\text{tr}[[P^d,G^{de}]R\tilde{\mathds{1}}^{db}P^b] \nonumber\\
=\text{tr}[[[P^d,G^{de}],P^d]R\tilde{\mathds{1}}^{db}(\mathds{1}-2P^b)]+\text{tr}[[P^d,G^{de}]R\tilde{\mathds{1}}^{db}]=\frac{1}{2}\text{tr}[[P^d,[[P^d,G^{de}],P^d]R]+M_4 \nonumber\\
=-\frac{1}{2}\text{tr}[[P^d,G^{de}]R]+M_4=M_4.
\end{gather*}

Lastly, we consider the simplest case $c_2$, for which $[G^{de},P^a]=[G^{de},P^b]=0$ and $[G^{de},G^{ab}]=0$ because of (a) and (c). In particular, also $[G^{de},\tilde{\mathds{1}}^{ab}]=0$. Working out $M_2$ results again in
\begin{gather*}
M_2=\text{tr}[[P^d,G^{de}](G^{ab}R\tilde{\mathds{1}}^{ab})]=\text{tr}[[P^d,G^{de}]\tilde{\mathds{1}}^{ab}G^{ab}R]=\text{tr}[[P^d,G^{de}]G^{ab}R]=0.
\end{gather*}
For $M_3$ the same derivation as in the case of $c_1$ can be used to show $M_3=-M_4$. Finally, for $M_1=\text{tr}[[P^d,G^{de}](\tilde{\mathds{1}}^{ab}R\tilde{\mathds{1}}^{ab})]=\text{tr}[[P^d,G^{de}]R(\tilde{\mathds{1}}^{ab})^2]=M_4$.
\end{proof}



\subsubsection{$\rm{PSU}(2^N)$ is a maximal subgroup of $\SO(4^{N}-1)$}\label{app_max}

In the main text we argue that $\PSU(2^N)$ is a subgroup of the time evolution group $\mathcal{T}_N$, which itself is (isomorphic to) a subgroup of $\SO(4^{N}-1)$. In order to conclude that $\mathcal{T}_N\simeq \PSU(2^N)$, we prove here that $\PSU(2^N)$ is (isomorphic to) a maximal subgroup of the larger group $\SO(4^{N}-1)$:
\begin{lemma}
$\PSU(2^N)$ acts in the adjoint representation on the state space of $N$ qubits and is a maximal subgroup of $\SO(4^{N}-1)$ for all $N\geq 2$.
\label{le:PSUNmaxsub}
\end{lemma}
\begin{proof}
The irreducible representations of $\PSU(N_F)$ are categorized by $N_F-1$ numbers. Each representation corresponds to a Young Tableau to which its dimension is intimately related \cite{Georgi:1982jb}. A survey of the irreducible representations of $\PSU(N_F)$ shows that whenever $N_F\geq 9$, the dimensions of the irreducible representations can be ranked from lowest to highest as: $1$, $N_F$, $\frac{1}{2}N_F(N_F-1)$, $\frac{1}{2}N_F(N_F+1)$, $\dim(\text{Ad})=N_F^2-1,\cdots$ where the dots refer to higher dimensional representations with dimensions larger than $N_F^2-1$. The $N$ qubit state space transforms in some representation $R$ of $\PSU(N_F=2^N)$. Let us first consider the case of $N>3$, for which $N_F=2^N>9$. If $R$ was reducible, it would have at least one copy of the trivial representation in its direct sum, since $\dim(R)=\dim(\text{Ad})$ is uneven and all lower dimensional representations are even except the trivial one, which would imply that $R$ leaves a 1-dimensional subspace invariant. However, this is not possible because the subgroup $\PSU(2)\times\cdots\times \PSU(2)\subseteq \PSU(2^N)$ which corresponds to the rotations of the individual qubits lies inside the subgroup $\PSU(2^N)$ and these transformations would certainly not leave any $1$-dimensional subspace invariant. Therefore the representation $R$ must be irreducible and it {\it must} be the adjoint representation since all other representations are of larger dimension. For $N\leq 3$ one observes from explicit tables of group dimensions \cite{Feger:2012aa} that the same reasoning applies and $R$ again equals the adjoint representation. The maximality of $\PSU(2^N)$ in $\SO(4^{N}-1)$ now directly follows from Dynkin's theorem \cite{Dynkin:1957aa} and the fact that $\PSU(2^N)$ is simple and its adjoint representation is faithful and irreducible (acting on the fundamental representation space of $\SO(4^{N}-1)$).
\end{proof}

\subsubsection{Evolving to product states}\label{app_product}

We shall now demonstrate the following claim of section \ref{sec_statesn2}.

\begin{lemma}\label{lem_product}
\leavevmode
 {\it Any} $\vec{r}$ satisfying (\ref{pure state}) can be brought to the configuration $\alpha_{z_1}=\alpha_{z_2}=\alpha_{zz}=1$ \texttt{bit} and all other $\alpha_i=0$ by performing successive $\ct_2'$-transformations of the form (\ref{eq:pentexpT}). 
\end{lemma}

\begin{proof}
First note that when two questions are swapped {\it within} one pentagon, the other three questions remain unchanged (cf.\ figure \ref{fig:pentagon}). This is the case because the remaining three questions do {\it not} appear in any of the two pentagons whose information contents are being swapped. For example, this can be explicitly seen in figure \ref{fig:pentagon}, where as $\alpha_{y_1}$ and $\alpha_{zx}$ are swapped within $\text{Pent}_5$ via $G^{\text{Pent}_1,\text{Pent}_2}$, the information content of $\alpha_{x_1},\alpha_{zz},\alpha_{zy}$ in $\text{Pent}_5$ are left invariant. Because of this property we can by repeating (at most 4 different) swap transformations (\ref{eq:pentexpT}) put all $1$ \texttt{bit} of information contained in $\text{Pent}_5$ in, e.g., question $Q_{y_1}$: 1) first rotate all information from $Q_{zx}$ to $Q_{y_1}$ with $T^{\text{Pent}_1,\text{Pent}_2}(t_1)$ for some $t_1$ such that $\alpha_{zx}=0$,\footnote{Recall from above that $T^{\text{Pent}_1,\text{Pent}_2}(t_1)$ acts as a rotation by $\pm t_1$ in the plane $(r_{y_1},r_{zx})$.} 2) then use $T^{\text{Pent}_3,\text{Pent}_5}(t_2)$ for some $t_2$ to rotate the information contained in $Q_{zy}$ into $Q_{y_1}$ which leaves $\alpha_{zx}=\alpha_{zy}=0$, 3) use $T^{\text{Pent}_1,\text{Pent}_6}(t_3)$ for some $t_3$ to map the information content of $Q_{zz}$ into $Q_{y_1}$ which leaves $\alpha_{zx}=\alpha_{zy}=\alpha_{zz}=0$, 4) finally use $T^{\text{Pent}_1,\text{Pent}_3}(t_4)$ for some $t_4$ to rotate the information from $Q_{x_1}$ into $Q_{y_1}$ which leaves $\alpha_{zx}=\alpha_{zy}=\alpha_{zz}=\alpha_{x_1}=0$. Since the time evolution group maps pure states to pure states and $I(\text{Pent}_5)=1$ \texttt{bit}, we conclude $\alpha_{y_1}=1$ \texttt{bit} after the four steps. The information content of questions in other pentagons is also transformed during these four successive transformations. However, since every employed transformation leaves the other three questions in $\text{Pent}_5$ invariant, this is not relevant for the argument. Nevertheless, all eight questions complementary to $Q_{y_1}$ will necessarily have $\alpha_i=0$ too, while the remaining $2$ \texttt{bits} will be distributed over the six questions compatible with $Q_{y_1}$.

The above information redistribution algorithm, by using appropriate combinations of transformations, can similarly be performed on {\it any} state satisfying (\ref{pure state}) to get $\alpha_{z_1}=1$ \texttt{bit}. In that case, the remaining $2$ \texttt{bits} will be contained in the boundary of the three compatible triangles with central vertex $\alpha_{z_1}=1$ \texttt{bit} (cf.\ (\ref{bdry}) and figure \ref{fig:1bitquestion}) and $\alpha_{x_2}=\alpha_{zx}$, $\alpha_{zy}=\alpha_{y_2}$, $\alpha_{zz}=\alpha_{z_2}$. Using the three latter equalities and the fact that the 6 boundary questions contain $2$ \texttt{bits} of information, it follows that $\alpha_{x_2}+\alpha_{y_2}+\alpha_{z_2}=1$ \texttt{bit}. We can evolve this $1$ \texttt{bit} of information into $\alpha_{z_2}=1$ \texttt{bit} by using local rotations of qubit 2: 1) first rotate around the $r_{x_2}$-axis with $T^{\text{Pent}_4,\text{Pent}_6}$ to get $\alpha_{y_2}=0$, 2) then rotate around the $r_{y_2}$-axis with $T^{\text{Pent}_2,\text{Pent}_6}$ while leaving $\alpha_{y_2}=0$ and putting $\alpha_{x_2}=0$ and thus $\alpha_{z_2}=1$ \texttt{bit}. Finally, we therefore reach the required product configuration $\alpha_{z_1}=\alpha_{z_2}=\alpha_{zz}=1$ \texttt{bit}, starting from any pure state. Note that this required at most six {\it different} successive transformations of the form (\ref{eq:pentexpT}).
\end{proof}

\subsubsection{Preservation of the complementarity inequalities}\label{app_comp}

Next, we show that $\ct'_2=\rm{PSU}(4)$ preserves all complementarity inequalities (\ref{compstrong}) -- provided (\ref{pure state}) is fulfilled. Since $\ct'_2$ preserves all pentagon equalities by construction, it suffices to check it for the triangle complementarity sets in appendix \ref{triangle-app} since all sets of mutually complementary questions are either contained in the pentagon or triangle sets on account of their maximality.

\begin{lemma}
\leavevmode
Any $\vec{r}$ solving (\ref{pure state}) also satisfies all triangle complementarity inequalities following from (\ref{compstrong}) for the triangle complementarity sets (\ref{eq:triangles}).
\end{lemma}
\begin{proof}
By inspection one verifies that any of the three pairs of questions contained in every of the triangle sets $\text{Tri}_i$ (\ref{eq:triangles}) also lies in a common pentagon set (\ref{eq:pentineq}). However, clearly not all three questions in a triangle set can lie in the same pentagon, for otherwise maximality of the triangle set would be violated. This implies that for every triangle set and every pair of questions contained in it there exists an information swap generator in (\ref{swapgen1}, \ref{swapgen2}) which swaps the information between the two questions of that pair and leaves the third question in the triangle set invariant (see the arguments in the proof of lemma \ref{lem_product}). For example, for $\text{Tri}_1$, $G^{\text{Pent}_1,\text{Pent}_4}$ swaps the information between $Q_{xx}$ and $Q_{z_2}$ and leaves $Q_{xy}$ invariant. Accordingly, the exponentiation (\ref{eq:pentexpT}) of $G^{\text{Pent}_1,\text{Pent}_4}$ rotates information continuously between $Q_{xx}$ and $Q_{z_2}$ and leaves $Q_{xy}$ invariant. In particular, there will always exist values of $t$ such that {\it all} information carried by $Q_{xx}$ and $Q_{z_2}$ can be evolved into one of the two questions, e.g., $Q_{xx}$. By subsequently applying the analogous rotation generated by $G^{\text{Pent}_2,\text{Pent}_4}$ to the pair $Q_{xx},Q_{xy}$ (which leaves $Q_{z_2}$ invariant), one can always evolve the entire information $I(\text{Tri}_1)(t)=\alpha_{xx}(t)+\alpha_{xy}(t)+\alpha_{z_2}(t)$, carried by $\text{Tri}_1$, into $I(\text{Tri}_1)(t)=I(\text{Tri}_1)(t+\Delta t)=\alpha_{xx}(t+\Delta t)$ such that $\alpha_{xy}(t+\Delta t)=\alpha_{z_2}(t+\Delta t)=0$ \texttt{bits} and no information has leaked out of the triangle (where $\alpha_i(t)=r_i(t)^2$). That is, if the triangle complementarity inequality following from (\ref{compstrong}) for $\text{Tri}_1$ was ever violated, $\alpha_{xx}+\alpha_{xy}+\alpha_{z_2}>1$ \texttt{bit}, there would exist a $T\in\ct_2'$ which evolves this configuration to $\alpha_{xx}>1$ \texttt{bit}. But this would violate the pentagon equalities which, by lemma \ref{le:N2checkB}, can never happen under $\ct_2'$ if (\ref{pure state}) is fulfilled. The same argument can be repeated for all 20 triangle sets such that we conclude that (\ref{pure state}), in fact, implies that the triangle complementarity inequalities hold. 
\end{proof}

In particular, $\ct_2'$ thus preserves all complementarity inequalities (\ref{compstrong}) once (\ref{pure state}) holds.

\subsubsection{Preservation of the correlation structure}\label{app_corr}

We also have to check that $\ct_2'=\rm{PSU}(4)$ leaves the correlation structure of figure \ref{fig_corr} invariant -- provided (\ref{pure state}) is fulfilled. For this purpose we recall that the correlation structure in figure \ref{fig_corr} encodes that a question in an (anti-)correlation triangle is the (anti-)correlation of the other two questions in the triangle. The correlation structure thus means that if (a) $Q_i=Q_j\leftrightarrow Q_k$ then $y_i=1$ implies $r_j=r_k$ and $y_i=0$ implies $r_j=-r_k$,\footnote{For example, $Q_{xx}=Q_{x_1}\leftrightarrow Q_{x_2}$ is the question ``are the answers to $Q_{x_1},Q_{x_2}$ correlated?" Since $y_{xx}$ is thus also the probability that the answers to $Q_{x_1},Q_{x_2}$ are correlated, this means that whenever $y_{xx}=1$ we must have $y_{x_1}=y_{x_2}$ and whenever $y_{xx}=0$ we must have $r_{x_1}=-r_{x_2}$.} and if (b) $Q_i=\neg(Q_j\leftrightarrow Q_k)$ then $y_i=1$ implies $r_j=-r_k$ and $y_i=0$ implies $r_j=r_k$, where $i,j,k=x_1,y_1,\ldots,zz$ and $i\neq j\neq k\neq i$ are question indices compatible with a triangle in figure \ref{fig_corr}. That is, since any $Q_i$ is contained in three triangles (c) if $Q_i=Q_j\leftrightarrow Q_k=\neg(Q_l\leftrightarrow Q_m)$ then $y_i=1$ implies $r_j=r_k$ and $r_l=-r_m$ simultaneously and $y_i=0$ implies $r_j=-r_k$ and $r_l=r_m$ simultaneously. Finally, (d) if $Q_i=Q_j\leftrightarrow Q_k=Q_l\leftrightarrow Q_m$ then $y_i=1$ implies $r_j=r_k$ and $r_l=r_m$ simultaneously and $y_i=0$ implies $r_j=-r_k$ and $r_l=-r_m$ simultaneously. We thus only show the statement for states with at least one $\alpha_i=1$ \texttt{bit} for which the correlation structure has meaning.

We recall from lemma \ref{lem_product} and the arguments of section \ref{sec_statesn2} that there exist precisely two $\ct_2'$-transitive sets solving (\ref{pure state}), namely
\ba
\cs^+_{QT}&:=&\{T\cdot(\vec{\delta}_{z_1}+\vec{\delta}_{z_2}+\vec{\delta}_{z_1z_2})\,\big|\,T\in\rm\ct_2'\},\nn\\
\cs^-_{QT}&:=&-\{T\cdot(\vec{\delta}_{z_1}+\vec{\delta}_{z_2}+\vec{\delta}_{z_1z_2})\,\big|\,T\in\rm\ct_2'\}.\nn
\ea
$\cs^+_{QT}$ is the set of pure quantum states, while $\cs^-_{QT}$ constitutes an equivalent but different representation of the pure quantum state space. These two sets are {\it not} connected via $\ct_2'$.

\begin{claim}
\leavevmode
\begin{enumerate}
\item Any $\vec{r}$ which solves (\ref{pure state}) and satisfies the correlation structure of figure \ref{fig_corr} lies in $\cs^+_{QT}$. This is the set corresponding to the convention of building bipartite questions from the individuals $Q_{x_1},Q_{x_2},Q_{y_1},Q_{y_2},Q_{z_1},Q_{z_2}$ using the XNOR connective $\leftrightarrow$.
\item Any $\vec{r}$ which solves (\ref{pure state}) and satisfies the correlation structure obtained by replacing correlation triangles in figure \ref{fig_corr} by anti-correlation triangles and vice versa lies in $\cs^-_{QT}$. This is the set corresponding to the convention of building bipartite questions from the individuals using the XOR connective $\neg(\cdot\,\leftrightarrow\,\cdot)$.
\end{enumerate}
Thus, in particular, in the XNOR convention, (\ref{pure state}) implies the correlation structure of figure \ref{fig_corr} which therefore is $\ct_2'$-invariant.
\end{claim}
\begin{proof}
Suppose $\vec{r}$ solves (\ref{pure state}). This implies that whenever $\alpha_i=1$ \texttt{bit}, then $\alpha_j=\alpha_k$ if either $Q_i=Q_j\leftrightarrow Q_k$ or $Q_i=\neg(Q_j\leftrightarrow Q_k)$ (as mentioned at the end of section \ref{sec_infodist} this follows from the pentagon identities contained in (\ref{pure state})). This means that $r_i=\pm1$ and either $r_j=r_k$ or $r_j=-r_k$. We wish to show consistency with (a)--(d). We shall illustrate the argument with the example of $\alpha_{z_1}=1$ \texttt{bit}. 
While the proof is straightforward, it involves many details such that we restrict to a sketch.

We adopt the notation of appendix \ref{app_geneq} and note that the conservation equation 
\ba
(P^4\cdot\vec{r})\cdot G^{46}\cdot\vec{r}=r_{z_2}r_{y_2}-r_{zy}r_{zz}-r_{yy}r_{yz}-r_{xy}r_{xz}=0\nn
\ea
reads $r_{z_2}r_{y_2}=r_{zy}r_{zz}$ once $\alpha_{z_1}=1$ \texttt{bit} such that all questions complementary to $Q_{z_1}$ carry $0$ \texttt{bits}. Together with $r_{y_2}=\pm r_{zy}$ and $r_{zz}=\pm r_{z_2}$, which is implied by the pentagon equalities as noted above, this entails for the right and lower triangles in figure \ref{fig:1bitquestion} that
\ba
r_{y_2}=+ r_{zy}\q\Leftrightarrow\q r_{zz}=+ r_{z_2}\q\q\text{and}\q\q
r_{y_2}=- r_{zy}\q\Leftrightarrow\q r_{zz}=- r_{z_2}\nn.
\ea

We can now employ finite time evolutions $T^{46}(t)=\exp(t\,G^{46})$ as in (\ref{eq:pentexpT}) which generate rotations in the $(y_2,z_2)$ and $(zy,zz)$ planes, both by an angle $-t$. Such a time evolution leaves $r_{x_2},r_{z_1},r_{zx}$, corresponding to the upper left triangle in figure \ref{fig:1bitquestion} invariant. In particular, we can start with a Bloch vector $r_{z_1}=r_{y_2}=r_{zy}=+1$ and all other $r_i=0$. This Bloch vector solves (\ref{pure state}) and is compatible with constructing the bipartite question $Q_{zy}=Q_{z_1}\leftrightarrow Q_{y_2}$ via the XNOR connective $\leftrightarrow$. Applying $T^{46}(t)=\exp(t\,G^{46})$ for all $t\in[0,2\pi]$ to this vector generates {\it all} configurations for which $r_{z_1}=+1$ and simultaneously $r_{y_2}=+ r_{zy}$ and $r_{zz}=+ r_{z_2}$, while all other $r_i=0$ and thus preserving that all of $I(\vec{r})=|\vec{r}|^2=3$ \texttt{bits} is carried by the five questions in the upper right and the lower triangle in figure \ref{fig:1bitquestion}. Similarly, by starting with the Bloch vector $r_{z_1}=r_{y_2}=-1$, $r_{zy}=+1$ and all other $r_i=0$, which again solves (\ref{pure state}) and is compatible with $Q_{zy}=Q_{z_1}\leftrightarrow Q_{y_2}$, one can generate all configurations for which $r_{z_1}=-1$ and simultaneously $r_{y_2}=- r_{zy}$ and $r_{zz}=- r_{z_2}$, while all other $r_i=0$ and thus preserving that all of $I(\vec{r})=|\vec{r}|^2=3$ \texttt{bits} is carried by the five questions in the upper right and the lower triangle in figure \ref{fig:1bitquestion}. Note, firstly, that the two states $r_{z_1}=r_{y_2}=r_{zy}=+1$ (all other $r_i=0$) and $r_{z_1}=r_{y_2}=-1$, $r_{zy}=+1$ (all other $r_i=0$) are connected by $T(t=\pi)=\exp(\pi\,G^{12})$ such that all the states we just discussed are connected by time evolution and thus clearly satisfy (\ref{pure state}). Secondly, note that {\it all} of these Bloch vectors are consistent with building the bipartite $Q_{zz}=Q_{z_1}\leftrightarrow Q_{z_2}$ using XNOR and, accordingly, with $Q_{zy}\leftrightarrow Q_{y_2}=Q_{z_1}=Q_{zz}\leftrightarrow Q_{z_2}$. Thirdly, note that we could have arrived at the same result by using the conservation equation $(P^2\cdot\vec{r})\cdot G^{25}\cdot\vec{r}=0$ and $T^{25}(t)$ which also leaves the questions in the upper left triangle of figure \ref{fig:1bitquestion} invariant. 

One can repeat the analogous argument with $G^{26}$ or $G^{45}$, both of which leave the upper right triangle in figure \ref{fig:1bitquestion} invariant and solely rotate the information between the other two triangles (while leaving $r_{z_1}$ invariant), to show that from $r_{z_1}=r_{z_2}=r_{zz}=1$ (all other $r_i=0$) one can generate by time evolution {\it all} states with $r_{z_1}=+1$ and simultaneously $r_{z_2}=r_{zz}$ and $r_{x_2}=r_{zx}$ and {\it all} states with $r_{z_1}=-1$ and simultaneously $r_{z_2}=-r_{zz}$ and $r_{x_2}=-r_{zx}$ and all other $r_i=0$. Since $r_{z_1}=r_{z_2}=r_{zz}=1$ (all other $r_i=0$) is connected by time evolution to the states of the previous paragraph all of these states are likewise related through time evolution group elements to all states of the previous paragraph. We again note that all of these states are consistent with constructing the bipartite $Q_{zx}=Q_{z_1}\leftrightarrow Q_{x_2}$ with the XNOR from the individuals $Q_{z_1},Q_{x_2}$ and, accordingly, with $Q_{zx}\leftrightarrow Q_{x_2}=Q_{z_1}=Q_{zz}\leftrightarrow Q_{z_2}$.

Next, we repeat the analogous argument with $G^{24}$ or $G^{56}$, both of which leave the lower triangle in figure \ref{fig:1bitquestion} invariant, to show that from $r_{z_1}=r_{y_2}=r_{zy}=+1$ (all other $r_i=0$) one can produce through time evolution group elements {\it all} states with $r_{z_1}=+1$ and simultaneously $r_{y_2}=r_{zy}$ and $r_{x_2}=r_{zx}$ and {\it all} states with $r_{z_1}=-1$ and simultaneously $r_{y_2}=-r_{zy}$ and $r_{x_2}=-r_{zx}$ and all other $r_i=0$. All of these states are clearly connected via time evolution group elements to all states of the previous two paragraphs and consistent with $Q_{zx}\leftrightarrow Q_{x_2}=Q_{z_1}=Q_{zy}\leftrightarrow Q_{y_2}$.

Combining the previous arguments, it is clear that by applying all possible products of $T^{46},T^{25},T^{26},T^{45},T^{24}$, $T^{56}$ for all possible values of $t\in[0,2\pi]$ to the states of the previous three paragraphs one generates {\it all} states with $r_{z_1}=+1$ and simultaneously $r_{x_2}=r_{zx}$ and $r_{zy}=r_{y_2}$ and $r_{zz}=r_{z_2}$ and {\it all} states with $r_{z_1}=-1$ and simultaneously $r_{x_2}=-r_{zx}$ and $r_{zy}=-r_{y_2}$ and $r_{zz}=-r_{z_2}$ and all other $r_i=0$ and $I(\vec{r})=|\vec{r}|^2=3$ \texttt{bits}. It is also clear that all these states satisfy (\ref{pure state}) and that no other states can be produced by combinations of $T^{46},T^{25},T^{26},T^{45},T^{24}$, $T^{56}$. But these are precisely all the states consistent with $Q_{z_1}=Q_{zy}\leftrightarrow Q_{y_2}=Q_{zz}\leftrightarrow Q_{z_2}=Q_{zx}\leftrightarrow Q_{x_2}$ and $\alpha_{z_1}=1$ \texttt{bit} and thus all the states consistent with the correlation structure of figure \ref{fig:1bitquestion}. In conclusion, all of these states are thus implied by (\ref{pure state}), provided one follows the convention to only build up bipartite questions with the XNOR connective from individual questions.

Had we instead started the above arguments with the state $r_{z_1}=-1$, $r_{y_2}=r_{zy}=+1$ and all other $r_i=0$, corresponding to the XOR connective $Q_{zy}=\neg(Q_{z_1}\leftrightarrow Q_{y_2})$ and solving (\ref{pure state}), we would have produced through time evolution all states consistent with $Q_{z_1}=\neg(Q_{zy}\leftrightarrow Q_{y_2)}=\neg(Q_{zz}\leftrightarrow Q_{z_2})=\neg(Q_{zx}\leftrightarrow Q_{x_2})$ and $\alpha_{z_1}=1$ \texttt{bit}. These correspond to the correlation structure of figure \ref{fig:1bitquestion}, except that all correlation triangles in it are replaced by anti-correlation triangles. 

Clearly, one can repeat the same arguments for any question $Q_i$ and Bloch vectors with $\alpha_i=1$ \texttt{bit}, finding that {\it all} states compatible with $\alpha_i=1$ \texttt{bit} and building bipartite questions with the XNOR are connected via $\ct_2'$ and likewise that {\it all} states compatible with $\alpha_i=1$ \texttt{bit} and building bipartite questions with the XOR are connected via $\ct_2'$.

Together with lemma \ref{lem_product} and the arguments of section \ref{sec_statesn2} it follows that {\it all} $3$ \texttt{bit} states consistent with the correlation structure of figure \ref{fig_corr} 
lie in $\cs^+_{QT}$. Similarly, it follows that {\it all} $3$ \texttt{bit} states consistent with the correlation structure corresponding to the convention of constructing bipartite questions with the XOR from individuals lie in $\cs^-_{QT}$.
\end{proof}

\subsection{Reconstructing $\ct_N$ and $\Sigma_N$ for $N>2$}\label{app_nl2}

\subsubsection{Deriving the 'swap generators' for $N>2$}\label{app_swapnl2}

All pairwise unitaries must be contained in $\ct_N$ and therefore require a representation on $\mathbb{R}^{4^N-1}$. Consider the gbit pair $(1,2)$. It is not difficult to convince oneself that the definition (and requirement) of isolated evolution under $\ct^{(12)}_2\subset\ct_N$ from section \ref{sec_psunl2} implies that {\it every} $T^{(12)}(t)\in\ct^{(12)}_2$ must be of the block-diagonal form
\ba
T^{(12)}(t) = \left(\begin{array}{ccc}\bar{T}^{(12)}(t) & 0 & 0 \\0 & \tilde{T}^{(12)}(t) & 0 \\0 & 0 & \mathds{1}_{(4^{N-2}-1)\times(4^{N-2}-1)}\end{array}\right),\label{T12}
\ea
where $\bar{T}^{(12)}(t)$ is the corresponding $15\times15$ $\ct_2$-matrix of section \ref{sec_psu4} and $\tilde{T}^{(12)}(t)$ is a $(4^N-1-15-(4^{N-2}-1))\times(4^N-1-15-(4^{N-2}-1))$ matrix which acts on the indices $(\mu_1\mu_2 0\cdots0)$ and $(\mu_1\mu_2 \mu_3\cdots\mu_N)$, respectively, of a Bloch vector $\vec{r}\in\Sigma_N$, where $(\mu_1\mu_2)\neq(00)$ and $(\mu_3\cdots\mu_N)\neq(0\cdots0)$. Therefore, the generators of $\ct^{(12)}_2$ must be of the following block-diagonal form:
\ba
G^{(12)} = \left(\begin{array}{ccc}\bar{G}^{(12)} & 0 & 0 \\0 & g^{(12)} & 0 \\0 & 0 & 0_{(4^{N-2}-1)\times(4^{N-2}-1)}\end{array}\right),\label{gennl2}
\ea
where $\bar{G}^{(12)}$ are (linear combinations of) the two-qubit information swap generators (\ref{swapgen1}, \ref{swapgen2}) and $g^{(12)}$ are the generators of $\tilde{T}^{(12)}$. The latter clearly also have to form a representation of $\mathfrak{psu}(4)$ in order for the $G^{(12)}$ to generate a $(4^N-1)\times(4^N-1)$ matrix representation of $\mathfrak{psu}(4)$ such that $g^{(12)}$ must be anti-symmetric too. Note that this resulting $\mathfrak{psu}(4)$ representation will thus be reducible. The analogous block-diagonal form holds for the pairwise unitaries and their generators of all other gbit pairs.

We shall now prove equation (\ref{psu4adj}). We shall do this in three steps, each given by a lemma. Note that the indices of the matrix $\tilde{T}^{(12)}_{(\mu_1\cdots\mu_N)(\nu_1\cdots\nu_N)}$ are {\it always} such that $(\mu_1\mu_2)\neq(00)$ and $(\mu_3\cdots\mu_N)\neq(0\cdots0)$ (and similarly for the $\nu$ indices). However, we can trivially extend $\tilde{T}^{(12)}$ to an $(4^N-1)\times(4^N-1)$ matrix by simply setting all new components corresponding to all remaining index combinations to zero. In this case we can let the indices $\mu,\nu$ run over all possible values.

\begin{lemma}\label{lem_6}
\leavevmode
$\tilde{T}^{(12)}_{(\mu_1\mu_2\mu_3\cdots\mu_N)(\nu_1\nu_2\nu_3\cdots\nu_N)}=M_{(\mu_1\mu_2)(\nu_1\nu_2)}(\mu_3,\ldots,\mu_N)\,\delta_{\mu_3\nu_3}\cdots\delta_{\mu_N\nu_N}$. Here the factor $M_{(\mu_1\mu_2)(\nu_1\nu_2)}(\mu_3,\ldots,\mu_N)$ is a $16\times16$ matrix which might depend on the values of the indices $(\mu_3\cdots\mu_N)$, and $M_{(00)(\nu_1\nu_2)}(\mu_3,\ldots,\mu_N)=0=M_{(\mu_1\mu_2)(00)}(\mu_3,\ldots,\mu_N)$.
\end{lemma}

\begin{proof}
We shall show that the matrix components $\tilde{T}^{(12)}_{(\mu_1\cdots\mu_N)(\nu_1\cdots\nu_N)}$ (for simplicity we drop here the argument $t$) 
vanish whenever $\mu_3\neq \nu_3$. By symmetry in the qubit labels, it then follows more generally that $\tilde{T}^{(12)}$ vanishes unless $\mu_3=\nu_3,\ldots, \mu_N=\nu_N$. (Clearly, the proof below can also be performed for the fourth, fifth and higher indices.) Throughout this proof we use that two questions $Q_{\mu_1\cdots\mu_N}$ and $Q_{\nu_1\cdots\nu_N}$ are complementary iff their indices differ in an odd number of non-zero indices \cite{Hoehn:2014uua}.

Consider now $\tilde{T}^{(12)}_{(\mu_1\cdots\mu_N)(\nu_1\cdots\nu_N)}$ with the indices $(\mu_1\cdots\mu_N)$ and $(\nu_1\cdots\nu_N)$ {\it fixed} and $\mu_3\neq \nu_3$. We shall henceforth also assume that $(\mu_1\mu_2)\neq(00)\neq(\nu_1\nu_2)$ and, likewise, $(\mu_3\cdots\mu_N)\neq(0,\cdots0)\neq(\nu_3\cdots\nu_N)$ for otherwise this component of $\tilde{T}^{(12)}$ is trivially zero. These two index sets will correspond to two questions $Q_{\mu_1\cdots\mu_N},Q_{\nu_1\cdots\nu_N}$. We shall now choose a further question $Q_{00\nu_3'\cdots\nu_N'}$ such that it is complementary to $Q_{\mu_1\cdots\mu_N}$ and compatible with $Q_{\nu_1\cdots\nu_N}$. At the end of the proof we shall show that this is always possible.

Since $Q_{\nu_1\cdots\nu_N},Q_{00\nu_3'\cdots\nu_N'}$ are compatible, whenever $O$ knows the answer to the two with certainty, he will also know with certainty the answer to their correlation $Q_{\nu_1\nu_2\tilde{\nu_3}\cdots\tilde{\nu_N}}=Q_{\nu_1\nu_2\nu_3\cdots\nu_N}\leftrightarrow Q_{00\nu_3'\cdots\nu_N'}$, where $(\tilde{\nu_3}\cdots\tilde{\nu_N})$ depend on $(\nu_3\cdots\nu_N)$ and $(\nu'_3\cdots\nu_N')$. Since $(\nu_3\cdots\nu_N)\neq(0\cdots0)$ it also holds that $(\tilde{\nu_3}\cdots\tilde{\nu_N})\neq(0\cdots0)$ \cite{Hoehn:2014uua}, however, the precise values of the $\tilde{\nu}_i$ will not matter. There exists a $3$ \texttt{bit} state in which only these three questions are answered with certainty, while for all other Bloch vector components $r_i=0$. Namely, after asking only $Q_{\nu_1\cdots\nu_N},Q_{00\nu_3'\cdots\nu_N'}$ to a system $S$ in the state of no information, $O$ will have certain information about these two questions and their correlation, however, will not know anything about any further question in the informationally complete set. We shall work with such $3$ \texttt{bit} states henceforth.

Thanks to the form of (\ref{T12}), the component $r_{00\nu_3'\cdots\nu_N'}=\pm1$ of the Bloch vector $\vec{r}$ corresponding to such a state is left invariant under the time evolution $\vec{r}\,':=T^{(12)}\cdot \vec{r}$, i.e.\ $r'_{00\nu_3'\cdots\nu_N'}=r_{00\nu_3'\cdots\nu_N'}=\pm1$. The complementarity inequalities (\ref{compstrong}) therefore imply that
\ba
0=r'_{\mu_1\cdots\mu_N}=\sum_{\beta_i} \tilde{T}^{(12)}_{(\mu_1\cdots\mu_N)(\beta_1\cdots\beta_N)}r_{\beta_1\cdots\beta_N}\nn
\ea
since $Q_{00\nu_3'\cdots\nu_N'}$ was chosen complementary to $Q_{\mu_1\cdots\mu_N}$. Given that $r_{00\nu_3'\cdots\nu_N'}=\pm1$ is left invariant and thus only the $r_{\nu_1\cdots\nu_N}, r_{{\nu_1}\nu_2\tilde{\nu}_3\cdots\tilde{\nu_N}}\in\{-1,+1\}$ can contribute (recall that all other $r_i=0$), the previous equation reduces to:\footnote{We note that $\tilde{T}^{(12)}_{(\mu_1\cdots\mu_N)({\nu_1}\nu_2\tilde{\nu}_3\cdots\tilde{\nu_N})}$ is not necessarily zero since $(\nu_1\nu_2)\neq(00)$ and $(\tilde{\nu}_3\cdots\tilde{\nu}_N)\neq(0\cdots0)$.}
\ba
0=r'_{\mu_1\mu_2\mu_3\cdots\mu_N}=\tilde{T}^{(12)}_{(\mu_1\cdots\mu_N)(\nu_1\cdots\nu_N)}r_{\nu_1\cdots\nu_N}+\tilde{T}^{(12)}_{(\mu_1\cdots\mu_N)({\nu_1}\nu_2\tilde{\nu}_3\cdots\tilde{\nu_N})}r_{{\nu_1}\nu_2\tilde{\nu}_3\cdots\tilde{\nu_N}} \label{eq:TN_st2a}
\ea
(no further summation over $\nu_i$ or $\tilde{\nu}_j$). Consider now the two specific configurations\footnote{Both are allowed since $Q_{\nu_1\nu_2\nu_3\cdots\nu_N},Q_{\nu_1\nu_2\tilde{\nu}_3\cdots\tilde{\nu}_N}$ are pairwise independent \cite{Hoehn:2014uua}.} (a) $r_{\nu_1\nu_2\nu_3\cdots\nu_N}=r_{{\nu_1}{\nu_2}\tilde{\nu_3}\cdots\tilde{\nu_N}}=1$ and (b) $r_{\nu_1\nu_2\nu_3\cdots\nu_N}=1, r_{{\nu_1}{\nu_2}\tilde{\nu_3}\cdots\tilde{\nu_N}}=-1$. (\ref{eq:TN_st2a}) must hold true for {\it both} (a) and (b) which is only possible if $\tilde{T}^{(12)}_{(\mu_1\mu_2\mu_3\cdots\mu_N)(\nu_1\nu_2\nu_3\cdots\nu_N)}=0=\tilde{T}^{(12)}_{(\mu_1\mu_2\mu_3\cdots\mu_N)(\tilde{\nu_1}\tilde{\nu_2}\tilde{\nu_3}\cdots\tilde{\nu_N})}$.

In this argument it was crucial that the invariant $Q_{00\nu_3'\cdots\nu_N'}$ was complementary to $Q_{\mu_1\cdots\mu_N}$ and compatible with $Q_{\nu_1\cdots\nu_N}$. Clearly, no such $Q_{00\nu_3'\cdots\nu_N'}$ with this property could exist if we had $(\mu_3\cdots\mu_N)=(\nu_3\cdots\nu_N)$. Hence, all that remains to be checked is whether we can always find a $Q_{00\nu_3'\cdots\nu_N'}$ with this property if $\mu_3\neq\nu_3$. By considering all the possible cases this can easily be shown to be true. For ease of notation, let us denote the relevant question as $Q^*:=Q_{00\nu_3'\cdots\nu_N'}$. First, for $N=3$ we must have $\mu_3,\nu_3\neq 0$ in order for $\tilde{T}^{(12)}$ not to vanish and we can choose $Q^*=Q_{00\nu_3}$. For $N>3$, we choose the question $Q^*$ according to the two cases where the indices $(\mu_4\cdots\mu_N)$ and $(\nu_4\cdots\nu_N)$ differ in either an {\it odd} or {\it even} amount of non-zero indices cases (we remind the reader that $\mu_3\neq\nu_3$).
\begin{itemize}
\item {\it Odd} number of differing non-zero indices such that $Q_{000\mu_4\cdots\mu_N}$ and $Q_{000\nu_4\cdots\nu_N}$ are complementary: take $Q^*=Q_{000\nu_4\cdots\nu_N}$.
\item {\it Even} number of differing non-zero indices such that $Q_{000\mu_4\cdots\mu_N}$ and $Q_{000\nu_4\cdots\nu_N}$ are compatible: 
\begin{itemize}
\item $\mu_3\neq0$: take $Q^*=Q_{00\nu_3\nu_4\cdots\nu_N}$ if $\nu_3\neq0$ or $Q^*=Q_{00\nu_3'\nu_4\cdots\nu_N}$, where any $\nu'_3\neq\mu_3$ suffices, if $\nu_3=0$.
\item $\mu_3=0$ (and thus $\nu_3\neq0$) and without loss of generality we assume that $\mu_4\neq0$ since there {\it must} be a non-zero index among $\mu_4,\ldots,\mu_N$: (i) if $\nu_4\neq0$, take $Q^*=Q_{00\nu_3'\nu_4'\nu_5\cdots\nu_N}$ with $\nu_4'\neq \mu_4$, and also $(\nu_3\nu_4)$ and $(\nu_3'\nu_4')$ differ in an even amount of non-zero indices\footnote{This comes down to the question if, given any two questions $Q_{0\mu_4}$ and $Q_{\nu_3\nu_4}$ where $\nu_3,\nu_4\neq0$, there is a third question which is complementary to $Q_{0\mu_4}$ {\it and} compatible with $Q_{\nu_3\nu_4}$. This is always possible \cite{Hoehn:2014uua}.}, (ii) if $\nu_4=0$ take $Q^*=Q_{000\nu_4'\nu_5\cdots\nu_N}$, where any $\nu'_4\neq\mu_4$ suffices.
\end{itemize}
\end{itemize}
We thus conclude that $\tilde{T}^{(12)}_{(\mu_1\mu_2\mu_3\cdots\mu_N)(\nu_1\nu_2\nu_3\cdots\nu_N)}$ vanishes unless $(\mu_3\cdots\mu_N)\equiv(\nu_3\cdots\nu_N)$ and thus $\tilde{T}^{(12)}_{(\mu_1\mu_2\mu_3\cdots\mu_N)(\nu_1\nu_2\nu_3\cdots\nu_N)}\sim \delta_{\mu_3\nu_3}\cdots\delta_{\mu_N\nu_N}$. The factor multiplying the delta's might depend on either the indices $\nu_3,\ldots,\nu_N$ or $\mu_3,\ldots,\mu_N$ which are fixed to be equal.
\end{proof}

It follows from lemma \ref{lem_6} that the block-matrix in the generators (\ref{gennl2}) is of the form
\ba
g^{(12)}_{(\mu_1\mu_2\mu_3\cdots\mu_N)(\nu_1\nu_2\nu_3\cdots\nu_N)}=\tilde{G}_{(\mu_1\mu_2)(\nu_1\nu_2)}(\mu_3,\ldots,\mu_N)\,\delta_{\mu_3\nu_3}\cdots\delta_{\mu_N\nu_N}\label{gennl2b}
\ea
with $\tilde{G}_{(00)(\nu_1\nu_2)}(\mu_3,\ldots,\mu_N)=0=\tilde{G}_{(\mu_1\mu_2)(00)}(\mu_3,\ldots,\mu_N)$. Note that
\ba
g^{(12)}_{(\mu_1\cdots\mu_N)(\nu'_1\cdots\nu'_N)}\,g^{(12)}_{(\nu'_1\cdots\nu'_N)(\nu_1\cdots\nu_N)}&=&\tilde{G}_{(\mu_1\mu_2)(\nu_1'\nu_2')}(\mu_3,\ldots,\mu_N)\,\delta_{\mu_3\nu_3'}\cdots\delta_{\mu_N\nu_N'} \times \nonumber\\
&&\tilde{G}_{(\nu_1'\nu_2')(\nu_1\nu_2)}(\nu_3'\cdots\nu_N')\,\delta_{\nu_3'\nu_3}\cdots\delta_{\nu_N'\nu_N}\nn\\
&=&\!\!\tilde{G}_{(\mu_1\mu_2)(\nu'_1\nu'_2)}(\mu_3,\ldots,\mu_N)\,\tilde{G}_{(\nu'_1\nu'_2)(\nu_1\nu_2)}(\mu_3,\ldots,\mu_N)\,\delta_{\mu_3\nu_3}\cdots\delta_{\mu_N\nu_N}\nn
\ea
and similarly for the higher powers of $g^{(12)}$ and therefore $M(t)=\exp(t\,\tilde{G})$ for $M$ given in lemma \ref{lem_6}. We are now interested in the representation of the pentagon swap generators corresponding to (\ref{swapgen1}, \ref{swapgen2}) on $\mathbb{R}^{4^N-1}$
\ba
G^{\text{Pent}_a^{(12)}, \text{Pent}_b^{(12)}} = \left(\begin{array}{ccc}{G}^{\text{Pent}_a, \text{Pent}_b} & 0 & 0 \\0 & g^{\text{Pent}_a^{(12)}, \text{Pent}_b^{(12)}} & 0 \\0 & 0 & 0_{(4^{N-2}-1)\times(4^{N-2}-1)}\end{array}\right),\label{gennl2d}
\ea
where ${G}^{\text{Pent}_a, \text{Pent}_b}$ is one of the 15 two-qubit swap generators in (\ref{swapgen1}, \ref{swapgen2}) and by (\ref{gennl2b})
\ba
g^{\text{Pent}_a^{(12)}, \text{Pent}_b^{(12)}}_{(\mu_1\mu_2\mu_3\cdots\mu_N)(\nu_1\nu_2\nu_3\cdots\nu_N)}=\tilde{G}^{\text{Pent}_a, \text{Pent}_b}_{(\mu_1\mu_2)(\nu_1\nu_2)}(\mu_3,\ldots,\mu_N)\,\delta_{\mu_3\nu_3}\cdots\delta_{\mu_N\nu_N}\label{gennl2c}.
\ea

\begin{lemma}\label{lem_7}
\leavevmode
$\tilde{G}^{\text{Pent}_a, \text{Pent}_b}_{(\mu_1\mu_2)(\nu_1\nu_2)}(\mu_3,\ldots,\mu_N)=0$ in (\ref{gennl2c}) if $Q_{\mu_1\mu_2}$ or $Q_{\nu_1\nu_2}$ is a question whose Bloch vector component is preserved under the two-qubit evolutions generated by the $G^{\text{Pent}_a,\text{Pent}_b}$.
\end{lemma}

\begin{proof}
It is instructive to consider a specific example, say, 
$G^{\text{Pent}_1,\text{Pent}_2}$ which, as seen in figure \ref{fig:pentagon}, preserves $r_{x_1},r_{x_2},r_{xx},r_{yy},r_{zz},r_{yz},r_{zy}$. 

Next, notice that $Q_{x_100\cdots},Q_{0x_20\cdots},Q_{xx0\cdots0}$, $Q_{00\mu_3\cdots\mu_N}$ for $(\mu_3\cdots\mu_N)\neq(0,\cdots0)$ are pairwise compatible since the indices of the questions disagree in none of the non-zero indices \cite{Hoehn:2014uua}. In fact, by theorem 3.1 in \cite{Hoehn:2014uua} (`Specker's principle'), they must also be mutually compatible such that there must exist a state in which the answers to all of these questions are known with certainty to $O$. For example, $r_{x_100\cdots0}=r_{0x_20\cdots0}=r_{xx0\cdots0}=r_{00\mu_3\cdots\mu_N}=+1$ and therefore, due to the XNOR properties, also $r_{x_10\mu_3\cdots\mu_N}=r_{0x_2\mu_3\cdots\mu_N}=r_{xx\mu_3\cdots\mu_N}=+1$ and all other $r_i=0$ must exist. This is a $7$ \texttt{bits} state. By construction, $T^{\text{Pent}_1^{(12)}, \text{Pent}_2^{(12)}}(t)=\exp(t\,G^{\text{Pent}_1^{(12)}, \text{Pent}_2^{(12)}})$ leaves the components $r_{x_100\cdots0}=r_{0x_20\cdots0}=r_{xx0\cdots0}=r_{00\mu_3\cdots\mu_N}=+1$ invariant. Consequently, $T^{\text{Pent}_1^{(12)}, \text{Pent}_2^{(12)}}(t)$ must also leave $r_{x_10\mu_3\cdots\mu_N}=r_{0x_2\mu_3\cdots\mu_N}=r_{xx\mu_3\cdots\mu_N}=+1$ invariant since these components are implied by $r_{x_100\cdots0}=r_{0x_20\cdots0}=r_{xx0\cdots0}=r_{00\mu_3\cdots\mu_N}=+1$. Furthermore, since time evolution cannot change the total information, also $r_i=0$ for all other components must be preserved. That is, $T^{\text{Pent}_1^{(12)}, \text{Pent}_2^{(12)}}(t)$ must leave this state invariant for all $t$. The above arguments and their conclusion are independent of the signs of the non-zero Bloch vector components. In other words, the time evolution must leave for example the following two states also invariant\footnote{As before, the XNOR properties dictate the sign of the other non-zero Bloch components as (1) $r_{x_10\mu_3\cdots\mu_N}=1,r_{0x_2\mu_3\cdots\mu_N}=r_{xx\mu_3\cdots\mu_N}=r_{xx0\cdots0}=-1$ and (2) $r_{x_10\mu_3\cdots\mu_N}=r_{xx\mu_3\cdots\mu_N}=r_{xx0\cdots0}=-1,r_{0x_2\mu_3\cdots\mu_N}=1$. The remaining components are $r_i=0$.}: (1) $r_{x_100\cdots0}=r_{00\mu_3\cdots\mu_N}=+1,r_{0x_20\cdots0}=-1$ and (2) $r_{x_100\cdots0}=-1,r_{0x_20\cdots0}=r_{00\mu_3\cdots\mu_N}=+1$. This is only possible if 
\ba
M^{\text{Pent}_1, \text{Pent}_2}_{(x_10)(x_10)}((\mu_3,\ldots,\mu_N);t)=M^{\text{Pent}_1, \text{Pent}_2}_{(0x_2)(0x_2)}((\mu_3,\ldots,\mu_N);t)= M^{\text{Pent}_1, \text{Pent}_2}_{(xx)(xx)}((\mu_3,\ldots,\mu_N);t)=1\nn
\ea
and 
\ba
M^{\text{Pent}_1, \text{Pent}_2}_{(\mu_1\mu_2)(x_10)}((\mu_3,\ldots,\mu_N);t)\equiv M^{\text{Pent}_1, \text{Pent}_2}_{(\mu_1\mu_2)(0x_2)}((\mu_3,\ldots,\mu_N);t)\equiv M^{\text{Pent}_1, \text{Pent}_2}_{(\mu_1\mu_2)(xx)}((\mu_3,\ldots,\mu_N);t)\equiv0\nn
\ea
for all $t$ and whenever $(\mu_1\mu_2)$ is neither of $(x_10),(0x_2),(xx)$ respectively. But this is only possible if $\tilde{G}^{\text{Pent}_1, \text{Pent}_2}_{(\mu_1\mu_2)(x_10)}(\mu_3,\ldots,\mu_N)\equiv \tilde{G}^{\text{Pent}_1, \text{Pent}_2}_{(\mu_1\mu_2)(0x_2)}(\mu_3,\ldots,\mu_N)\equiv \tilde{G}^{\text{Pent}_1, \text{Pent}_2}_{(\mu_1\mu_2)(xx)}(\mu_3,\ldots,\mu_N)\equiv0$ for all $\mu_1,\mu_2$.

By means of an analogous state one can show similarly that $\tilde{G}^{\text{Pent}_1, \text{Pent}_2}_{(\mu_1\mu_2)(yy)}(\mu_3,\ldots,\mu_N)\equiv \tilde{G}^{\text{Pent}_1, \text{Pent}_2}_{(\mu_1\mu_2)(zz)}(\mu_3,\ldots,\mu_N)\equiv \tilde{G}^{\text{Pent}_1, \text{Pent}_2}_{(\mu_1\mu_2)(yz)}(\mu_3,\ldots,\mu_N)\equiv \tilde{G}^{\text{Pent}_1, \text{Pent}_2}_{(\mu_1\mu_2)(zy)}(\mu_3,\ldots,\mu_N)\equiv0$ for all $\mu_1,\mu_2$.

One argues in complete analogy for all other $\tilde{G}^{\text{Pent}_a,\text{Pent}_b}$. Using the anti-symmetry of $\tilde{G}$ one finds the claimed result.
\end{proof}

We have thus shown that $\tilde{G}^{\text{Pent}_a, \text{Pent}_b}_{(\mu_1\mu_2)(\nu_1\nu_2)}(\mu_3,\ldots,\mu_N)$ could only be non-zero if both questions $Q_{\mu_1\mu_2},Q_{\nu_1\nu_2}$ are among the eight questions whose information content is swapped under the swaps corresponding to ${G}^{\text{Pent}_a, \text{Pent}_b}$. We shall now strengthen this result further.

\begin{lemma}\label{lem_8}
\leavevmode
$\tilde{G}^{\text{Pent}_a, \text{Pent}_b}_{(\mu_1\mu_2)(\nu_1\nu_2)}(\mu_3,\ldots,\mu_N)\equiv{G}^{\text{Pent}_a, \text{Pent}_b}_{(\mu_1\mu_2)(\nu_1\nu_2)}$ for all $(\mu_3\cdots\mu_N)$, where ${G}^{\text{Pent}_a, \text{Pent}_b}_{(\mu_1\mu_2)(\nu_1\nu_2)}$ is one of the 15 two-qubit swap generators (\ref{swapgen1}, \ref{swapgen2}), and we define $G^{\text{Pent}_a,\text{Pent}_b}_{(00)(\nu_1\nu_2)}:=0=:G^{\text{Pent}_a,\text{Pent}_b}_{(\mu_1\mu_2)(00)}$.
\end{lemma}

\begin{proof}
For concreteness, consider, again, $\tilde{G}^{\text{Pent}_1, \text{Pent}_2}$. \\~
{\bf(a)} We firstly argue that $\tilde{G}^{\text{Pent}_1, \text{Pent}_2}_{(\mu_1\mu_2)(\nu_1\nu_2)}(\mu_3,\ldots,\mu_N)=0$ if ${G}^{\text{Pent}_1, \text{Pent}_2}_{(\mu_1\mu_2)(\nu_1\nu_2)}=0$. To this end, consider a state with $r_{xy0\cdots0}=r_{zx0\cdots0}=r_{zx\mu_3\cdots\mu_N}=+1$ for $(\mu_3\cdots\mu_N)\neq(0\cdots0)$. Such a state must exist since $Q_{xy0\cdots0},Q_{zx0\cdots0},Q_{zx\mu_3\cdots\mu_N}$ are compatible and pairwise independent. (Recall that two questions are compatible iff they disagree in an even number (including zero) of non-zero indices \cite{Hoehn:2014uua}.) By theorem 3.1 in \cite{Hoehn:2014uua} (`Specker's principle'), these are also mutually compatible such that a state must exist in which the answers to these questions are fully known to $O$. Furthermore, since by figure \ref{fig_corr} $Q_{xy}\leftrightarrow Q_{zx}=\neg Q_{yz}$ we must also have $r_{yz}=r_{yz\mu_3\cdots\mu_N}=-1$ and, similarly, $r_{00\mu_3\cdots\mu_N}=r_{xy\mu_3\cdots\mu_N}=+1$. For all other components, we may have $r_i=0$.

Consider now $T^{\text{Pent}^{(12)}_1, \text{Pent}^{(12)}_2}(t)=\exp(t\,G^{\text{Pent}^{(12)}_1, \text{Pent}^{(12)}_2})$ acting on this state. By construction, $r_{yz}=r_{yz\mu_3\cdots\mu_N}=-1$ and $r_{00\mu_3\cdots\mu_N}=+1$ are left invariant (the first two since $Q_{yz}$ is contained in neither of $\text{Pent}_1,\text{Pent}_2$ and thanks to lemma \ref{lem_7}). Furthermore, it follows from appendix \ref{app_swapmain} that $T^{\text{Pent}^{(12)}_1, \text{Pent}^{(12)}_2}(t)$ preserves the pentagon identities (\ref{pentin}) at the two-qubit level. Given that $T^{\text{Pent}^{(12)}_1, \text{Pent}^{(12)}_2}(t)$ transfers information within the pairs $Q_{xy0\cdots0},Q_{0z_20\cdots0}$ and $Q_{zx0\cdots0},Q_{y_100\cdots0}$ (see figure \ref{fig:pentagon}) and given the state above, it is clear that 
\ba
r_{0z_20\cdots0}^2(t)+r_{zx0\cdots0}^2(t)=1\label{xyz}
\ea
must thus hold for all $t\in\mathbb{R}$ under $T^{\text{Pent}^{(12)}_1, \text{Pent}^{(12)}_2}(t)$ acting on our initial state. 

Next, we note that $Q_{0z_20\cdots0},Q_{zx0\cdots0},Q_{yx\mu_3\cdots\mu_N}$ form a mutually complementary set. Hence, by (\ref{compstrong}), it must always hold $r_{0z_20\cdots0}^2(t)+r_{zx0\cdots0}^2(t)+r_{yx\mu_3\cdots\mu_N}(t)\leq 1$ and thanks to (\ref{xyz}) therefore also $r_{yx\mu_3\cdots\mu_N}(t)=0$ for all $t\in\mathbb{R}$. Given the behaviour of our state under $T^{\text{Pent}^{(12)}_1, \text{Pent}^{(12)}_2}(t)$, by lemma \ref{lem_6} we must therefore have
\ba
r_{yx\mu_3\cdots\mu_N}(t)&=&M^{\text{Pent}_1, \text{Pent}_2}_{(yx)(zx)}((\mu_3,\ldots,\mu_N);t)\,r_{zx\mu_3\cdots\mu_N}+M^{\text{Pent}_1, \text{Pent}_2}_{(yx)(xy)}((\mu_3,\ldots,\mu_N);t)\,r_{xy\mu_3\cdots\mu_N}\nn\\
&=&M^{\text{Pent}_1, \text{Pent}_2}_{(yx)(zx)}((\mu_3,\ldots,\mu_N);t)+M^{\text{Pent}_1, \text{Pent}_2}_{(yx)(xy)}((\mu_3,\ldots,\mu_N);t)\overset{!}{=}0,\q\q\forall\,t\in\mathbb{R}.\nn
\ea
Repeating the same steps with the initial state $r_{xy0\cdots0}=r_{zx\mu_3\cdots\mu_N}=r_{yz0\cdots0}=r_{yz\mu_3\cdots\mu_N}=+1$, $r_{00\mu_3\cdots\mu_N}=r_{xy\mu_3\cdots\mu_N}=-1$ (and all other $r_i=0$), one concludes that also 
\ba
M^{\text{Pent}_1, \text{Pent}_2}_{(yx)(zx)}((\mu_3,\ldots,\mu_N);t)-M^{\text{Pent}_1, \text{Pent}_2}_{(yx)(xy)}((\mu_3,\ldots,\mu_N);t)\overset{!}{=}0,\q\q\forall\,t\in\mathbb{R},\nn
\ea
such that 
\ba
M^{\text{Pent}_1, \text{Pent}_2}_{(yx)(zx)}((\mu_3,\ldots,\mu_N);t)=M^{\text{Pent}_1, \text{Pent}_2}_{(yx)(xy)}((\mu_3,\ldots,\mu_N);t)\overset{!}{=}0,\q\q\forall\,t\in\mathbb{R}.\nn
\ea
But this can only be true if also 
\ba
\tilde{G}^{\text{Pent}_1, \text{Pent}_2}_{(yx)(zx)}(\mu_3,\ldots,\mu_N)=\tilde{G}^{\text{Pent}_1, \text{Pent}_2}_{(yx)(xy)}(\mu_3,\ldots,\mu_N)=0.\nn
\ea
These components also vanish for $G^{\text{Pent}_1,\text{Pent}_2}$ at the two-qubit level (\ref{swapgen1}). By complete analogy one shows that also for all other cases $\tilde{G}^{\text{Pent}_1, \text{Pent}_2}_{(\mu_1\mu_2)(\nu_1\nu_2)}(\mu_3,\ldots,\mu_N)=0$ if ${G}^{\text{Pent}_1, \text{Pent}_2}_{(\mu_1\mu_2)(\nu_1\nu_2)}=0$. 
\\~
{\bf(b)} Secondly, we now show $\tilde{G}^{\text{Pent}_1, \text{Pent}_2}_{(\mu_1\mu_2)(\nu_1\nu_2)}(\mu_3,\ldots,\mu_N)\equiv{G}^{\text{Pent}_1, \text{Pent}_2}_{(\mu_1\mu_2)(\nu_1\nu_2)}$. For this purpose, consider again the state above. Under (a) we have just shown that $\tilde{G}^{\text{Pent}_1, \text{Pent}_2}_{(\mu_1\mu_2)(\nu_1\nu_2)}(\mu_3,\ldots,\mu_N)\neq0$ is only possible if ${G}^{\text{Pent}_1, \text{Pent}_2}_{(\mu_1\mu_2)(\nu_1\nu_2)}\neq0$. This means that $M^{\text{Pent}_1, \text{Pent}_2}((\mu_3,\ldots,\mu_N);t)=\exp(t\,\tilde{G}^{\text{Pent}_1, \text{Pent}_2})$ could at most transfer information within the pairs $(Q_{y_10\mu_3\cdots\mu_N},Q_{zx\mu_3\cdots\mu_N})$, $(Q_{xy\mu_3\cdots\mu_N},Q_{0z_2\mu_3\cdots\mu_N}),$ $(Q_{z_10\mu_3\cdots\mu_N},Q_{yx\mu_3\cdots\mu_N})$ and $(Q_{xz\mu_3\cdots\mu_N},Q_{0y_2\mu_3\cdots\mu_N})$ for $(\mu_3\cdots\mu_N)\neq(0\cdots0)$. But since the total information must be preserved this implies that
\ba
r^2_{xy\mu_3\cdots\mu_N}(t)+r^2_{0z_2\mu_3\cdots\mu_N}(t)=1,\q\q\forall\,t\in\mathbb{R},\label{xyza}
\ea
must hold for $\vec{r}(t)=T^{\text{Pent}^{(12)}_1, \text{Pent}^{(12)}_2}(t)\,\vec{r}(0)$, where $\vec{r}(0)$ is our initial state above. Similarly, from the pentagon equalities (\ref{pentin}) it follows for the time evolution of this state that also
\ba
r^2_{xy0\cdots0}(t)+r^2_{0z_20\cdots0}(t)=1,\q\q\forall\,t\in\mathbb{R}.\label{xyzb}
\ea
From the complementarity inequalities (\ref{compstrong}) it must also hold
\ba
r_{xy0\cdots0}^2(t)+r_{0z_2\mu_3\cdots\mu_N}^2(t)\leq1,\q\q\q r^2_{0z_20\cdots0}(t)+r^2_{xy\mu_3\cdots\mu_N}(t)\leq1,\q\q\forall\,t\in\mathbb{R}.\nn
\ea
From adding up (\ref{xyza}, \ref{xyzb}) it, in fact, follows, that these inequalities must be saturated:
\ba
r_{xy0\cdots0}^2(t)+r_{0z_2\mu_3\cdots\mu_N}^2(t)=1,\q\q\q r^2_{0z_20\cdots0}(t)+r^2_{xy\mu_3\cdots\mu_N}(t)=1,\q\q\forall\,t\in\mathbb{R}.\nn
\ea
This implies that for the time evolution of our initial state,
\ba
\left(\begin{array}{c}r_{0z_20\cdots0}(t) \\r_{xy0\cdots0}(t)\end{array}\right)=\left(\begin{array}{c}s_1\,r_{0z_2\mu_3\cdots\mu_N}(t) \\s_2\,r_{xy\mu_3\cdots\mu_N}(t)\end{array}\right),\q\q\forall\,t\in\mathbb{R},\nn
\ea
where $s_1,s_2$ are two signs to be determined. From the state at $t=0$, however, we know that $s_2=+1$. Furthermore, we noted above that $r_{00\mu_3\cdots\mu_N}=+1$ is invariant under $T^{\text{Pent}^{(12)}_1, \text{Pent}^{(12)}_2}(t)$. But this implies that whenever $r_{0z_20\cdots0}(t)=\pm1$, we must also have $r_{0z_2\mu_3\cdots\mu_N}=\pm1$ since $Q_{0z_2\mu_3\cdots\mu_N}=Q_{0z_20\cdots0}\leftrightarrow Q_{00\mu_3\cdots\mu_N}$. This entails also $s_1=+1$ and therefore
\ba
\left(\begin{array}{c}r_{0z_20\cdots0}(t) \\r_{xy0\cdots0}(t)\end{array}\right)=\left(\begin{array}{c}r_{0z_2\mu_3\cdots\mu_N}(t) \\r_{xy\mu_3\cdots\mu_N}(t)\end{array}\right),\q\q\forall\,t\in\mathbb{R}.\nn
\ea
This is only possible if, indeed, $\tilde{G}^{\text{Pent}_1, \text{Pent}_2}_{(xy)(0z_2)}(\mu_3,\ldots,\mu_N)\equiv{G}^{\text{Pent}_1, \text{Pent}_2}_{(xy)(0z_2)}$ for {\it all} values of the indices $\mu_3,\ldots,\mu_N$. By completely analogous reasoning, it follows for all other components that $\tilde{G}^{\text{Pent}_1, \text{Pent}_2}_{(\mu_1\mu_2)(\nu_1\nu_2)}(\mu_3,\ldots,\mu_N)\equiv{G}^{\text{Pent}_1, \text{Pent}_2}_{(\mu_1\mu_2)(\nu_1\nu_2)}$. This implies that $\tilde{G}^{\text{Pent}_1, \text{Pent}_2}_{(\mu_1\mu_2)(\nu_1\nu_2)}(\mu_3,\ldots,\mu_N)$ only depends on its indices $(\mu_1\mu_2)$ and $(\nu_1\nu_2)$ and thus it can be interpreted as a proper $16\times16$ matrix.

Finally, using similar states and arguments, one shows that, in generality the above also holds for the other pair of pentagon indices, $\tilde{G}^{\text{Pent}_a, \text{Pent}_b}_{(\mu_1\mu_2)(\nu_1\nu_2)}(\mu_3,\ldots,\mu_N)\equiv{G}^{\text{Pent}_a, \text{Pent}_b}_{(\mu_1\mu_2)(\nu_1\nu_2)}$, for all $a,b=1,\ldots,6$.
\end{proof}

Lemmas \ref{lem_6}--\ref{lem_8}, together with (\ref{gennl2d}, \ref{gennl2c}), thus indeed give the desired result (\ref{psu4adj}). It is also clear that (\ref{psu4adj}) generate a (reducible) representation of $\ct_2^{(12)}\simeq\rm{PSU}(4)$ on $\mathbb{R}^{4^N-1}$.

\subsubsection{Quantum theory generators of pairwise unitaries for $N>2$ qubits in the adjoint representation}\label{app_QTswap}

Here we shall argue that in the adjoint representation, the fundamental generators of the $\PSU(4)$ subgroup of $\PSU(2^N)$ that involves all time evolutions of the subsystem made up of qubits $1$ and $2$ are of the following form:
\begin{eqnarray}
G^{(\omega_1\omega_2 0\cdots0)}_{(\mu_1\cdots\mu_N)(\nu_1\cdots\nu_N)}:&=&f^{(\omega_1\omega_2 0\cdots0)}_{(\mu_1\cdots\mu_N)(\nu_1\cdots\nu_N)}=\f{1}{2^N}\,\text{tr}[[\sigma_{\omega_1\omega_2 0\cdots0},\sigma_{\mu_1\cdots\mu_N}]\,\sigma_{\nu_1\cdots\nu_N}] \nonumber\\
&=& f^{(\omega_1\omega_2)}_{(\mu_1\mu_2)(\nu_1\nu_2)}\delta_{\mu_3\nu_3}\cdots\delta_{\mu_N\nu_N}, \label{eq:comalgN>2} 
\end{eqnarray}
where $f^{(\omega_1\omega_2)}_{(\mu_1\mu_2)(\nu_1\nu_2)}$ are the generators of $\PSU(4)$ in the adjoint representation corresponding to the $4\times4$ Pauli operators as given in (\ref{eq:comalg}). The above generalizes trivially to the generators of the $\PSU(4)$ time evolution subgroup of the subsystem of any pair of qubits $i$ and $j$. The Pauli operators are $\sigma_{\mu_1\cdots\mu_N}=(\sigma_{\mu_1}\otimes\cdots\otimes\sigma_{\mu_N})$ and satisfy $\text{tr}[\sigma_{\mu_1\cdots\mu_N}\cdot\sigma_{\nu_1\cdots\nu_N}]=2^N\delta_{\mu_1\nu_1}\cdots\delta_{\mu_N\nu_N}$. Working out the trace in (\ref{eq:comalgN>2}) and using the tensor property $\text{tr}[A\otimes B]=\text{tr}[A]\text{tr}[B]$ results in:
\begin{gather}
\f{1}{2^N}\,\text{tr}[(\sigma_{\omega_1}\cdot\sigma_{\mu_1}\cdot\sigma_{\nu_1})\otimes(\sigma_{\omega_2}\cdot\sigma_{\mu_2}\cdot\sigma_{\nu_2})\otimes(\sigma_{\mu_3}\cdot\sigma_{\nu_3})\otimes\cdots\otimes(\sigma_{\mu_N}\cdot\sigma_{\nu_N}) \nonumber\\
-(\sigma_{\mu_1}\cdot\sigma_{\omega_1}\cdot\sigma_{\nu_1})\otimes(\sigma_{\mu_2}\cdot\sigma_{\omega_2}\cdot\sigma_{\nu_2})\otimes(\sigma_{\mu_3}\cdot\sigma_{\nu_3})\otimes\cdots\otimes(\sigma_{\mu_N}\cdot\sigma_{\nu_N})] \nonumber\\
=\f{1}{2^N}\,\text{tr}[(\sigma_{\omega_1}\cdot\sigma_{\mu_1}\cdot\sigma_{\nu_1})\otimes(\sigma_{\omega_2}\cdot\sigma_{\mu_2}\cdot\sigma_{\nu_2})-(\sigma_{\mu_1}\cdot\sigma_{\omega_1}\cdot\sigma_{\nu_1})\otimes(\sigma_{\mu_2}\cdot\sigma_{\omega_2}\cdot\sigma_{\nu_2})] \nonumber\\
\times\text{tr}[\sigma_{\mu_3}\cdot\sigma_{\nu_3}]\cdots\text{tr}[\sigma_{\mu_N}\cdot\sigma_{\nu_N}] \nonumber\\
=\f{1}{2^N}\,(4f^{(\omega_1\omega_2)}_{(\mu_1\mu_2)(\nu_1\nu_2)})(2\delta_{\mu_3\nu_3})\cdots(2\delta_{\mu_N\nu_N})=f^{(\omega_1\omega_2)}_{(\mu_1\mu_2)(\nu_1\nu_2)}\delta_{\mu_3\nu_3}\cdots\delta_{\mu_N\nu_N}.\nn
\end{gather}

We noted before in appendix \ref{app_swap} that the two-qubit adjoint generators $(G^{(\omega_1\omega_2)})_{(\mu_1\mu_2)(\nu_1\nu_2)}$ $:=f^{(\omega_1\omega_2)}_{(\mu_1\mu_2)(\nu_1\nu_2)}$ of quantum theory coincide with the swap generators (\ref{swapgen1}, \ref{swapgen2}) of the reconstruction. Using the correspondence $Q_{\mu_1\mu_2}\longleftrightarrow \sigma_{\mu_1\mu_2}:=\sigma_{\mu_1}\otimes\sigma_{\mu_2}$ with $\sigma_0=\mathds{1}$, the ordering of coincidence was such that $G^{(\omega_1\omega_2)}\equiv\pm G^{\text{Pent}_a,\text{Pent}_b}$ where $Q_{\omega_1\omega_2}$ is the unique question in $\text{Pent}_a\cap\text{Pent}_b$ left invariant by the swap.\footnote{In appendix \ref{app_swap} we still used the distinct but equivalent index notation with $i,j$ labeling the questions. However, the equivalence is immediate by identifying $i:=\omega_1\omega_2$.}

But this immediately implies that also (\ref{eq:comalgN>2}) coincide with the reconstructed $\ct_2^{(12)}=\rm{PSU}(4)$ generators (\ref{psu4adj}) (see also appendix \ref{app_swapnl2}). Namely, the ordering of coincidence is such that, firstly, $Q_{\mu_1\mu_20\cdots0}$ corresponds to $\sigma_{\mu_1\mu_20\cdots0}:=\sigma_{\mu_1}\otimes\sigma_{\mu_2}\otimes\mathds{1}\otimes\cdots\otimes\mathds{1}$ and, secondly, $G^{\text{Pent}_a^{(1,2)},\text{Pent}_b^{(1,2)}}$ coincides with the adjoint representation of $\sigma_{\mu_1\mu_20\cdots0}$ corresponding to the unique question $Q_{\mu_1\mu_20\cdots0}$ in $\text{Pent}_a^{(12)}\cap\text{Pent}_b^{(12)}$. 

\subsubsection{Evolving to product states for $N>2$ in the reconstruction}\label{app_prodnl2}

Also for $N>2$ all candidate pure states can be evolved to a product form.
\begin{lemma}
Using the time evolution group $\ct_N\simeq \PSU(2^N)$, any $N$ gbit pure state $\vec{r}$ can be transformed to a state with information distribution $\alpha_{z_1}=\cdots=\alpha_{z_1\cdots z_N}=1$ \texttt{bit} and all remaining questions in the informationally complete set $\cq_{M_N}$ carrying zero \texttt{bits}.
\end{lemma}
\begin{proof}
Consider the hermitian traceless matrix $\chi:=\sum_{\mu_i} r_{\mu_1\cdots\mu_N}\sigma_{\mu_1}\otimes\cdots\otimes\sigma_{\mu_N}$, where $r_{\mu_1\cdots\mu_N}$ are the Bloch vector components relative to our question basis. In section \ref{sec_n>2} and appendices \ref{app_swapnl2} and \ref{app_QTswap}, it was shown that the representation of $\ct_N=\rm{PSU}(2^N)$, written in the Bloch vector question basis, is exactly the adjoint representation of $\SU(2^N)$ relative to a basis of Pauli operators, which are themselves the generators of the fundamental representation of $\SU(2^N)$. The ordering of coincidence of the respective generators corresponds precisely to the pairing between $r_{\mu_1\cdots\mu_N}$ and $\sigma_{\mu_1\cdots\mu_N}:=\sigma_{\mu_1}\otimes\cdots\otimes\sigma_{\mu_N}$ in $\chi$. $\chi$ thus transforms as $\chi\rightarrow U\,\chi\, U^{\dagger}$ with $U\in \SU(2^N)$ in the fundamental representation whenever $\vec{r}\rightarrow T\cdot\vec{r}$ with $T\in\ct_N$. Since $\chi$ is hermitian it is possible to diagonalize it with some matrix $U\in\SU(2^N)$, i.e., such that $\chi'=\sum_{\mu_i}\,r_{\mu_1\cdots\mu_N}'\,\sigma_{\mu_1\cdots\mu_N}$ is diagonal and $\vec{r}\,'=T\cdot\vec{r}$ with $T\in\ct_N$. The Pauli operators $\vec{\sigma}$ form a basis of all hermitian matrices and therefore {\it only} those $r_{\mu_1\cdots\mu_N}'$ which multiply the diagonal $\sigma_{\mu_1\cdots\mu_N}$'s will be non-zero and the other components of $\vec{r}$ must be zero. There are exactly $2^N-1$ of such $\sigma_{\mu_1\cdots\mu_N}$'s, namely exactly the ones where only $\sigma_z$ or $\mathds{1}$ appear in the tensor products \cite{lawrence2002mutually}, i.e.\ $\sigma_{z_1}=\sigma_z\otimes\mathds{1}\otimes\cdots\otimes\mathds{1},\sigma_{z_2}=\mathds{1}\otimes\sigma_z\otimes\mathds{1}\otimes\cdots\otimes\mathds{1},\cdots,\sigma_{z_1\cdots z_{N}}=\sigma_z\otimes\cdots\otimes\sigma_z$ and therefore only the $2^N-1$ components $r'_{z_1},\ldots,r'_{z_1\cdots z_N}$ are non-zero. If $\vec{r}$ was a pure state, then $|\vec{r}|^2=2^N-1$ \texttt{bits} and also $|\vec{r}\,'|^2=\sum_{\mu_i}(r_{\mu_1\cdots\mu_N}')^2=2^N-1$ \texttt{bits} because $\ct_N$ preserves the Bloch vector length. There are now two possibilities: (1) less than $2^N-1$ of the $(r_{z_1}',\ldots,r_{z_1\cdots z_N}')$ are non-zero. This is only possible if at least one of them has $|r_i'|>1$ and thus $\alpha'_i>1$ \texttt{bit} which is illegal such that in this case the original $\vec{r}$ could not\footnote{By construction, the time evolution group must map legal states to legal states.} have been a legal pure state. (2)  Exactly $2^N-1$ of the $(r_{z_1}',\ldots,r_{z_1\cdots z_N}')$ are non-zero. Since $\alpha_i'=(r_i)^2\leq 1$ \texttt{bit}, it follows that precisely $\alpha_i'=(r_i')^2=1$ \texttt{bit} for $i=z_1,\ldots, z_1\cdots z_N$. Hence, every legal pure state can be time evolved to a state with information distribution $\alpha_{z_1}=\cdots=\alpha_{z_1\cdots z_N}=1$ \texttt{bit}.
\end{proof}

\subsubsection{$\rm{PSU}(2^N)$ preserves all complementarity inequalities}\label{app_comppres}

In section \ref{sec_statesnl2}, we concluded that the set of states $\Sigma_N$ implied by the principles (and background assumptions) is precisely the set of (pure and mixed) $N$-qubit quantum states. We shall now check for consistency that all states in $\Sigma_N$ (and thus all quantum states) indeed satisfy the complementarity inequalities (\ref{compstrong}).\footnote{Principles \ref{lim} and \ref{pres} are trivially satisfied because all pure Bloch vectors, which are generated by the length conserving group action of $\ct_N$ on $\vec{r}_z:=\vec{\delta}_{z_1}+\cdots+\vec{\delta}_{z_1\cdots z_N}$, have a length of $2^N-1$ \texttt{bits} (corresponding to $N$ independent \texttt{bits}) and the mixed state vectors are of length smaller than $2^N-1$ \texttt{bits}, since they are convex combinations of pure state vectors.} To this end, we might as well perform the check directly in quantum theory. In particular, we recall that in the correspondence $Q_{\mu_1\cdots\mu_N}\longleftrightarrow \sigma_{\mu_1\cdots\mu_N}$ the Bloch vector description relative to our question basis and the quantum description relative to the Pauli operator basis fully coincide. Thus, in order to show that {\it all} states in $\Sigma_N$ satisfy all complementarity inequalities (\ref{compstrong}), we may show that the quantum states satisfy these equalities relative to the Pauli operator basis. 

Using our knowledge of quantum theory, we may henceforth effortlessly switch between the Bloch and hermitian representation by defining for any $\vec{r}\in\Sigma_N$ the density matrix $\rho:=(\mathds{1}+\vec{r}\cdot\vec{\sigma})/2^N:=(\mathds{1}+\sum r_{\mu_1\cdots\mu_N}\sigma_{\mu_1}\otimes\cdots\otimes\sigma_{\mu_N})/2^N$, for which the following statements will hold:

\begin{itemize}
\item For any state $\vec{r}\in\Sigma_N$, $\rho$ transforms as $\rho\rightarrow U\rho U^{\dagger}$ with $U\in\SU(2^N)$, whenever $\vec{r}\rightarrow T\cdot\vec{r}$ for some $T\in\ct_N\simeq\PSU(2^N)$.
\item The density matrix $\rho$ is positive-semidefinite and the quantum probability function $\text{tr}[\rho.(\mathds{1}+\sigma_i)/2]\in[0,1]$ for any Pauli operator $\sigma_i$.
\item For any pair of states $\vec{r},\vec{r}\,'\in\Sigma_N$ with corresponding density matrices $\rho,\rho'$ respectively, the quantum transition probability $\text{tr}[\rho.\rho']\in[0,1]$.
\end{itemize}

We begin with a lemma restricting the Bloch vector components of states featuring information solely in a single set of non-commuting Pauli operators:

\begin{lemma}\label{lem_markus}\footnote{The authors are indebted to Markus M\"uller for the proof of this lemma.}
Suppose we have a collection of $n$ traceless, $2^N\times 2^N$ hermitian and unitary matrices $\{\sigma_i\}_{i=1}^n$ that anti-commute:
\begin{equation}
   \sigma_i^\dagger=\sigma_i,\quad \sigma_i^2=\mathds{1},\quad \sigma_i\sigma_j=-\sigma_j\sigma_i \enspace (i\neq j).\nn
\end{equation}
The operator
\begin{equation}
   S=\mathds{1}+\sum_{i=1}^n r_i\sigma_i\geq 0\nn
\end{equation}
is positive-semidefinite if and only if $|r|^2 := \sum_{i=1}^n r_i^2 \leq 1$.

\end{lemma}

\begin{proof}
Consider the traceless and hermitian matrix $M:=S-\mathds{1}=\sum_i r_i\sigma_i$. Then
\begin{equation*}
   M^2 = \sum_{ij} r_i r_j \sigma_i \sigma_j = \sum_i r_i^2 \sigma_i^2 +\sum_{i<j} r_i r_j \sigma_i \sigma_j + \sum_{i>j} r_i r_j \sigma_i \sigma_j
   =|r|^2 \mathds{1}+\sum_{i<j}r_i r_j \sigma_i \sigma_j -\sum_{i>j} r_i r_j \sigma_j \sigma_i.
\end{equation*}
Exchanging the names of the variables in the last sum, $i\leftrightarrow j$, shows that both sums are actually equal, and
we get
\begin{equation*}
   M^2 = |r|^2 \mathds{1}.
\end{equation*}
It follows that every eigenvalue of the Hermitian matrix $M$ must be either $+|r|$ or $-|r|$. In fact, since $M$ is traceless, it must have {\it both} $+|r|$ or $-|r|$ as eigenvalues. The eigenvalues of the matrix $S$ are therefore $1\pm|r|$ and $S$ is positive-semidefinite if and only if $|r|\leq 1$. 
\end{proof}

A set of hermitian and traceless Pauli operators $\{\sigma_{\mu^{(1)}_1}\otimes\cdots\otimes\sigma_{\mu^{(1)}_N},\ldots,\sigma_{\mu^{(n)}_1}\otimes\cdots\otimes\sigma_{\mu^{(n)}_N}\}$ which, under $Q_{\mu_1\cdots\mu_N}\longleftrightarrow\sigma_{\mu_1\cdots\mu_N}$, correspond to a set of mutually complementary questions\footnote{We remind the reader that the questions $Q_{\mu_1\cdots\mu_N}$ and $Q_{\mu'_1\cdots\mu'_N}$ are complementary if and only if they have exactly an uneven amount of non-zero indices in which they differ.} $\{Q_{\mu^{(1)}_1\cdots\mu^{(1)}_N},\ldots,Q_{\mu^{(n)}_1\cdots\mu^{(n)}_N}\}$ will satisfy the conditions in the above lemma. The reason is that the $N=1$ qubit $2\times 2$ Pauli operators anti-commute themselves and therefore any pair of $2^N\times 2^N$ Pauli operators $P_1=\sigma_{\mu_1}\otimes\cdots\otimes\sigma_{\mu_N}, P_2=\sigma_{\mu'_1}\otimes\cdots\otimes\sigma_{\mu'_N}$, which differ in an uneven amount of non-zero indices {\it must} anti-commute as\footnote{Without loss of generality we assume that the first uneven $n$ indices are non-zero and different between the two Pauli operators.} $P_1.\cdot P_2=(\sigma_{\mu_1}\cdot \sigma_{\mu'_1})\otimes\cdots\otimes(\sigma_{\mu_n}\cdot \sigma_{\mu'_n})\otimes\cdots=(-1)^n(\sigma_{\mu'_1}\cdot \sigma_{\mu_1})\otimes\cdots\otimes(\sigma_{\mu'_n}\cdot \sigma_{\mu_n})\otimes\cdots=-P_2\cdot P_1$. Therefore, all Bloch vectors with a length of exactly $1$ \texttt{bit} whose only non-zero components $r_{\mu^{(1)}_1\cdots\mu^{(1)}_N},\ldots,r_{\mu^{(n)}_1\cdots\mu^{(n)}_N}$ correspond to a set of mutually complementary questions $\{Q_{\mu^{(1)}_1\cdots\mu^{(1)}_N},\ldots,Q_{\mu^{(n)}_1\cdots\mu^{(n)}_N}\}$ (including maximal sets) will constitute a valid quantum state because its corresponding density matrix $\rho=(\mathds{1}+\sum r_{\mu_1\cdots\mu_N}\sigma_{\mu_1}\otimes\cdots\otimes\sigma_{\mu_N})/2^N$ will be positive semi-definite as follows from the lemma above. (The same is true, of course, if the length of the vector would be less than $1$ \texttt{bit}).

Note that this lemma does not immediately imply the same for {\it arbitrary} quantum states which can also have non-zero Bloch vector components outside of just one non-commuting Pauli operator set. We shall, however, establish this generalization next.

\begin{lemma}
Every quantum state satisfies the complementarity inequalities (\ref{compstrong}) in the correspondence $Q_{\mu_1\cdots\mu_N}\longleftrightarrow\sigma_{\mu_1\cdots\mu_N}$. Equivalently, every state in $\Sigma_N$ satisfies (\ref{compstrong}).
\end{lemma}

\begin{proof}
Suppose there was a state $\vec{r}$ featuring {\it more} than $1$ \texttt{bit} of information in a set of mutually complementary questions $\{Q_{\mu_1^{(1)}\cdots\mu_N^{(1)}},\ldots,Q_{\mu_1^{(n)}\cdots\mu_N^{(n)}}\}$. This implies that the length of the Bloch vector components corresponding to those complementary questions is larger than $1$ \texttt{bit}, i.e.\ $r_c:=\sqrt{\sum_{i=1}^n r_{\mu_1^{(i)}\cdots\mu_N^{(i)}}^2}>1$. Lemma \ref{lem_markus} entails that all Bloch vectors whose only non-zero components are labeled by these indices, $\{(\mu_1^{(1)}\cdots\mu_N^{(1)}),\ldots,(\mu_1^{(n)}\cdots\mu_N^{(n)})\}$, and that are exactly of length $1$ \texttt{bit} are legal quantum states and thus also legal states in $\Sigma_N$. 
Hence, we may define the legal Bloch vector $\vec{r}\,'=-\sum_{i=1}^n r_{\mu_1^{(i)}\cdots\mu_N^{(i)}}\vec{\delta}_{\mu_1^{(i)}\cdots\mu_N^{(i)}}/r_c\in\Sigma_N$ of length $1$ \texttt{bit}, which corresponds to a legal quantum state. The transition probability $\text{tr}[\rho\cdot\rho']=\text{tr}[(\mathds{1}+\vec{r}\cdot\vec{\sigma})\cdot(\mathds{1}+\vec{r}\,'\cdot\vec{\sigma})]/4^N=(1+\vec{r}\cdot\vec{r}\,')/2^N=(1-r_c)/2^N<0$, however, is negative for this pair of states and therefore $\vec{r}$ {\it cannot} have been a legal quantum state. Thus, it can neither be contained in $\Sigma_N$. 
\end{proof}

\section{The question set}\label{app_B}

In this appendix, we derive the characteristic attributes of the question set $\cq_N$, quoted in section \ref{sec_qn}.

\subsection{Question vectors are $1$-\texttt{bit} states}\label{app_B0}

We begin with a result that helps to characterize $\cq_N$.

\begin{lemma}\label{lem_Qch}
Every $\vec{q}\in\mathbb{R}^{D_N}$ with $y(\vec{q}|\vec{r})\in[0,1]$ $\forall\,\vec{r}\in\Sigma_N$ is a quantum state (in the Bloch vector representation). If, in addition, there exists $\vec{r}_Q\in\Sigma_N$ with $I(\vec{r}_Q)=|\vec{r}_Q|^2=1$ \texttt{bit} such that $y(\vec{q}|\vec{r}_Q)=1$, then $\vec{q}\equiv\vec{r}_Q$.
\end{lemma}

\begin{proof}
Having established the coincidence of $\Sigma_N$ with the set of density matrices over $\mathbb{C}^{2^N}$, we are permitted to work in the hermitian representation of quantum states. Consider $\tilde{\rho}=\f{1}{2^N}(\mathds{1}+\vec{q}\cdot\vec{\sigma})$. It is well-known that $\tilde{\rho}$ is a quantum state if and only if $\tilde{\rho}\geq0$ and $\tr\tilde{\rho}=1$. Since $\tr\tilde{\rho}=1$ by construction, $\tilde{\rho}$ could only fail to be a quantum state if $\tilde{\rho}$ was not positive semi-definite. But then there would exist a quantum state $\rho$ such that $\tr(\rho\tilde{\rho})<0$. This is equivalent to $\vec{q}\cdot\vec{r}<-1$, where $\vec{r}$ is the Bloch vector representation of $\rho$. Since this would be in contradiction with $y(\vec{q}|\vec{r})=1/2(\vec{r}\cdot\vec{q}+1)\in[0,1]$ we conclude that $\vec{q}$ is a quantum state also. 

It is clear that $|\vec{q}|^2> 1$ \texttt{bit} is impossible for otherwise $y(\vec{q}|\vec{r}=\vec{q})>1$. Suppose now that there exists $\vec{r}_Q\in\Sigma_N$ with $|\vec{r}_Q|^2=1$ \texttt{bit} such that $y(\vec{q}|\vec{r}_Q)=1$. This condition can only be fulfilled if $\vec{r}_Q=\vec{q}$ which also implies $|\vec{q}|^2=1$ \texttt{bit}. 
\end{proof}

\subsection{Geometry of the set of Pauli operators}\label{app_paulitime}

We prove two geometric properties of the set of Pauli operators on $\mathbb{C}^{2^N}$, ultimately showing the set to be isomorphic to $\mathbb{CP}^{2^{N}-1}$.

\begin{lemma}\label{lem_pauli1}
$\rm{PSU}(2^N)$ acts transitively on the Pauli operators and these account for all traceless hermitian operators on $\mathbb{C}^{2^N}$ with eigenvalues equal to $\pm1$.
\end{lemma}
\begin{proof}
By definition, any Pauli operator $P$ is hermitian and traceless. Therefore $P$ can be represented as $P=\vec{n}\cdot\vec{\sigma}$ for some $\vec{n}\in\mathbb{R}^{4^N-1}$, since the matrices $\vec{\sigma}$ in (\ref{pauli}) form a basis of hermitian and traceless matrices. Any hermitian matrix is diagonalizable by some matrix $U\in \SU(2^N)$ and thus we can write $P=\vec{n}\cdot\vec{\sigma}=UDU^{\dagger}$ where $D$ is a diagonal matrix with the eigenvalues of $P$ along its diagonal. Since $P$ is a Pauli operator, the diagonal matrix $D$ will contain equal amounts of plus and minus ones along its diagonal. Given any diagonal matrix $D$ of the form above, there {\it always} exists an orthogonal permutation matrix $P_{\sigma}$ which will permute the $\pm 1$'s on the diagonal of $D$ to the $\pm$ configuration found for the matrix $\sigma_{z_1}:=\sigma_z\otimes\mathds{1}\cdots\otimes\mathds{1}$, i.e.\ $D=P_{\sigma}\cdot\sigma_{z_1}\cdot P_{\sigma}^t$. If $P_{\sigma}$ happens to be an odd permutation matrix, we may consider the even permutation $P_{\sigma}\cdot P_{\sigma_0}\in \SU(2^N)$ instead with determinant 1, where $P_{\sigma_0}$ is any 2-cycle permutation which permutes two rows (and the corresponding columns) of $\sigma_{z_1}$ that both contain $+1$ and thus that leaves $\sigma_{z_1}$ invariant. Therefore, without loss of generality we have $D=P_{\sigma}\cdot\sigma_{z_1}\cdot P_{\sigma}^t$ for some $P_{\sigma}\in \SU(2^N)$ and thus $P=\vec{n}\cdot\vec{\sigma}=(UP_{\sigma})\sigma_{z_1}(UP_{\sigma})^{\dagger}$ with $UP_{\sigma}\in \SU(2^N)$. We conclude that all Pauli operators are related by conjugation with unitaries to the diagonal Pauli operator $\sigma_{z_1}$.
\end{proof}

\begin{lemma}\label{lem_pauli2}
The set of Pauli operators is isomorphic to the set of pure quantum states $\mathbb{C}\mathbb{P}^{2^N-1}$.
\end{lemma}
\begin{proof}
We may use the fact that the matrices $\vec{\sigma}$ are exactly the {\it fundamental} generators of $\PSU(2^N)$ and therefore, by an appropriate adjoint transformation $T\in\ct_N$ on the vector $\vec{n}$, we get $(T\cdot\vec{n})\cdot\vec{\sigma}=(UP_{\sigma})^{\dagger}\cdot(\vec{n}\cdot\vec{\sigma})\cdot(UP_{\sigma})=\sigma_{z_1}$ and thus $T\cdot\vec{n}=\vec{\delta}_{z_1}$ because the $\vec{\sigma}$ matrices are linearly independent. We have thus shown that the vector $\vec{n}$ which parametrizes the Pauli operator $P$ is connected via the time evolution group to $\vec{\delta}_{z_1}$ and the set of Pauli operators is therefore isomorphic to\footnote{We denote the orbit of some vector $\vec{q}\in\mathbb{R}^{4^N-1}$ under the time evolution group action as $\ct_N\cdot \vec{q}:=\{T\cdot \vec{q} \ | \ T\in\ct_N\}$ for short.} $\ct_N\cdot\vec{\delta}_{z_1}$. Note now that the unit vector $\vec{\delta}_{z_1}\in\mathbb{R}^{4^N-1}$ is related by an $\SO(4^N-1)$ rotation $O$ to the vector $(\vec{\delta}_{z_1}+\cdots+\vec{\delta}_{z_1\cdots z_N})/\sqrt{2^N-1}$. The group action of $\ct_N$ on $\vec{\delta}_{z_1}$ therefore results in $\ct_N\cdot \vec{\delta}_{z_1}=O^t(O\cdot \ct_N\cdot O^t)O\cdot \vec{\delta}_{z_1}=O^t(O\cdot \ct_N\cdot O^t)(\vec{\delta}_{z_1}+\cdots+\vec{\delta}_{z_1\cdots z_N})/\sqrt{2^N-1}\simeq \mathbb{CP}^{2^{N}-1}$. We used first that equivalent representations lead to isomorphic orbits, secondly that because of transitivity of the pure state space under the action of $\ct_N$, the orbit of the pure state $\vec{r}=(\vec{\delta}_{z_1}+\cdots+\vec{\delta}_{z_1\cdots z_N})$ equals $\mathbb{CP}^{2^{N}-1}$, and lastly that $(O^t/\sqrt{2^N-1})$ is an invertible matrix. We conclude that the set of Pauli operators is isomorphic to $\mathbb{CP}^{2^{N}-1}$ and parametrized by $\ct_N\cdot\vec{\delta}_{z_1}$.
\end{proof}

\subsection{Structure of $\cq_N$}\label{app_qn}

As argued in the main text and in \cite{Hoehn:2014uua}, for $N=1$ qubit the question set $\mathcal{Q}_1$ is isomorphic to the set of pure states $\mathbb{CP}^1$. In the following, we show that (\ref{Qcharacter}) in section \ref{sec_qn} similarly implies that the question set $\cq_N$ is also isomorphic to the set of pure states $\mathbb{CP}^{2^{N}-1}$ for $N>1$ qubits.

\begin{lemma}
Equation (\ref{Qcharacter}) implies that the set of vectors $\vec{n}$ parametrizing the Pauli operators $\vec{n}\cdot\vec{\sigma}$ coincides with the set of all question vectors $\vec{q}$. Therefore, $\cq_N$ is isomorphic to the set of Pauli operators and thereby to the set of pure states such that $\cq_N\simeq\mathbb{C}\mathbb{P}^{2^N-1}$. In particular, $\cq_N$, in its $1$ \texttt{bit} vector representation, inherits a transitive action of the time evolution group $\ct_N=\rm{PSU}(2^N)$ from $\Sigma_N$.
\end{lemma}
\begin{proof}
By equation (\ref{Qcharacter}) the question vectors correspond to legal quantum states, which themselves evolve in the adjoint representation of $\PSU(2^N)$. Therefore, we may form a hermitian operator by contracting the question vector components with the Pauli operators $\vec{q}\cdot\vec{\sigma}:=q_{\mu_1\cdots\mu_N}\sigma_{\mu_1}\otimes\cdots\otimes\sigma_{\mu_N}$ in the same component ordering as for the state vectors. Hence, for every $U\in\SU(2^N)$ we have $U\cdot(\vec{q}\cdot\vec{\sigma})\cdot U^{\dagger}\equiv(T\cdot\vec{q})\cdot\vec{\sigma}$ for $T\in\ct_N$ and thus an action of $\ct_N$ on $\cq_N$ inherited from the states. 

We may equivalently reformulate (\ref{Qcharacter}) in terms of the operator $\vec{q}\cdot\vec{\sigma}$ corresponding to a question $Q\in\cq_N$:
\begin{itemize}
\item[(a)] The condition $|\vec{q}|^2=1$ \texttt{bit} implies 
\ba
\text{tr}[(\vec{q}\cdot\vec{\sigma})^2]=2^N\,|\vec{q}|^{2}=2^N.\nn
\ea
\item[(b)] The requirement of $0\leq y(Q|\vec{r})=(1+\vec{q}\cdot\vec{r})/2\leq1$
$\forall\,\vec{r}\in\Sigma_N$ is equivalent to 
\ba
0\leq\text{tr}[\rho(\mathds{1}+\vec{q}\cdot\vec{\sigma})]/2\leq 1\q\q\Rightarrow\q\q-1\leq \text{tr}[\rho(\vec{q}\cdot\vec{\sigma})]\leq 1,\nn
\ea
for all quantum states $\rho$, where $\rho=(\mathds{1}+\vec{r}\cdot\vec{\sigma})/2^N$ is the density matrix corresponding to $\vec{r}$.
\end{itemize}

All hermitian operators are diagonalizable and thus there must exist a $T\in \ct_N$ which diagonalizes $((T\cdot\vec{q})\cdot\vec{\sigma})=U\cdot(\vec{q}\cdot\vec{\sigma})\cdot U^{\dagger}=D=\text{diag}(D_1,D_2,\ldots,D_{2^N})$, with diagonal elements $D_i$. Note that if $\vec{q}\cdot\vec{\sigma}$ satisfies (a) and (b) above, then so will the operator $D$, since the first constraint is left invariant and the second is related to a valid time evolved state $T^t\cdot\vec{r}$. (a) implies $\text{tr}[D^2]=\sum_{i=1}^{2^N} D_i^2=2^N$. By taking now the diagonal density matrices $\rho_1=\text{diag}(1,0,\ldots,0), \rho_2=\text{diag}(0,1,\ldots,0), \ldots, \rho_{2^N}=\text{diag}(0,\ldots,0,1)$, corresponding to the pure states\footnote{The other Bloch components are fixed by the correlation and complementarity structure.} $(r_{z_1}=1,r_{z_2}=1,\ldots, r_{z_N}=1),(r_{z_1}=-1,r_{z_2}=1,\ldots, r_{z_N}=1),\ldots, (r_{z_1}=-1,r_{z_2}=-1,\ldots, r_{z_N}=-1)$ respectively, (b) becomes $-1\leq \text{tr}[\rho_i\,D]=D_i\leq 1$. These constraints can only be satisfied if $D_i^2=1$ for every index $i$ and therefore $D_i=\pm 1$. Together with $\text{tr}[D]=0$, we have that $D$ is a Pauli operator as follows from the proof of lemma \ref{lem_pauli1}. Since according to lemma \ref{lem_pauli1} $\ct_N$ acts transitively on the set of all Pauli operators, we get directly that $\ct_N$ also acts {\it transitively} on $\cq_N$ and that the set of hermitian operators $\vec{q}\cdot\vec{\sigma}$ corresponding to all questions $Q\in\cq_N$, forms a subset of the Pauli operators. Conversely, every Pauli operator $\vec{n}\cdot\vec{\sigma}$ is of the form $T\cdot\vec{\delta}_{z_1}$ for some $T\in\ct_N$ and satisfies (a) and (b). From (\ref{Qcharacter}) (and lemma \ref{lem_Qch}) it then follows that the vectors $\vec{n}$ which parametrize the Pauli operators correspond to valid questions $Q\in\cq_N$.

Therefore, the set of question vectors coincides with the set of vectors that parametrize the Pauli operators $\vec{n}\cdot\vec{\sigma}$, under the identification $\vec{n}=\vec{q}$. We have shown in lemma \ref{lem_pauli2} that these vectors are all connected to $\vec{\delta}_{z_1}$ by time evolution and they form a set that is isomorphic to $\ct_N\cdot \vec{q}_{z_1}\simeq\mathbb{CP}^{2^{N}-1}$. 
Accordingly, we obtain an explicit isomorphism between the set of Pauli operators and the question set $\cq_N$ by contracting each question vector $\vec{q}$ (corresponding to some $Q\in\cq_N$) with the matrices $\vec{\sigma}$ in (\ref{pauli}). We conclude that $\ct_N$ acts transitively on $\cq_N$ and that $\cq_N$ is isomorphic to the set of Pauli operators. Therefore, $\cq_N$ is also isomorphic to the set of pure states such that $\cq_N\simeq\mathbb{CP}^{2^{N}-1}$.
\end{proof}

\subsection{(Non-)uniqueness of pure state decompositions}\label{app_decomp}

In section \ref{sec_decomp} we quoted the following result:

\begin{lemma}
The decomposition of a pure state vector $\vec{r}_{\text{pure}}=\vec{q}_1+\cdots+\vec{q}_{2^N-1}$ in terms of question vectors $\vec{q}_i$ for $Q_i\in\cq_N$ is \emph{unique} for $N=1,2$ and \emph{non-unique} for $N\geq3$.
\end{lemma}

\begin{proof}
The transitivity of the $\ct_N$ group action on the set of pure states and $\cq_N$ entails that such a decomposition of any pure state is unique if and only if it is unique for the pure state $\vec{(\delta}_{z_1}+\vec{\delta}_{z_2}+\cdots+\vec{\delta}_{z_1\cdots z_N})$. 
The `only if' direction is trivial, so let us assume now that the decomposition of $\vec{\delta}_{z_1}+\cdots\vec{\delta}_{z_1\cdots z_N}$ was unique. There is a $T\in\ct_N$ such that $\vec{\delta}_{z_1}+\cdots\vec{\delta}_{z_1\cdots z_N}=T\cdot \vec{\tilde r}_{\text{pure}}=(T\cdot\vec{\tilde{q}}_1)+\cdots+(T\cdot\vec{\tilde{q}}_{2^N-1})$. Since $(T\cdot\vec{\tilde{q}}_1),\ldots,(T\cdot\vec{\tilde{q}}_{2^N-1})$ are valid question vectors, they must be uniquely equal (up to permutations) to $\vec{\delta}_{z_1},\ldots,\vec{\delta}_{z_1\cdots z_N}$ by assumption. Thus, $\vec{\tilde{q}}_1,\ldots,\vec{\tilde{q}}_{2^N-1}$ are uniquely equal to $T^{-1}\cdot\vec{\delta}_{z_1},\ldots,T^{-1}\cdot\vec{\delta}_{z_1\cdots z_N}$ and the decomposition of $\vec{\tilde r}_{\text{pure}}$ is thus unique. Therefore, without loss of generality we will consider henceforth the pure state $\vec{r}_{\text{pure}}=(\vec{\delta}_{z_1}+\vec{\delta}_{z_2}+\cdots+\vec{\delta}_{z_1\cdots z_N})$. 

Suppose now that there was a second, decomposition of $\vec{r}_{\rm pure}$ in terms of a question set $\vec{q}_i$, $i=1,\ldots 2^N-1$. Since any $\vec{q}_i$ must be answered with `yes' in $\vec{r}_{\rm pure}$, the Born rule (\ref{born1}) implies $y(\vec{q}_i|\vec{r}_{\rm pure})=1$ $\Leftrightarrow$ $\vec{q}_i\cdot(\vec{\delta}_{z_1}+\vec{\delta}_{z_2}+\cdots+\vec{\delta}_{z_1\cdots z_N})=1$, $i=1,\ldots,2^N-1$. The triangle inequalities then imply
\ba
\sum_{j\in\{z_1,\ldots,z_1\cdots z_N\}}\,(\vec{q}_i\cdot\vec{\delta}_j)^2\geq\left(\sum_{j\in\{z_1,\ldots,z_1\cdots z_N\}}\,\vec{q}_i\cdot\vec{\delta}_j\right)^2=1.\nn
\ea
As each question $\vec{q}_i$ must be of length $1$ \texttt{bit} and the $4^N-1$ question vectors $\vec{\delta}_{x_1},\ldots,\vec{\delta}_{z_1\cdots z_N}$ of an informationally complete set  are orthonormal, it also follows that
\ba
1=\sum_{j\in\{x_1,\ldots,z_1\cdots z_N\}}\,(\vec{q}_i\cdot\vec{\delta}_j)^2\geq\sum_{j\in\{z_1,\ldots,z_1\cdots z_N\}}\,(\vec{q}_i\cdot\vec{\delta}_j)^2,\nn
\ea
and therefore 
\ba
\sum_{j\in\{z_1,\ldots,z_1\cdots z_N\}}\,(\vec{q}_i\cdot\vec{\delta}_j)^2=1.\nn
\ea 
Hence, the questions $\vec{q}_i$ lie in the span of $\vec{\delta}_{z_1},\ldots,\vec{\delta}_{z_1\cdots z_N}$, i.e.\ $\vec{q}_i=\sum_{j\in\{z_1,\ldots,z_1\cdots z_N\}}\,(\vec{q}_i\cdot\vec{\delta}_j)\,\vec{\delta}_j$.

Let us now consider the hermitian matrix $\vec{r}_{\text{pure}}\cdot\vec{\sigma}=\sum_i \vec{q}_i\cdot\vec{\sigma}=\sum_{j\in\{z_1,\ldots,z_1\cdots z_N\}}\,\vec{\delta}_j\cdot\vec{\sigma}$. Every individual hermitian matrix $\vec{q}_i\cdot\vec{\sigma}$ appearing in the sum {\it must} be diagonal because $\vec{q}_i$ lies in the span of the questions $\vec{\delta}_{z_1},\ldots,\vec{\delta}_{z_1\cdots z_N}$ which, when contracted with $\vec{\sigma}$, yield the diagonal Pauli operators $\sigma_{z_1}\otimes\mathds{1}\otimes\cdots\otimes\mathds{1},\ldots, \sigma_{z_1}\otimes\cdots\otimes\sigma_{z_N}$. Moreover, $\vec{q}_i\cdot\vec{\sigma}$ must be a Pauli operator since $\vec{q}_i$ is a legal question vector. Therefore $\vec{q}_i\cdot\vec{\sigma}$ is a diagonal matrix with $2^{N-1}$ plus ones and $2^{N-1}$ minus ones along the diagonal and there are exactly $\begin{pmatrix}2^N\\ 2^{N-1}\end{pmatrix}$ of such diagonal Pauli operators. The decomposition of the pure state $\vec{r}_{\text{pure}}$ is now unique if and only if the decomposition of the matrix $\vec{r}_{\text{pure}}\cdot\vec{\sigma}$ in terms of diagonal Pauli operators is unique. 

For $N=1$ this decomposition is trivially unique.

For $N=2$ there are precisely six diagonal Pauli operators; these are exactly the three operators $\sigma_{z_1}\otimes\mathds{1},\mathds{1}\otimes\sigma_{z_2}$ and $\sigma_{z_1}\otimes\sigma_{z_2}$, as well as the operators formed by multiplying them by $-1$. The Pauli operators form a basis of traceless hermitian matrices and therefore the matrices $\vec{q}_i\cdot\vec{\sigma}$ must be exactly the three Pauli operators $\sigma_{z_1}\otimes\mathds{1},\mathds{1}\otimes\sigma_{z_2},\sigma_{z_1}\otimes\sigma_{z_2}$ and we conclude that the decomposition for $N=2$ is also unique.

For $N>2$ the decomposition is, however, no longer unique. Consider for example the simplest case of $N=3$ and $\vec{r}_{\rm pure}=\vec{\delta}_{z_1}+\vec{\delta}_{z_2}+\cdots+\vec{\delta}_{z_1z_2}+\cdots+\vec{\delta}_{z_1z_2z_3}$. Let us conjugate the hermitian matrix $P\cdot(\vec{r}_{\text{pure}}\cdot\vec{\sigma})\cdot P^t=P\cdot\text{diag}(7,-1,\ldots,-1)\cdot P^t=\text{diag}(7,-1,\ldots,-1)$ with a permutation matrix $P$ which permutes two pairs of rows and columns, $2\leftrightarrow 3,4\leftrightarrow 5$, and therefore leaves $\vec{r}_{\text{pure}}\cdot\vec{\sigma}$ invariant. The permutation is even such that $P\in\SU(8)$ and $P$ thus defines an element in the isotropy subgroup associated to $\vec{r}_{\rm pure}$. However, we note that the conjugation with $P$ will not leave the matrices $(\vec{\delta}_{z_1}\cdot\vec{\sigma})=\sigma_{z_1}\otimes\mathds{1}\otimes\mathds{1}=\text{diag}(1,1,1,1,-1,-1,-1,-1),\ldots,(\vec{\delta}_{z_1z_2z_3}\cdot\vec{\sigma})=\sigma_{z_1}\otimes\sigma_{z_2}\otimes\sigma_{z_3}=\text{diag}(1,-1,-1,1,-1,1,1,-1)$ invariant\footnote{The diagonal elements here correspond to choosing the ordering of the diagonal of the density matrix $(\mathds{1}+\vec{r}\cdot\vec{\sigma})/8$ in terms of z 'up' or 'down' of the three qubits as $|+++>, |-++>, |+-+>, |--+>, |++->, |-+->, |+-->, |--->$.}. A simple check shows that the conjugation with $P$ results in seven new Pauli operators $P\cdot(\vec{\delta}_i\cdot\vec{\sigma})\cdot P^t$, $i=z_1,z_2,\ldots,z_1z_2,\ldots,z_1z_2z_3$, which all correspond to legal {\it but different} question vectors than the questions $\vec{\delta}_{z_1},\ldots, \vec{\delta}_{z_1z_2 {z_3}}$ appearing in the original decomposition. $P$ is thus {\it not} contained in the isotropy subgroup associated to $\vec{\delta}_{z_1},\ldots,\vec{\delta}_{z_1z_2z_3}$ and the seven new Pauli operators define a distinct decomposition of the pure state.

One may convince oneself that $P\cdot(\vec{\delta}_{z_1}\cdot\vec{\sigma})\cdot P^t=\text{diag}(1,1,1,-1,1,-1,-1,-1)$, in fact, represents the question $Q=(Q_{z_1} \wedge Q_{z_2}) \vee (Q_{z_1} \wedge Q_{z_3}) \vee (Q_{z_2} \wedge Q_{z_3})$. Similarly, the other Pauli operators $P\cdot(\vec{\delta}_j\cdot\vec{\sigma})\cdot P^t$ will also correspond to legal questions. Note that whenever $Q$ gives 'yes', the probability that the question $Q_{z_1}$ is also answered with `yes' cannot be $1/2$ as $3$ out of $4$ states representing $Q=$ `yes' also feature $Q_{z_1}=$ `yes', and similarly for the questions $Q_{z_2}$ and $Q_{z_3}$. This question $Q$ is therefore {\it not} fully pairwise independent of either of the questions $\vec{\delta}_{z_1},\ldots,\vec{\delta}_{z_1z_2z_3}$. 

 Since $P\in\SU(8)$ we have $P\cdot(\vec{\delta}_j\cdot\vec{\sigma})\cdot P^t=((T\cdot \vec{\delta}_j)\cdot\vec{\sigma})$ for some $T\in\ct_3$. The seven questions $(T\cdot \vec{\delta}_{z_1}),\ldots,(T\cdot \vec{\delta}_{z_1z_2z_3})$ are independent and compatible because so are the $\vec{\delta}_j$. Accordingly, a system of three qubits, in the pure state $\vec{r}_{\text{pure}}$, also answers 'yes' to these 7 questions because of the Born rule. In other words, even though having full information of either of the questions $\vec{\delta}_j$ individually is {\it not} the same as having full information of either of the $T\cdot \vec{\delta}_j$ individually, knowing the answer to {\it all} seven questions $\vec{\delta}_{z_1},\ldots, \vec{\delta}_{z_1\cdots r_{z_3}}$ at the same time is equivalent to knowing the answer to $T\cdot \vec{\delta}_{z_1},\ldots, T\cdot \vec{\delta}_{z_1z_2z_3}$ simultaneously.
 
The same conclusion of non-uniqueness of the pure state decomposition in terms of question vectors holds for all $N\geq3$ because the $2(2^N-1)$ diagonal Pauli operators given by $\sigma_{z_1},\ldots,\sigma_{z_1\cdots z_N}$ and their negatives is a strict subset of the $\begin{pmatrix}2^N\\ 2^{N-1}\end{pmatrix}$ diagonal Pauli operators for $N\geq3$. 
\end{proof}


\providecommand{\href}[2]{#2}\begingroup\raggedright\endgroup


\begin{thebibliography}{10}

\bibitem{Hoehn:2014uua}
P.~A. H\"ohn, ``{Toolbox for reconstructing quantum theory from rules on
  information acquisition},'' {\it Quantum} {\bf1} (2017) 38, 
\href{http://arxiv.org/abs/1412.8323}{{\ttfamily arXiv:1412.8323 [quant-ph]}}.

\bibitem{bell1964einstein}
J.~S. Bell, ``On the einstein-podolsky-rosen paradox,'' {\em Physics}
  {\bfseries 1} no.~3, (1964) 195--200.

\bibitem{bell1966problem}
J.~S. Bell, ``On the problem of hidden variables in quantum mechanics,'' {\em
  Reviews of Modern Physics} {\bfseries 38} no.~3, (1966) 447.

\bibitem{clauser1969proposed}
J.~F. Clauser, M.~A. Horne, A.~Shimony, and R.~A. Holt, ``Proposed experiment
  to test local hidden-variable theories,'' {\em Physical review letters}
  {\bfseries 23} no.~15, (1969) 880.

\bibitem{popescu1994quantum}
S.~Popescu and D.~R\"ohrlich, ``Quantum nonlocality as an axiom,'' {\em
  Foundations of Physics} {\bfseries 24} no.~3, (1994) 379--385.

\bibitem{pawlowski2009information}
M.~Paw{\l}owski, T.~Paterek, D.~Kaszlikowski, V.~Scarani, A.~Winter, and
  M.~{\.Z}ukowski, ``Information causality as a physical principle,'' {\em
  Nature} {\bfseries 461} no.~7267, (2009) 1101--1104.
  
\bibitem{Fivel:1994nj}
D.~I.~Fivel, ``{How interference effects in mixtures determine the rules of quantum mechanics},''
   {\em Phys. Rev. A} {\bfseries 50} (1994) 2108.

\bibitem{Sorkin:1994dt}
D.~Sorkin, Rafael, ``{Quantum mechanics as quantum measure theory},''
  \href{http://dx.doi.org/10.1142/S021773239400294X}{{\em Mod.Phys.Lett.}
  {\bfseries A9} (1994) 3119--3128},
\href{http://arxiv.org/abs/gr-qc/9401003}{{\ttfamily arXiv:gr-qc/9401003
  [gr-qc]}}.

\bibitem{Barnum:2014fk}
H.~Barnum, M.~P. M\"uller, and C.~Ududec, ``Higher-order interference and
  single-system postulates characterizing quantum theory,''
  \href{http://arxiv.org/abs/1403.4147}{{\ttfamily 1403.4147}}.
  \url{http://arxiv.org/abs/1403.4147}.

\bibitem{Rovelli:1995fv}
C.~Rovelli, ``{Relational quantum mechanics},''
  \href{http://dx.doi.org/10.1007/BF02302261}{{\em Int.J.Theor.Phys.}
  {\bfseries 35} (1996) 1637--1678},
\href{http://arxiv.org/abs/quant-ph/9609002}{{\ttfamily arXiv:quant-ph/9609002
  [quant-ph]}}.

\bibitem{zeilinger1999foundational}
A.~Zeilinger, ``A foundational principle for quantum mechanics,'' {\em
  Foundations of Physics} {\bfseries 29} no.~4, (1999) 631--643.

\bibitem{Brukner:ys}
C.~Brukner and A.~Zeilinger, ``Information and fundamental elements of the
  structure of quantum theory,'' {\em in "Time, Quantum, Information", edited
  by L.. Castell and O. Ischebeck (Springer, 2003)} ,
  \href{http://arxiv.org/abs/quant-ph/0212084}{{\ttfamily quant-ph/0212084}}.
  \url{http://arxiv.org/abs/quant-ph/0212084}.

\bibitem{Brukner:1999qf}
C.~Brukner and A.~Zeilinger, ``Operationally invariant information in quantum
  measurements,'' {\em Phys. Rev. Lett.} {\bfseries 83} (1999) 3354--3357,
  \href{http://arxiv.org/abs/quant-ph/0005084}{{\ttfamily quant-ph/0005084}}.
  \url{http://arxiv.org/abs/quant-ph/0005084}.

\bibitem{Brukner:2002kx}
C.~Brukner and A.~Zeilinger, ``Young's experiment and the finiteness of
  information,'' {\em Phil. Trans. R. Soc. Lond. A} {\bfseries 360} (2002)
  1061, \href{http://arxiv.org/abs/quant-ph/0201026}{{\ttfamily
  quant-ph/0201026}}. \url{http://arxiv.org/abs/quant-ph/0201026}.

\bibitem{brukner2009information}
{\v{C}}.~Brukner and A.~Zeilinger, ``Information invariance and quantum
  probabilities,'' {\em Foundations of Physics} {\bfseries 39} no.~7, (2009)
  677--689.

\bibitem{Brukner:vn}
C.~Brukner, M.~Zukowski, and A.~Zeilinger, ``The essence of entanglement,''
  \href{http://arxiv.org/abs/quant-ph/0106119}{{\ttfamily quant-ph/0106119}}.
  \url{http://arxiv.org/abs/quant-ph/0106119}.

\bibitem{spekkens2007evidence}
R.~W. Spekkens, ``Evidence for the epistemic view of quantum states: A toy
  theory,'' {\em Physical Review A} {\bfseries 75} no.~3, (2007) 032110.

\bibitem{Spekkens:2014fk}
R.~W. Spekkens, ``Quasi-quantization: classical statistical theories with an
  epistemic restriction,'' \href{http://arxiv.org/abs/1409.5041}{{\ttfamily
  1409.5041}}. \url{http://arxiv.org/abs/1409.5041}.

\bibitem{Paterek:2010fk}
T.~Paterek, B.~Dakic, and C.~Brukner, ``Theories of systems with limited
  information content,'' {\em New J. Phys.} {\bfseries 12} (2010) 053037,
  \href{http://arxiv.org/abs/0804.1423}{{\ttfamily 0804.1423}}.
  \url{http://arxiv.org/abs/0804.1423}.

\bibitem{weizsaeckerbook}
C.~F. von Weizs\"acker, {\em The Structure of Physics}.
\newblock Springer-Verlag, Dordrecht, 2006.

\bibitem{gornitz2003introduction}
T.~G{\"o}rnitz and O.~Ischebeck, {\em An Introduction to Carl Friedrich von
  Weizs\"acker's Program for a Reconstruction of Quantum Theory}.
\newblock Time, Quantum and Information. Springer, 2003.

\bibitem{lyre1995quantum}
H.~Lyre, ``Quantum theory of ur-objects as a theory of information,'' {\em
  International Journal of Theoretical Physics} {\bfseries 34} no.~8, (1995)
  1541--1552.

\bibitem{Hardy:2001jk}
L.~Hardy, ``{Quantum theory from five reasonable axioms},''
\href{http://arxiv.org/abs/quant-ph/0101012}{{\ttfamily arXiv:quant-ph/0101012
  [quant-ph]}}.

\bibitem{Dakic:2009bh}
B.~Dakic and C.~Brukner, ``Quantum theory and beyond: Is entanglement
  special?,'' {\em Deep Beauty: Understanding the Quantum World through
  Mathematical Innovation, Ed. H. Halvorson (Cambridge University Press, 2011)
  365-392} (11, 2009) , \href{http://arxiv.org/abs/0911.0695}{{\ttfamily
  0911.0695}}. \url{http://arxiv.org/abs/0911.0695}.

\bibitem{masanes2011derivation}
L.~Masanes and M.~P. M{\"u}ller, ``A derivation of quantum theory from physical
  requirements,'' {\em New Journal of Physics} {\bfseries 13} no.~6, (2011)
  063001.

\bibitem{Mueller:2012ai}
M.~P. M\"uller and L.~Masanes, ``{Information-theoretic postulates for quantum
  theory},''
\href{http://arxiv.org/abs/1203.4516}{{\ttfamily arXiv:1203.4516 [quant-ph]}}.

\bibitem{Masanes:2012uq}
L.~Masanes, M.~P. M\"uller, R.~Augusiak, and D.~Perez-Garcia, ``Existence of an
  information unit as a postulate of quantum theory,'' {\em PNAS vol 110 no 41
  page 16373 (2013)} (08, 2012) ,
  \href{http://arxiv.org/abs/1208.0493}{{\ttfamily 1208.0493}}.
  \url{http://arxiv.org/abs/1208.0493}.

\bibitem{chiribella2011informational}
G.~Chiribella, G.~M. D'Ariano, and P.~Perinotti, ``Informational derivation of
  quantum theory,'' {\em Physical Review A} {\bfseries 84} no.~1, (2011)
  012311.

\bibitem{de2012deriving}
G.~de~la Torre, L.~Masanes, A.~J. Short, and M.~P. M\"uller, ``Deriving quantum
  theory from its local structure and reversibility,'' {\em Physical Review
  Letters} {\bfseries 109} no.~9, (2012) 090403.

\bibitem{Mueller:2012pc}
M.~P. M\"uller and L.~Masanes, ``{Three-dimensionality of space and the quantum
  bit: how to derive both from information-theoretic postulates},''
  \href{http://dx.doi.org/10.1088/1367-2630/15/5/053040}{{\em New J. Phys. 15,}
  {\bfseries 053040} (2013) },
\href{http://arxiv.org/abs/1206.0630}{{\ttfamily arXiv:1206.0630 [quant-ph]}}.

\bibitem{Hardy:2013fk}
L.~Hardy, ``Reconstructing quantum theory,''
  \href{http://arxiv.org/abs/1303.1538}{{\ttfamily 1303.1538}}.
  \url{http://arxiv.org/abs/1303.1538}.

\bibitem{Garner:2014uq}
A.~J.~P. Garner, M.~P. M{\"u}ller, and O.~C.~O. Dahlsten, ``The quantum bit
  from relativity of simultaneity on an interferometer,''
  \href{http://arxiv.org/abs/1412.7112}{{\ttfamily 1412.7112}}.
  \url{http://arxiv.org/abs/1412.7112}.

\bibitem{Oeckl:2014uq}
R.~Oeckl, ``A first-principles approach to physics based on locality and
  operationalism,'' \href{http://arxiv.org/abs/1412.7731}{{\ttfamily
  1412.7731}}. \url{http://arxiv.org/abs/1412.7731}.

\bibitem{kochen2013reconstruction}
S.~Kochen, ``A reconstruction of quantum mechanics,'' {\em arXiv preprint
  arXiv:1306.3951} (2013) .
  
  \bibitem{2008arXiv0805.2770G}
P.~Goyal, ``From Information Geometry to Quantum Theory,'' {\em New J. Phys.} {\bfseries 12} (2010) 023012,
  \href{http://arxiv.org/abs/0805.2770}{{\ttfamily 0805.2770}}.
  \url{http://arxiv.org/abs/0805.2770}.

\bibitem{Smerlak:2006gi}
M.~Smerlak and C.~Rovelli, ``{Relational EPR},''
  \href{http://dx.doi.org/10.1007/s10701-007-9105-0}{{\em Found.Phys.}
  {\bfseries 37} (2007) 427--445},
\href{http://arxiv.org/abs/quant-ph/0604064}{{\ttfamily arXiv:quant-ph/0604064
  [quant-ph]}}.

\bibitem{hartle1968quantum}
J.~B. Hartle, ``Quantum mechanics of individual systems,'' {\em American
  Journal of Physics} {\bfseries 36} no.~8, (1968) 704--712.

\bibitem{zheng1996quantum}
Q.~Zheng and T.~Kobayashi, ``Quantum optics as a relativistic theory of
  light,'' {\em Physics Essays} {\bfseries 9} (1996) 447--459.

\bibitem{Fuchs:fk}
C.~A. Fuchs, ``Quantum mechanics as quantum information (and only a little
  more),'' \href{http://arxiv.org/abs/quant-ph/0205039}{{\ttfamily
  quant-ph/0205039}}. \url{http://arxiv.org/abs/quant-ph/0205039}.

\bibitem{Caves:2002uq}
C.~M. Caves, C.~A. Fuchs, and R.~Schack, ``Quantum probabilities as bayesian
  probabilities,'' {\em Phys. Rev. A} {\bfseries 65} (2002) 022305,
  \href{http://arxiv.org/abs/quant-ph/0106133}{{\ttfamily quant-ph/0106133}}.
  \url{http://arxiv.org/abs/quant-ph/0106133}.

\bibitem{caves2002unknown}
C.~M. Caves, C.~A. Fuchs, and R.~Schack, ``Unknown quantum states: the quantum
  de finetti representation,'' {\em Journal of Mathematical Physics} {\bfseries
  43} no.~9, (2002) 4537--4559.

\bibitem{petersen}
A.~Petersen, ``The philosophy of niels bohr,'' {\em Bulletin of the Atomic
  Scientists} {\bfseries 19} (1963) 8.

\bibitem{Krein1940}
M.~Krein and D.~Milman, ``On extreme points of regular convex sets,'' {\em
  Studia Mathematica} {\bfseries 9} no.~1, (1940) 133--138.
  \url{http://eudml.org/doc/219061}.






\bibitem{hw2}
P.~A. H\"ohn and C.~S.~P. Wever, ``A reconstruction of rebit quantum theory
  from rules on information acquisition,'' {\em in preparation} .

\bibitem{barrett2007information}
J.~Barrett, ``Information processing in generalized probabilistic theories,''
  {\em Physical Review A} {\bfseries 75} no.~3, (2007) 032304.

\bibitem{Masanes:2011kx}
L.~Masanes, M.~P. M\"uller, D.~Perez-Garcia, and R.~Augusiak, ``Entanglement
  and the three-dimensionality of the bloch ball,'' {\em Journal of
  Mathematical Physics 55-122203 (2014)} (11, 2011) ,
  \href{http://arxiv.org/abs/1111.4060}{{\ttfamily 1111.4060}}.
  \url{http://arxiv.org/abs/1111.4060}.

\bibitem{Brukner:ly}
C.~Brukner and A.~Zeilinger, ``Quantum measurement and shannon information, a
  reply to m. j. w. hall,''
  \href{http://arxiv.org/abs/quant-ph/0008091}{{\ttfamily quant-ph/0008091}}.
  \url{http://arxiv.org/abs/quant-ph/0008091}.

\bibitem{Brukner:2001ve}
C.~Brukner and A.~Zeilinger, ``Conceptual inadequacy of the shannon information
  in quantum measurements,'' {\em Phys.Rev.A} {\bfseries 63} (2001) 022113,
  \href{http://arxiv.org/abs/quant-ph/0006087}{{\ttfamily quant-ph/0006087}}.
  \url{http://arxiv.org/abs/quant-ph/0006087}.

\bibitem{Bengtsson}
I.~Bengtsson and K.~Zyczkowski, {\em Geometry of quantum states: an
  introduction to quantum entanglement}.
\newblock Cambridge University Press, 2006.

\bibitem{bremner2002practical}
M.~J. Bremner, C.~M. Dawson, J.~L. Dodd, A.~Gilchrist, A.~W. Harrow,
  D.~Mortimer, M.~A. Nielsen, and T.~J. Osborne, ``Practical scheme for quantum
  computation with any two-qubit entangling gate,'' {\em Physical review
  letters} {\bfseries 89} no.~24, (2002) 247902.

\bibitem{Harrow:2008aa}
A.~Harrow, ``Exact universality from any entangling gate without inverses,''
  {\em Q. Inf. Comp.} {\bfseries 9} (2009) 773.

\bibitem{nielsen2010quantum}
M.~A. Nielsen and I.~L. Chuang, {\em Quantum computation and quantum
  information}.
\newblock Cambridge university press, 2010.

\bibitem{lawrence2002mutually}
J.~Lawrence, {\v{C}}.~Brukner, and A.~Zeilinger, ``Mutually unbiased binary
  observable sets on n qubits,'' {\em Physical Review A} {\bfseries 65} no.~3,
  (2002) 032320.

\bibitem{Hoehn:2014vua}
 P.~A. H\"ohn and M.~P. M\"uller,
  ``An operational approach to spacetime symmetries: Lorentz transformations from quantum communication,'' {\em New J.\ Phys.} {\bf18} (2016) 063026;
  arXiv:1412.8462 [quant-ph].


\bibitem{Raffenetti:1969}
R.~C. Raffenetti and K.~Ruedenberg, ``Parametrization of an orthogonal matrix
  in terms of generalized eulerian angles,''
  \href{http://dx.doi.org/10.1002/qua.560040725}{{\em International Journal of
  Quantum Chemistry} {\bfseries 4} no.~S3B, (1969) 625--634}.
  \url{http://dx.doi.org/10.1002/qua.560040725}.

\bibitem{Raffenetti:1972}
D.~K. Hoffman, R.~C. Raffenetti, and K.~Ruedenberg, ``Generalization of euler
  angles to n-dimensional orthogonal matrices,''
  \href{http://dx.doi.org/http://dx.doi.org/10.1063/1.1666011}{{\em Journal of
  Mathematical Physics} {\bfseries 13} no.~4, (1972) 528--533}.
  \url{http://scitation.aip.org/content/aip/journal/jmp/13/4/10.1063/1.1666011%
}.

\bibitem{Georgi:1982jb}
H.~Georgi, ``{Lie Algebras in Particle Physics. From Isospin to Unified
  Theories},''
{\em Front.Phys.} {\bfseries 54} (1982) 1--255.

\bibitem{Feger:2012aa}
R.~Feger and T.~Kephart, ``LieArt -- a mathematica application for Lie algebras
  and representation theory,'' \href{http://arxiv.org/abs/1206.6379}{{\ttfamily
  arXiv:1206.6379}}.

\bibitem{Dynkin:1957aa}
E.~Dynkin, ``Maximal subgroups of the classical groups,'' {\em AMS Transl.}
  {\bfseries 6} (1957) 245--378.

\end{thebibliography}
\end{document}